\documentclass[12pt]{article}

\usepackage{eurosym}
\usepackage[left=2.2cm, right=2.2cm, top=2.2cm, bottom=2.2cm]{geometry}
\usepackage{amsmath,amsfonts,amssymb,amsthm}
\usepackage{lscape}
\usepackage{graphicx}
\usepackage{natbib} 
\usepackage{amsmath,amssymb,graphicx,setspace,rotating}
\usepackage[outdir=./]{epstopdf}
\usepackage[outdir=./]{epsfig} 
\usepackage{pdflscape}
\usepackage[section]{placeins}    
\usepackage{color}  
\usepackage{threeparttable} 
\usepackage{bm}      
                        
\DeclareMathOperator*{\plim}{plim}    
\newtheorem{theorem}{Theorem}

\newtheorem{proposition}{Proposition} 
\newtheorem{assumption}{Assumption}  
\newtheorem{lemma}{Lemma}

\newcommand{\pto}{\overset{p}{\to}} 
\newcommand{\dto}{\overset{d}{\to}} 
\newtheorem{remark}{Remark}

\def\1{1\!{\rm l}}

\def \R {\mathbb{R}}
\def \E {\mathbb{E}}

\newcommand\floor[1]{\lfloor#1\rfloor}

\begin{document}

\title{Indirect Inference with a Non-Smooth Criterion Function\thanks{
    We would like to thank the Editor, Jianqing Fan, an Associate Editor and three anonymous referees for their constructive comments that greatly improved the paper.
    We also thank
    Jean-Jacques Forneron,    
    Pedro Sant'Anna,
    as well as 
    participants at the 2018 Shandong Econometrics Conference (Shandong University)
    the Workshop on Advances in Econometrics 2018 (Dogo Onsen)
    and
    BU 2019 Pi-day conference (Boston University). 
This is the revised version of a paper previously circulated under the title     
``Derivative-Based Optimisation with a Non-Smooth Simulated Criterion.'' 
First version: August 8, 2017.
  }
  \vspace{1.0cm}
} 

\date{\today}

\author{
  David T. Frazier\footnote{ 
  Department of Econometrics and Business Statistics, Monash University    
  (\texttt{david.frazier@monash.edu}).}
  \hspace{0.7cm}
  Tatsushi Oka\footnote{ 
  Department of Econometrics and Business Statistics, Monash University    
  (\texttt{tatsushi.oka@monash.edu}).}
  \hspace{0.7cm}
  Dan Zhu\footnote{Department of Econometrics and Business Statistics, Monash University (\texttt{dan.zhu@monash.edu}).}
  \vspace{0.7cm}
}

\maketitle

\begin{abstract}
Indirect inference requires simulating realisations of endogenous variables from the model under study. When the endogenous variables are discontinuous functions of the model parameters, the resulting indirect inference criterion function is discontinuous and does not permit the use of derivative-based optimisation routines. Using a change of variables technique, we propose a novel simulation algorithm that alleviates the discontinuities inherent in such indirect inference criterion functions, and permits the application of derivative-based optimisation routines to estimate the unknown model parameters. Unlike competing approaches, this approach does not rely on kernel smoothing or bandwidth parameters. Several Monte Carlo examples that have featured in the literature on indirect inference with discontinuous outcomes illustrate the approach, and demonstrate the superior performance of this approach over existing alternatives.
\end{abstract}

\vspace{1cm}

\vspace{0.5cm}

\newpage

\section{Introduction}

Simulation-based estimation methods, such as the method of simulated moments (\citealp{mcfadden1989method}, \citealp{duffie1993simulated}) and indirect inference (\citealp{Smith1993}, \citealp{GMR1993}, \citealp{GT1996}), are widely used inference procedures that are applicable to any model where simulation of data is possible. These methods are particularly useful in settings where the underlying structural model is too difficult for maximum likelihood estimation,  but where simulation from the model is straightforward. 

Given a fixed value of the unknown model parameters, simulation-based methods require the user to simulate synthetic realisations of endogenous variables, often referred to as simulated outcomes, from the underlying structural model. Once these simulated realisations have been generated, statistics based on the simulated data are calculated and then compared against statistics based on the observed data. Estimators for the unknown model parameters are then obtained by minimising a well-defined distance between the simulated summary statistics and their observed counterparts.

However, in many interesting cases,  the simulated outcomes from the structural model of interest are discontinuous transformations of the underlying model parameters, i.e., small changes in the parameter values can lead to substantial changes in the simulated data. Hence, the resulting sample criterion functions used in estimation will be discontinuous functions of the model parameters, even though the corresponding limit of the sample criterion function will often be differentiable in the model parameters. Indeed, the simulation of discontinuous outcomes, and the resulting discontinuity of the sample criterion function, is a relatively common occurrence in the indirect inference (hereafter, II) literature. Notable examples where II has been applied in the context of discontinuous simulated outcomes include the following: dynamic labor market models that are subject to the so-called ``initial-conditions'' problem (\citealp{AL2000}); switching-type models, such as autoregressive models with exponential marginal distributions (\citealp{IC2006}); structural models of human capital accumulation with learning (\citealp{Nagypal2006});  first-price auction models with heterogeneous bidders (\citealp{li2015}); certain dynamic sample selection models (\citealp{altonji2013modeling}); and the application of II to dynamic discrete choice models (see, e.g., \citealp{bruins2015generalized} for a discussion).

{Two potential solutions that can partially circumvent the difficulties encountered in II estimation with a discontinuous criterion function are the use of numerical derivatives within an optimisation scheme, and the use of derivative-free optimisation approaches. A simple solution to this issue would be to apply finite-differencing derivative estimates within a Newton-Raphson or quasi-Newton algorithm, even though the criterion function may be discontinuous in finite samples. The logic behind such an approach is often based on the following notion: if we could construct a criterion function using an infinite number of simulated samples, the resulting criterion function would be smooth enough to permit the use of numerical derivatives. While producing an infinite number of simulation is infeasible, if one takes the number of simulations used in II to be very large, which significantly increases the required computational effort, this would effectively smooth the discontinuous criterion function and give a basis for the use of such numerical derivatives within parameter estimation. \cite{gottard2017estimating} provide simulation results supporting this approach, when the number of simulations is larger than the sample size. However, computational issues aside, even with a large simulation size, estimating these derivatives using finite-differencing requires specifying a tuning parameter that imparts bias on the resulting estimates. In practice, the tuning parameter induces a trade-off between bias and variance and may have a substantial effect
on the estimators in finite samples, especially when the underlying model is discontinuous (see, e.g., \citealp{glynn1989optimization, andrieu2011gradient,detemple2005asymptotic}).}




{An alternative approach to finding II estimators in this setting is to use derivative-free methods, such as simplex-based algorithms or genetic algorithms. These methods can be quite useful when the dimension of the model parameters is relatively small, however, such methods often encounter difficulties when the dimension of the parameters is large (see \citealp{bruins2015generalized}, for a discussion of this issue). }

Recently, building on the initial work of \cite{keane2003generalized} and \cite{IC2006}, \cite{bruins2015generalized} have proposed a generalized II (GII) approach to alleviate the issue of discontinuous, as functions of the model parameters, simulated outcomes. The GII approach replaces the discontinuous simulated outcomes by a kernel-smoothed version that depends on a bandwidth parameter. For a positive value of the bandwidth parameter, GII allows the use of derivative-based optimisation routines to estimate the unknown model parameters. Furthermore, under certain regularity conditions, including that the bandwidth parameter shrinks to zero fast enough, the GII approach produces consistent and asymptotically normal estimates of the model parameters.
However, the application of GII can encounter certain difficulties in practice: the GII approach relies on a somewhat arbitrary choice of kernel function and bandwidth parameter; for any fixed sample size,  the use of artificially smoothed simulated outcomes imparts a non-negligible bias on the resulting parameter estimates; as is generally true with kernel smoothing methods, the choice of the bandwidth parameter is crucial for obtaining reliable performance.

The goal of this paper is to propose a novel simulation algorithm that yields a differentiable sample II criterion function in situations where data simulated under the structural model is discontinuous. Unlike the aforementioned GII approach, this new approach does not rely on any smoothing approaches, nor does it require the user to select a bandwidth parameter before the procedure can be implemented.

The key to this new II approach is a local change of variables (COV) technique that draws inspiration from the Hessian optimal partial proxy (HOPP) method of \cite{joshi2016optimal}. The HOPP method is a COV technique that allows the construction of unbiased estimators, up to a third-order term, for the derivatives of certain expectations that are of keen interest in financial mathematics, such as the so-called ``Greeks'' that are associated with option pricing. Additional COV strategies for calculating derivatives of similar expectations can be found in \cite{fu1994optimization}, \cite{Lyuu}, \cite{chan2009minimal} and \cite{peng2018new}, with these ideas first introduced by \cite{glynn1987likelilood} in the study of discrete-event systems.\footnote{We refer the interested reader to \cite{fu2006gradient} for an overview of these methods.} 

Unlike the problems to which the HOPP method is applied, which focuses on estimating derivatives of an expectation at a point, in II we are interested in obtaining uniformly, over the parameter space, consistent estimates for the derivatives of a simulated sample criterion function. {We propose a modification  of the HOPP approach  that can be applied to II estimation and demonstrate that this new approach results in II criterion functions that are continuously differentiable in the model parameters.} As a result, this new procedure permits the use of derivative-based optimisation routines to estimate the unknown model parameters, even though the original simulated outcomes are discontinuous. Critically, derivatives calculated from II criterion functions that use this technique are uniformly, over the parameter space, consistent estimators of their corresponding limit counterparts.

The approach considered herein amounts to a direct approximation of the II criterion function in a neighborhood of the point where the original criterion function is discontinuous. As such, our approach is a form of ``generalized indirect inference''. However, unlike the GII approach of \cite{bruins2015generalized}, which relies on a \textit{global} kernel smoothing approximation, the GII approach proposed herein relies on a \textit{local} approximation. To differentiate these two GII approaches, hereafter we refer to our approach as change of variable generalized indirect inference (GII-COV), while the kernel-based approach of \cite{bruins2015generalized} is referred to as kernel generalized indirect inference (GII-K). 

We demonstrate that our GII-COV approach yields consistent estimators for the derivatives of the simulated moments used within II estimation. As a result, GII-COV allows the consistent application of derivative-based optimisation routines to produce computationally efficient II parameter estimates, even though the model under study produces discontinuous simulated outcomes. A direct result of the GII-COV approach is a criterion function that is twice-continuously differentiable, uniformly in the parameters, which ensures that estimators obtained from this approach will have standard asymptotic properties, under fairly weak regularity conditions.

While numerical differentiation is the most common tool for calculating derivatives in econometrics, the computational tool we use for derivative calculation in this paper is automatic differentiation. This technique is common in computer science and financial mathematics (\citealp{glasserman2003monte}) and generally leads to faster derivative calculations than finite-differencing techniques, especially in a high-dimensional settings. Automatic differentiation {is a numerical procedure for estimating derivatives and can be viewed as a type of optimal finite-differencing derivative estimator; in particular, numerical derivatives calculated via automatic differentiation do not exhibit the bias associated with finite-differencing derivative estimates}. Instead, {derivatives calculated via automatic differentiation produce} exact numerical derivatives of the function under consideration, up to floating point errors. Therefore, the use of automatic differentiation will not only speed up the execution of a derivative-based optimisation algorithm but also produces results that are free from the bias inherent in finite-differencing derivative estimates.

The remainder of the paper is organised as follows. Section 2 supplies the general set-up and notation for the structural model and briefly reviews II estimation procedures. Section 3 proposes a change of variables technique that we use within our GII approach, and demonstrates that this approach to II permits the consistent application of derivative-based optimisation routines to estimate the unknown model parameters. Illustrative examples showcase the precise implementation details regarding this approach to II. Section 4 discusses the asymptotic properties of this approach. In Section 5, we apply our GII approach to several dynamic discrete choice models that have featured in the literature on II with discontinuous outcomes, and compare the resulting parameter estimates against the GII approach of \cite{bruins2015generalized} and two popular derivative-free methods. The results demonstrate that our approach compares favourably to existing approaches. Section 6 concludes. All proofs and tables are relegated to the appendix.

\section{Model, Examples and Standard Indirect Inference}

In this section we first present the model setup and describe
the standard indirect inference (hereafter, II) approach. In addition, we briefly examine the application of II in several economic examples.

In the remainder of the paper, we use the following notation.
Consider a $p\times1$ vector $x=(x_1,\cdots,x_p)'$
and a $q \times p$ matrix $A$.
We use $\|x\|$ to denote the Euclidean norm,
use $\|A\|$ to denote the operator norm (i.e.,
$\|A\| = \sup_{z\in \R^{p}: \|z\|=1}\| Az \|/\|z\|$
and define $\|x\|^{2}_{W}:=x'Wx$
for a $p \times p$ matrix $W$.
Let $f(x)=(f_{1}(x),...,f_{q}(x))'$
be
a $q\times1$ vector function consisting of differentiable scalar functions.
For $j=1,\dots, q$,
we denote by $\partial_{x_{i}}f_{j}(x)$ the derivative of $f_{j}(x)$
with respect to the $i$-th  component of $x$
for $i = 1, \dots, n$,
and the gradient  of $f_{j}(x)$ with respect to $x$
is denoted by
$\partial_{x} f_{j}(x)=
\big (\partial_{x_{1}}f_{j}(x),\dots,\partial_{x_{p}}f_{j}(x)\big)'$.
The gradient of the vector function $f(x)$ is given by
the $q \times p$ matrix
$\partial_{x} f(x) = \big (\partial_{x} f_{1}(x), \dots, \partial_{x} f_{q}(x)
\big)'$.
For $\delta>0$, define the $\delta$-neighborhood of the point $x^*\in\mathbb{R}^{p}$ as $\mathcal{N}_{\delta}(x^*):=\{x\in\mathbb{R}^{p}:\|x-x^*\|\leq\delta\}$. Also, let $\1[S]$ denote the indicator function on the set $S$.

\subsection{Models and Examples}\label{sec:me}

Assume the researcher wishes to conduct inference on unknown parameters $\theta\in\Theta\subseteq \mathbb{R}^{d_{\theta}}$, where $\Theta\subseteq\mathbb{R}^{d_{\theta}} $ denotes the parameter space of $\theta$, with $d_\theta$ its size, that govern the behavior of an endogenous variable $y$, whose support is $\mathcal{Y}$. Conditional on an exogenous variable $x$, with support $\mathcal{X}$, and an unobservable state variable $s$, with support $\mathcal{S}$, the endogenous variable $y$ evolves according to the following (causal) structural model:
\begin{flalign}
\label{eq:outcome}
y_{}&=
g (s_{}; \theta),\\ 
s_{}&= h (x_{},\epsilon_{};\theta^{}),\label{eq:state}
\end{flalign}where $\epsilon_{}$ is an error term that is independent and identically distributed (iid) according to the {known} cumulative distribution functions $F_{\epsilon}(\cdot)$, with corresponding density function $f_{\epsilon}(\cdot)$, and whose support is $\mathcal{E}$. The exogenous variables $x_{}$ are independent of $\epsilon_{}$, and the functions $g: \mathcal{S} \to \mathcal{Y}$ and $h:\mathcal{X} \times \mathcal{E} \to \mathcal{S}$ are known up to the unknown parameters $\theta^{}$.

We are interested in cases where the function $g(\cdot)$ is discontinuous in the state variable $s_{}$, and where, due to the complexity of the structural model, likelihood-based inference for $\theta$ is infeasible or prohibitively difficult. This model setup in \eqref{eq:outcome}-\eqref{eq:state} is fairly common in econometrics, and covers a wide variety of models. For illustration purposes, we provide four classes of examples considered under our framework.

\paragraph{Example 1 (Binary Choice Models with Serially Dependent Errors).}
Suppose that we observe a panel of realisations 
$\{(x_{it}', y_{it})' \in \mathbb{R}^{d_{x}}{\times}\mathbb{R}, i=1, \dots, n, t=1, \dots, T\}$
with cross-sectional unit $i $
and time period $t$,
generated from a binary choice model with autoregressive (AR) errors:
\begin{eqnarray*}
  y_{it}=\1[x_{it}'\gamma+v_{it} >0]
  \ \ \ \mathrm{with} \ \
  v_{it}=\rho v_{i,t-1}+\epsilon_{it},
\end{eqnarray*}
where the variable $x_{it}$ is a $d_{x} \times 1$ vector of exogenous variables with support $\mathcal{X} \subseteq \R^{d_{x}}$, the unobserved variable $v_{it}$ is generated by an AR(1) model
with unobserved iid innovation $\epsilon_{it}$, which follows the known distribution $F_{\epsilon}(\cdot)$. Here,
the state variable is $s_{it}=x_{it}'\gamma+\rho v_{i,t-1}+\epsilon_{it}$
and the structural parameters are $\theta = (\gamma', \rho)'$.
\hfill\(\Box\)

\paragraph{Example 2 (Ordered Probit Model with Individual Effects).}
Let $y_{it}$ be a categorical variable taking values in $\{0, 1, \dots, J\}$
for individual $i=1, \dots, n$ at time $t = 1, \dots, T$.
Given an observed, non-constant, vector $x_{it}$,
we assume that $y_{it} = j$ according to the model
\begin{eqnarray*}
  y_{it}=
  \left \{
  \begin{array}{ll}
    0, & \mathrm{if} \ \ s_{it} \le \delta_{1}\\
    1, & \mathrm{if} \ \ \delta_{1} < s_{it} \le \delta_{2} \\
     &  \vdots\\
    J, & \mathrm{if} \ \ \delta_{J} < s_{it}.
  \end{array}
  \right.
\end{eqnarray*}
Here,
$\delta_{1}, \dots, \delta_{J}$ are unknown threshold parameters
and $s_{it}$ is a state variable, which can be interpreted as the individuals latent utility, given by 
\begin{eqnarray*}
  s_{it}=x_{it}'\gamma + \sigma v_{i}+w_{it},
\end{eqnarray*}
where 
$v_{i}$ is an iid, time-invariant and individual-specific unobservable random variable following
$N(0,1)$, while
$w_{it}$ is an unobservable iid innovation following $N(0,1)$, and we define $\epsilon_{it}=(v_i,w_{it})'$. The structural parameters are $\theta=(\delta_{1},\dots,\delta_{J},\gamma,\sigma)'$.
\hfill\(\Box\)

\paragraph{Example 3 (Switching-type Models).}
{Let $v_{t}$ be iid exponentially distributed  with unity intensity parameter and let $u_{t}$ be iid uniform $[0,1]$, with $v_t$ and $u_t$ also independent.} The exponential autoregressive process evolves according to  
$$y_{t}=\phi_{}y_{t-1}+\mu\cdot v_{t}\1\left[u_{t}\leq \phi_{}\right],$$ where $0\leq\phi_{}<1$, $\mu{>}0$, { $\epsilon_t{=}(v_t,u_t)'$; see, e.g., \cite{DiCal2006} for II estimation of this model.  The state variable $s_{t}{=}(\mu v_{t},u_t)'$ is iid} and the structural parameters are $\theta{=}(\mu,\phi)'$.
\hfill\(\Box\)

\paragraph{Example 4 (G/G/1 Queue).} 
Let $y_{i}$ denote the inter-departure time for the $i$-th customer. Let $w_{i}$ be the corresponding inter-arrival time and $v_{i}$ the service time, with $v_{i}$ independent of $w_{i}$. Let $\mathbb{E}[w]>\mathbb{E}[v]$, and assume that we know  $v_{i}\sim f_{v}(\cdot;\theta_v)$ and $w_i\sim f_{w}(\cdot;\theta_w)$, with
both $f_{v}(\cdot)$ and $f_{w}(\cdot)$ known up to the unknown
parameters structural parameters $\theta=(\theta_v',\theta_w')'$. The inter-departure time process $\{y_i\}_{i=1}^{n}$ evolves according to, for $j>i$, 
\begin{flalign*}
y_{j}=\begin{cases}
v_j,& \text{ if }\sum_{i=1}^{j}w_{i}\leq\sum_{i=1}^{j-1}y_i\\v_j+\sum_{i=1}^{j}w_{i}-\sum_{i=1}^{j-1}y_i,&\text{ if }\sum_{i=1}^{j}w_{i}>\sum_{i=1}^{j-1}y_i.
\end{cases}
\end{flalign*}{In this example, the state variable is given by $s_i=(v_i,w_{i})'$} and specific parametric assumptions on $v_i$ and $w_i$ can be considered for the purposes of II estimation; see, e.g., {{\cite{heggland2004estimating}}} for a discussion of II estimation in queuing models. \hfill\(\Box\)
\\

{In what follows, to simplify discussion and notations, when discussing general quantities we will only consider a ``cross-sectional'' sample $\{(x_{i}',y_{i})':i=1,\dots,n\}$ from the structural model in equations \eqref{eq:outcome}-\eqref{eq:state}. However, we note that in cases where the observed data has a panel structure, with fixed time dimension $T$, as in Examples 1 and 2, this cross-sectional sample can always be obtained by redefining the observed variables. However, to avoid notational clutter, we will avoid such a scheme and instead focus only on cross-sectional setting. }




\subsection{Standard Indirect Inference}

The focus of II is to conduct estimation and inference on the true parameters of the structural model, denoted throughout by $\theta^0$,
when
the model has a complex parametric structure.
Even when the structural model is complex, it is often easy to simulate data
from the structural model given parameter values.

The first step of II is to estimate an intermediate or  auxiliary model, using the observed data and simulated data, separately.
II estimates of the structural parameters are then obtained
by minimising a well-chosen distance
between the two sets of estimated auxiliary parameters.

To formalise the above,
we denote the simulated unobservables 
by
$\{\epsilon_{i}^{r}\}_{i=1}^{n}$, for $r= 1, \dots, R$, 
where $R$ is the number of simulations for each observation, {and where each $\epsilon_i^r$ is generated iid from the known distribution $F_\epsilon$}. 
We can then construct simulated outcomes
$\{ {y}_{i}^{r}(\theta) \}_{i=1}^{n}$
for any $\theta\in \Theta$
according to
\begin{eqnarray}
  \label{sim1} 
  y_{i}^{r}(\theta)
  &=&
      g
      \big (
      s_{i}^{r}(\theta); \theta
      \big ),       \\
  \notag
  s_{i}^{r}(\theta)
  &=&
      h
      \big (
      x_{i}, \epsilon_{i}^{r};\theta
      \big).
\end{eqnarray}
To form an auxiliary model for the dependent variable $y_{i}$,
let $z_{i}$ be a vector of covariates that
consists of observed variables
with support $\mathcal{Z}$.
The auxiliary model is a tractable parametric model that attempts to capture the relationship between
$y_{i}$ and $z_{i}$.
Let $\mathcal{B} \subseteq \mathbb{R}^{d_{\beta }}$ be the parameter space for the auxiliary parameters with $d_{\beta }\geq d_{\theta }$.
We consider that
the auxiliary model implies 
{moment conditions characterized by
some known moment function 
$m:\mathcal{Y} {\times} \mathcal{Z} {\times} \mathcal{B} \to \mathbb{R}^{d_\beta}$, which satisfies, for some $\beta ^{0} \in \mathcal{B}$, }
\begin{equation}
  \mathbb{E}
  \left  [
m(y_{i},z_{i},\beta^{0} )
  \right ]=0.
  \label{aux}
\end{equation}%
We denote by $\hat{\beta}$ and $\hat{\beta}^{r}(\theta)$
estimators of the auxiliary parameter based on
the observed dataset and the $r$-th dataset simulated from the structural model, respectively. 
The auxiliary parameter estimates $\hat{\beta}$ and
$\hat{\beta}^{r}(\theta)$ are respectively
given as 
the solution to the sample and simulated counterpart of equation \eqref{aux}:
for $\theta \in \Theta$ and $r=1, \dots, R$,
\begin{eqnarray*}
  \frac{1}{n}
  \sum_{i=1}^{n}
  m(y_{i},z_{i},\beta) = 0
  \hspace{1cm}
  \mathrm{and} 
  \hspace{1cm}
  \frac{1}{n}  
  \sum_{i=1}^{n}
  m \big (y_{i}^{r}(\theta),z_{i},\beta \big ) =0.
\end{eqnarray*}

The most common II methods to estimate $\theta^0$ correspond to the so-called ``trinity'' of classical hypothesis tests:
Wald, Lagrange multiplier (LM) and likelihood ratio (LR).
The LM and Wald approaches to II estimation allow computationally simple and efficient estimation of $\theta^0$, in the class of II estimators. In contrast, the LR approach to II does not in general deliver efficient estimators (\citealp{GMR1993}). 
Thus, we focus on the LM and Wald approaches for II, respectively denoted by
LM-II and W-II in what follows.

The LM-II approach is based on the simulated auxiliary moments $M_{n}: \Theta \times \mathcal{B} \to \mathbb{R}^{d_{\beta}}$, given by 
\begin{eqnarray*}
  M_{n}(\theta, \beta)
  :=
  \frac{1}{nR}
  \sum_{i=1}^{n}
  \sum_{r=1}^{R}  
  m \big (y_{i}^{r}(\theta),z_{i}, \beta \big ).
\end{eqnarray*}
For LM-II, the criterion function $Q_{n}^{\text{LM}}: \Theta \to [0,\infty)$
is a quadratic form in the simulated auxiliary moments, evaluated at $\beta=\hat{\beta}$:
\begin{eqnarray*}
  Q_{n}^{\text{LM}}(\theta)
  :=
  \big \|
  M_{n}(\theta, \hat{\beta})
  \big \|^{2}_{\Omega_n},
\end{eqnarray*}
where
$\Omega_{n}$
is a sequence of positive-definite weighting matrices.
The W-II estimator is calculated using the (normed) difference between $\hat{\beta}$ and some version of $\hat{\beta}^r(\theta)$. A common version of the W-II estimator is to  define 
the average simulated auxiliary estimator
$
\bar{\beta}^{R}(\theta)
  :=
  R^{-1}
  \sum_{r=1}^{R}\hat{\beta}^{r}(\theta)$ and then minimize the criterion function $Q_{n}^{\text{W}}: \Theta \to [0,\infty)$, where
\begin{eqnarray*}
  Q_{n}^{\text{W}}(\theta)
  :=
  \big \|
  \bar{\beta}^{R}(\theta)
  -
  \hat{\beta}
  \big \|^{2}_{\Omega_n}.
\end{eqnarray*} 
The LM-II and W-II estimators of $\theta^0$ are then defined as the minimiser of their corresponding criterion functions.

When $y_{i}^r(\theta)$ is discontinuous in $\theta$,  as is the case in the structural model \eqref{sim1}, the resulting II criterion function is discontinuous in $\theta$
and derivative-based optimisation procedures cannot necessarily be trusted to deliver accurate estimates of $\theta^0$. We further examine this discontinuity in the confines of Example 1. 

\vspace{0.5cm}
\paragraph{Example 1 (cont.).}
Let 
$\{(u_{i1}^{r}, \dots, u_{iT}^{r})\}_{i=1}^{n}$
be simulated uniform random variables that are
used to generate simulated outcomes,
for $r=1, \dots, R$.
For fixed $\theta = (\gamma', \rho)'$,
the simulated unobservable term $v^{r}_{it}$ is constructed recursively as
$v_{it}^{r}=\rho v_{i,t-1}^{r}+F_{\epsilon}^{-1}(u_{it}^{r})$, with $v_{i1}^r=F_{\epsilon}^{-1}(u_{i1}^{r})$,
 and the simulated state variable is given by
$s_{it}^{r}(\theta)= x_{it}'\gamma + \rho v^r_{i,t-1}+F_{\epsilon}^{-1}(u_{it}^{r})$. 
Simulated outcomes $y_{it}^{r}(\theta)$ are then generated via
\begin{eqnarray} 
  y_{it}^{r}(\theta)
  =
  \1 [s_{it}^{r}(\theta) > 0 ]
  =
  \1 [
	F_{\epsilon}(-  x_{it}'\gamma - \rho   v^r_{i,t-1} ) 
	< u_{it}^{r}
  ]\label{new1}.
\end{eqnarray}
For this model, \citet{Li2010} and \citet{bruins2015generalized} suggest the linear probability model as an auxiliary model for II:
\begin{eqnarray*}
  y_{it}=z_{it}'\beta+\nu_{it},  
\end{eqnarray*}
where
$ z_{it}=[x_{it}',x_{i,t-1}']'$
and
$\nu_{it}$ is an error term.
We set 
$m(y_{it}, z_{it}, \beta)= z_{it}(y_{it} - z_{it}'\beta )$
as the moment function,
and 
the auxiliary parameter estimates are given by 
\begin{eqnarray*}
  \hat{\beta}
  =
  \bigg (\sum_{i=1}^{n}\sum_{t=2}^{T}z_{it}z_{it}' \bigg)^{-1}
  \sum_{i=1}^{n}\sum_{t=2}^{T}z_{it}y_{it}
  \ \ \mathrm{and} \ \
  \hat{\beta}^{r}(\theta)
  =
  \left(\sum_{i=1}^{n}\sum_{t=2}^{T}z_{it}z_{it}' \right)^{-1}
  \sum_{i=1}^{n}\sum_{t=2}^{T}z_{it}y_{it}^{r}(\theta),
\end{eqnarray*}
for $r=1, \dots, R$.
{
  Given a parameter value $\theta {\in} \Theta$ and simulated samples $\{y_{it}^{r}(\theta)\}_{r=1}^{R}$,
  one can then construct the criterion function for 
  the LM-II or the W-II approach. However, because the map $\theta \mapsto y_{it}^{r}(\theta)$ is discontinuous,
  derivatives of $Q_{n}^{\text{LM}}(\theta)$ and $Q_{n}^{\text{W}}(\theta)$
  need not exist.}
\hfill\(\Box\)
 
\vspace{0.5cm}


In each of the examples treated in the previous subsection, II estimation is based on the discontinuous mapping $\theta \mapsto y^r_{i}(\theta)$. As a result, derivative-based optimisation procedures may not deliver accurate estimates for the unknown model parameters.
To solve this issue, we propose to alter the standard II simulation approach by introducing a (sequence of) change of variables that will alleviate the discontinuity in the map $\theta \mapsto y^{r}_{i}(\theta)$. This alternative simulation approach will deliver II criterion functions that allow the application of derivative-based optimisation routines that can consistently estimate the unknown model parameters, even though the original model is non-smooth in the parameters.

\section{A New Generalized Indirect Inference Approach}\label{general}
Let $\theta^*$ denote a point at which we wish to evaluate the LM-II or W-II criterion function. According to equation \eqref{sim1}, $y_{i}^{r}(\theta)$ is discontinuous and the derivatives $\partial_{\theta} Q^{\text{LM}}_{n}(\theta^*)$
and 
$\partial_{\theta} Q^{\text{W}}_{n}(\theta^*)$
need not exist. The logic behind our approach is the observation that II is not required to simulate outcomes that perfectly represent the actual data generating process, so long as the difference between the two vanishes as the sample size increases. 

Define the population criterion functions for the LM and Wald approaches as 
\begin{eqnarray*}
  Q^{\text{LM}}(\theta)
  :=
  \big \|
  M(\theta, \beta^{0})
  \big\|_{\Omega}^{2}
  \ \ \ \mathrm{and} \ \ \
  Q^{\text{W}}(\theta)
  :=
  \big \|
  \beta(\theta) - \beta^{0}  
  \big \|^{2}_{\Omega},
\end{eqnarray*}
where 
$
M(\theta, \beta)
:=
\mathbb{E}
[
{m}(y_{i}^{r}(\theta),z_{i}, \beta)
]
$,
$\beta(\theta)$ is 
the solution to 
$
M(\theta, \beta)=0
$
and 
$\Omega$ is some positive definite matrix, and where $Q^{{l}}(\theta)$, $l=\text{LM,W}$, is assumed to be twice continuously differentiable. 
In a nutshell, our approach is to construct a ``generalized" II criterion, say
$Q_{n}^{\text{LM}}(\theta,\theta^*)$
or
$Q_{n}^{\text{W}}(\theta,\theta^*)$
for $\theta, \theta^{\ast} \in \Theta$, satisfying: for $l=\text{LM, W}$,
\begin{eqnarray*}
  \plim_{n\rightarrow\infty}
  \partial^{}_{\theta}Q^{l}_{n}(\theta,\theta^*)\big{|}_{\theta=\theta^*}
  =\partial_{\theta} Q^{l}(\theta^{\ast})
  \ \ \ \mathrm{and} \ \ \
  \plim_{n\rightarrow\infty}\partial^{2}_{\theta }Q^{l}_{n}(\theta,\theta^*)\big{|}_{\theta=\theta^*}=\partial^{2}_{\theta} Q^{l}(\theta^{\ast}).
\end{eqnarray*}
In the following section, we demonstrate how to construct such a criterion using a ``change of variables'' (hereafter, COV) technique. This COV will ensure that a small change in the parameter, say from $\theta^*$ to $\theta^*\pm\delta$, for small $\delta>0$, will not drastically alter the simulated data, thus alleviating the discontinuity. 

\subsection{Assumptions}

To present the COV technique used in this paper, we first introduce assumptions on the structural model in equations \eqref{eq:outcome}-\eqref{eq:state}
and the auxiliary model. 
The assumptions presented below are not restrictive and, with minor modifications, cover all of the examples referenced in this paper.

\vspace{0.3cm}
\begin{assumption}\label{one} \ 
  \normalfont
  \begin{itemize}
  \item[(a)]
    The observed random variables
    $(x_{i}', y_{i})'$  
    are iid~across cross-sectional unit $i=1,\dots,n$, and
    $x_{i}$ is independent of the innovation $\epsilon_{i}$.     
  \item[(b)]
    The innovation $\epsilon_{i}$ is iid
    and
    with known
    continuously differentiable density, $f_\epsilon$.
  \item[(c)]
    The function
    $h(x, \epsilon; \theta)$
    is twice continuously differentiable in
    $\epsilon \in \mathcal{E}$
    and   
    $\theta \in \Theta$
    for any $x \in \mathcal{X}$.
  \item[(d)] The parameter spaces $\Theta$ and $\mathcal{B}$ are compact. 
  \end{itemize}
\end{assumption}
 
\begin{assumption}\label{two}\normalfont \
  \begin{itemize}
  \item[(a)]
    Assumption 1 holds for the simulated process
    for all $\theta \in \Theta$.
  \item[(b)]
    $m(y, z, \beta)$ is continuous at each $\beta\in\mathcal{B}$, for any
    $(y,z) \in \mathcal{Y} \times \mathcal{Z}$.
  \item[(c)]
For all $i=1,\dots, n$,
there exists a random variable $\bar{m}_{i}$ such that
$\|m(y, z_{i}, \beta)\|\leq
\bar{m}_{i}$
with 
$\mathbb{E}
[
\bar{m}_{i}^2] 
<\infty$
for all $\beta\in\mathcal{B}$
and for all $y \in \mathcal{Y}$.
  \end{itemize}
\end{assumption}

\begin{assumption}
  \label{g} \normalfont \ 
  \begin{itemize}
  \item[(a)]
    For any $i = 1, \dots, n $ and $\theta \in \Theta$,
    there exist a finite integer $J_{}$
    and 
    a collection of random functions
    $\{c_{i}^{j}(\theta) \}_{j=0}^{J_{}+1}$
    with 
    $c_{i}^{j}:\Theta {\to} [0,1]$
    such that 
    the function $g(\cdot)$ in (\ref{eq:outcome})
    can be written as 
    \begin{eqnarray*}
      g(s_{i}; \theta)
      =
      \sum_{j=0}^{J}\alpha_{j} 
      \1[c_{i}^{j}(\theta) < u_{i} \le c_{i}^{j+1}(\theta)].
    \end{eqnarray*}
    Here
    $\alpha_{0}, \dots, \alpha_{J}$
    are known constants
    and 
    the random variable $u_{i}:=F_{\epsilon}(\epsilon_{i})$ follows the standard uniform. 
    Also, the random function $\{c_{i}^{j}(\theta)\}_{j=0}^{J+1}$
    are twice-continuously differentiable 
    and satisfy that 
    $c_{i}^{0}(\theta) = 0$,
    $c_{i}^{J+1}(\theta) = 1$,
    and 
    $c_{i}^{j}(\theta) < c_{i}^{j+1}(\theta)$.  
 \item[(b)]
   For each $\ell \in \{1,2\}$,
   there exists a random variable $\nabla^{\ell} \bar{c}_{i}$
   so that 
   $\sup_{\theta \in \Theta}\|\partial_{\theta}^{\ell} c_{i}^{j}(\theta)\| \le
   \nabla^{\ell} \bar{c}_{i}
   $
   and
   $E| \nabla^{\ell} \bar{c}_{i} |^2 < \infty$
   for any $(i, j) \in \{1, \dots, n\}{\times}\{0, \dots, J\}$.
\end{itemize}

\end{assumption}
\vspace{0.5cm} 

{The above assumptions are fairly weak, and in the supplemental appendix, we demonstrate that, up to minor modifications, Assumptions 1-3 are satisfied for each of the examples presented in Section 2.1. The specific interpretation of the assumptions is as follows.} Assumption \ref{one} requires that the discontinuities in $y_{i}^{r}(\theta)$ arise only from $g(\cdot)$ and ensures that we can present the COV in the most general context. {The Assumption of iid data in \ref{one}(a) can be extended to independent non-identically distributed (inid) data, or weakly dependent data, at the cost of further notation and more involved technical arguments. In particular, under these more general assumption, only the structure of the random functions $c_j^j(\theta)$ that partition the support of $u_{i}$ will change; we refer the interested reader to the supplementary appendix for an example with weakly dependent data.}
Throughout the remainder, we will refer to $\{c^{j}_{i}(\theta)\}_{j=0}^{J +1}$ as critical point functions.
Assumption \ref{two} imposes regularity conditions on the simulated moment function. The regularity conditions in Assumption \ref{two} are standard in the literature on II with discontinuous outcomes.
Assumption \ref{g} restricts the analysis to univariate $\epsilon_{i}$, however, the extension to multivariate $\epsilon_{i}$ is almost automatic, and can be accomplished by further decomposing the simulation algorithm into corresponding scalar innovations. Assumption \ref{g} formalizes the structure of the model in a way that ensures we can sequester the discontinuities into various regions of the support for $\epsilon_{i}$. As such, Assumption \ref{g} requires an upper bound on the number of possible discontinuities
in $y^{r}_{i}(\theta)$ that arise from the structural function  $g(\cdot)$.
The term $J$ can be interpreted as the maximum number of discontinuities allowed by the model.  
For example, in the dynamic binary panel models, $J=1$ for all $i$ and $t$.
Under Assumption \ref{g}(a),
the discontinuous sizes, $\alpha_{0}, \dots, \alpha_{J}$, {are assumed to be
known constants, which holds for the first two examples examined in Section 2.1. A minor modification of Assumption \ref{g}(a) also covers the second two examples covered in Section 2.1. Assumption \ref{g}(a) is employed for simplicity and can be extended, at the cost of additional notation, to the case where
the discontinuous sizes, $\{\alpha_{j}\}_{j=0}^{J}$, are twice-differentiable functions of
the structural parameters.
Assumptions \ref{g}(a)-(b) are needed to establish uniform convergence
over the parameter space and are used for our asymptotic analysis given later.

{Assumption \ref{g}(a) restricts the class of distributions for the errors $\epsilon_i$ to be continuous. While it may be feasible to extend this approach to cases where $\epsilon_i$ is discrete, given that these examples are much less frequent in econometrics than their continuous counterpart, we do not consider such situations herein.  Before concluding, we note that even though $\epsilon_i$ is restricted to be a continuous random variable, this assumption is satisfied in a wide variety of examples, such as those given in Examples 1-4, as well as censored-type models, such as Tobit-type models.}

\subsection{Approximate Derivatives}

Under Assumptions \ref{one}-\ref{g}, we carry out a change of variables on the original uniform random variables $\{u_{i}^{r}\}$ to construct new simulated outcomes that will be differentiable.
Let $\theta^* \in \Theta$ denote a point at which we wish to evaluate the derivative of the function $\theta \mapsto y^r_{i}(\theta)$, and  consider the following transformation of $u_i^r$:\footnote{There are many such transformations that will accomplish our goal, see \cite{chan2009minimal} and \cite{Lyuu}. The above transformation is chosen as there is theoretical evidence to suggest that $u_{i}(\theta,\theta^*)$ in \eqref{new_U} is optimal, in terms of minimizing mean squared error, for indicator functions (\citealp{joshi2016optimal}). Since many of the commonly encountered discontinuities in II arise from simulating indicator function, this optimality should transfer to our settings. 
}
\begin{equation}
  u_{i}^{r}(\theta,\theta^*)
  :=
  c^{j}_{i}(\theta)
  +
  \frac{
    c^{j+1}_{i}(\theta)-c^{j}_{i}(\theta)
  }{
    c^{j+1}_{i}(\theta^*)-c^{j}_{i}(\theta^*)
  }
  \{u_{i}^{r}-c^{j}_{i}(\theta^*)\},
  \label{new_U}
\end{equation}
for
$c^{j}_{i}(\theta^*)<u_{i}^{r}\leq c^{j+1}_{i}(\theta^*)$
with $j = 0, \dots, J$, and
the corresponding Jacobian term of this COV
is defined as 
$w^{r}_{i}(\theta,\theta^*):=
\partial u_{i}^{r}(\theta,\theta^*)/ \partial u_{i}^{r}
$, i.e.,
for
$c^{j}_{i}(\theta^*)<u_{i}^{r}\leq c^{j+1}_{i}(\theta^*)$,
\[
w^{r}_{i}(\theta,\theta^*)=\frac{c^{j+1}_{i}(\theta)-c^{j}_{i}(\theta)}{c^{j+1}_{i}(\theta^*)-c^{j}_{i}(\theta^*)}.
\]

Under Assumption \ref{g}(a), the difference
$
  c^{j+1}_{i}(\theta^*)-c^{j}_{i}(\theta^*)
$
is strictly positive almost surely for all $\theta^{\ast} \in \Theta$, and all $(i, j)\in \{1, \dots, n\} {\times} \{0, \dots, J\}$. {Thus, there exists a constant $\bar{w}$
	such that
	$|w_{i}^{r}(\theta, \theta^{\ast})| \le \bar{w}$ 
	for all $(\theta, \theta^{\ast}) \in \Theta^2$.
 This boundedness of the Jacobian term $w_i^r(\theta,\theta^*)$, for all $(\theta,\theta^*) \in\Theta^2$, is critical to derive the uniform Laws of Large Numbers that will be needed to obtain uniform convergence of our proposed II criterion function and their derivatives.}
Moreover, under Assumption \ref{g}(b), 
thus there exists a $\nabla \bar{w}_{i}$
such that
$\|\partial_{\theta}w_{i}^{r}(\theta,\theta^{\ast})\|
\le \nabla \bar{w}_{i}
$
for all $(\theta, \theta^{\ast}) \in \Theta^2$
and
$\mathbb{E}|\nabla \bar{w}_{i}^2| < \infty$.

Given the transformed series $\{u_{i}^{r}(\theta, \theta^{\ast})\}$,
for $r=1, \dots, R$,
we then construct new simulated outcomes according to 
\begin{eqnarray*}
  y_{i}^{r}(\theta,\theta^*)
  :=
  \sum_{j=0}^{J}
  \alpha_{j}
  \1\left[c_{i}^{j}(\theta)<u_{i}^{r}(\theta,\theta^*)\leq c_{i}^{j+1}(\theta)\right].
\end{eqnarray*}
II can now proceed by replacing the moment function
$m\big (y_{i}^{r}(\theta), z_{i}, \beta \big )$
with the following moment function:
\begin{eqnarray*}
  {m}_{i}^{r}(\theta,\theta^*,\beta)
  :=
  m\big (y_{i}^{r}(\theta,\theta^*),z_{i},\beta \big)\cdot
  w_{i}^{r}(\theta,\theta^* ),
\end{eqnarray*}
and the moment conditions for II estimation can be defined through the function
$M_{n}: \Theta {\times} \Theta {\times} \mathcal{B} \to \mathbb{R}^{d_{\beta}}$,
given by 
\begin{eqnarray*}
  M_{n}(\theta, \theta^{\ast}, \beta)
  :=
  \frac{1}{nR}
  \sum_{i=1}^{n}
  \sum_{r=1}^{R}  
  {m}_{i}^{r}(\theta,\theta^*,\beta).
\end{eqnarray*}

The result below shows that the derivatives of the moment function $M_{n}(\theta, \theta^{\ast}, \beta)$  
 with respect to $\theta$ are unbiased and uniformly consistent estimators of their corresponding limit counterparts. 

\vspace{0.5cm}
\begin{proposition}\label{mc1}
  Suppose that 
  Assumptions \ref{one}-\ref{g} hold.
  Then, 
  the first and second derivatives
  of $M_{n}(\theta, \theta^{\ast}, \beta)$ with respect to $\theta$, calculated at $\theta=\theta^*$,
  exist and are unbiased estimators for
  the first and second derivatives  of
  $M(\theta^*, \beta)$ with respect to $\theta$, calculated at $\theta=\theta^*$. Furthermore,
  as $n\rightarrow\infty$,  
  \begin{eqnarray*}
    \partial_{\theta}
    M_{n}(\theta, \theta^{\ast}, \beta)
    \big{|}_{\theta=\theta^*}
    \pto
    \partial_{\theta}
    M(\theta^{\ast}, \beta)    
     \ \ \mathrm{and}  \ \ 
    \partial_{\theta_{k}}
    \partial_{\theta_{l}}
    M_{n}(\theta, \theta^{\ast}, \beta)
    \big{|}_{\theta=\theta^*}
    \pto
    \partial_{\theta_{k}}
    \partial_{\theta_{l}}
    M(\theta^{\ast}, \beta),
  \end{eqnarray*}
  uniformly in $(\theta^{\ast},\beta) \in \Theta \times \mathcal{B}$
  and for every $k ,l \in \{1 \dots, d_{\theta}\}$.
\end{proposition}
\vspace{0.5cm}

The above result demonstrates that simulated moments produced by this procedure have the derivatives with respect to $\theta$, evaluated at $\theta^*$,
which are consistent estimates of their limit counterparts. Therefore, this COV approach allows us to construct ``generalized'' LM-II and Wald-II criterion functions as
\begin{eqnarray*}
  Q^{\text{LM}}_{n}(\theta, \theta^{\ast})
  :=
  \big \|
  M_{n}(\theta, \theta^{\ast}, \hat{\beta})
  \big\|_{\Omega_{n}}^{2}
  \ \ \ \mathrm{and} \ \ \
  Q^{\text{W}}_{n}(\theta, \theta^{\ast})
  :=
  \big \|
  \bar{\beta}^{R}(\theta,\theta^{\ast}) - \hat{\beta}
  \big \|^{2}_{\Omega_{n}},
\end{eqnarray*}
where 
$\bar{\beta}^{R}(\theta,\theta^*)
:=
R^{-1}\sum_{r=1}^{R}\hat{\beta}^{r}(\theta,\theta^*)$
with 
the estimator 
$\hat{\beta}^{r}(\theta,\theta^*)$
satisfying
the following moment condition: 
$
  n^{-1}
  \sum_{i=1}^{n}m_{i}^r(\theta,\theta^*,\beta) = 0 
$
for each $r = 1 \dots, R$
and given $(\theta, \theta^{\ast}) \in \Theta^2$.  We can then define generalized II estimators based on this COV approach, and denoted by $\hat{\theta}^{\text{LM}}$ and $\hat{\theta}^{\text{W}}$, via the minimisation problem
\begin{flalign}
  \hat{\theta}^{{l}}
  =\arg\min_{\theta\in\Theta}Q^{l}_{n}(\theta,\theta),
  \hspace{0.5cm} \mathrm{for} \ l \in \{\text{LM, W}\}.\label{aii_1}
\end{flalign}
Hereafter, we refer to such II estimators as GII change of variable (GII-COV) estimators.

\vspace{0.05cm}
\paragraph{Example 1 (cont.).} 
To implement our GII-COV approach, recall that the standard simulated outcome was
$y_{it}^r(\theta)=\1[F_{\epsilon} (-x_{it}'\gamma - \rho v_{i,t-1}^{r} ) <u_{it}^{r}].$
Now, consider
the critical point functions:\footnote{Recall that, while the critical point functions can depend on the simulated data set $r$, we alleviate this dependence to simplify notations.}
\begin{flalign*}
c_{it}^{0}(\theta) = 0 , \ \  \
c_{it}^{1}(\theta)=F_{\epsilon}
(- x_{it}'\gamma-\rho v_{i,t-1}^{r} ) 
\ \ \ \mathrm{and} \ \ \
c_{it}^{2}(\theta) = 1. 
\end{flalign*}
Let $\theta^* = (\gamma^{\ast \prime}, \rho^{\ast})'$ be a point at which we wish to
evaluate the function $\theta \mapsto y_{it}^{r}(\theta)$. As in (\ref{new_U}), we construct $u_{it}^{r}(\theta, \theta^{\ast})$
as follows: 
\begin{eqnarray*}
  u_{it}^{r}(\theta, \theta^{\ast}) =
  \left \{
  \renewcommand{\arraystretch}{1.6}
  \begin{array}{ll}
    \dfrac{
    c^{1}_{it}(\theta)
    }{
    c^{1}_{it}(\theta^*)
    }
    u_{it}^{r}, 
    &
      \text{ if }u_{it}^{r}\le c^{1}_{it}(\theta^*) \\[0.3cm]
    c_{it}^{1}(\theta)
    +
    \dfrac{
    1-c_{it}^{1}(\theta)
    }{
    1-c_{it}^{1}(\theta^*)
    }
    \{ u_{it}^{r}-c^{1}_{it}(\theta^*) \}
    ,
    &\text{ if } c^{1}_{it}(\theta^*) < u_{it}^{r}.
  \end{array}
  \right .
\end{eqnarray*}
The corresponding Jacobian term $w_{it}^{r}(\theta, \theta^{\ast})$ depends on $\theta$ only though the critical point function $c_{it}^{1}(\theta)$
and is differentiable in $\theta$.
The new simulated outcomes are generated exactly as the original outcomes,
$y_{it}^{r}(\theta)=\1[c_{it}^{1}(\theta) < u_{it}^{r} ]$, except that the new outcomes are simulated by replacing $u_{it}^{r}$ with the uniforms $u_{it}^{r}(\theta,\theta^*)$.
That is,
$
y_{it}^{r}(\theta, \theta^{\ast})=
\1[c_{it}^{1}(\theta) < u_{it}^{r}(\theta, \theta^{\ast})  ]
$. 
Using  $y_{it}^{r}(\theta,\theta^*)$ and $w^{r}_{it}(\theta,\theta^*)$,
we can then construct the approximation to the moment function
\begin{eqnarray*}
  M_{n}(\theta, \theta^{\ast}, \beta)
  =
  \frac{1}{nR}
  \sum_{i=1}^{n}
  \sum_{r=1}^{R}
  \sum_{t=2}^{T}
  z_{it}
  \big (
  y_{it}^{r}(\theta, \theta^{\ast}) - z_{it}'\beta
  \big )
  w_{it}^{r}(\theta, \theta^{\ast}),
\end{eqnarray*}
and to the II binding function,
\begin{eqnarray*}
  \bar{\beta}^{R}(\theta,\theta^*)
  =
  \bigg(
  \sum_{i=1}^{n}\sum_{t=2}^{T}z_{it}z_{it}'
  \bigg)^{-1}
  \sum_{i=1}^{n}\sum_{t=2}^{T}z_{it}
  \frac{1}{R}
  \sum_{r=1}^{R}y_{it}^{r}(\theta,\theta^*)w^{r}_{it}(\theta,\theta^*).  
\end{eqnarray*}

By the definition of $u_{it}^{r}(\theta, \theta^{\ast})$ above, 
we have that
$c_{it}^{1}(\theta) < u_{it}^{r}(\theta, \theta^{\ast}) $
if and only inf 
$  c^{1}_{it}(\theta^*) < u_{it}^{r}$,
which yields an alternative representation:
$
  y_{it}^{r}(\theta, \theta^{\ast})
  = 
\1[c_{it}^{1}(\theta^{\ast}) < u_{it}^{r}].
$
Thus, 
both
$
M_{n}(\theta, \theta^{\ast}, \beta)
$
and 
$
\bar{\beta}^{R}(\theta,\theta^*)
$
are functions of $\theta$
only through the Jacobian term $w_{it}(\theta, \theta^*)$,
which is differentiable in $\theta$. Hence, both approximation functions are differentiable in $\theta$;
for example, 
\begin{eqnarray*}
  \partial_{\theta} M_{n}(\theta, \theta^{\ast}, \beta)
  =
  \frac{1}{nR}
  \sum_{i=1}^{n}
  \sum_{r=1}^{R}
  \sum_{t=2}^{T}
  z_{it}
  \big (
  \1[c_{it}^{1}(\theta^{\ast}) < u_{it}^{r}]
  - z_{it}'\beta
  \big )
  \partial_{\theta} w_{it}^{r}(\theta, \theta^{\ast}).
\end{eqnarray*}
Moreover, by Proposition \ref{mc1}, the approximation derivative 
$\partial_{\theta} M_{n}(\theta, \theta^{\ast}, \beta)$,
evaluated at $\theta = \theta^{\ast}$,
is an unbiased and uniformly consistent estimator for the population counterpart
$\partial_{\theta} M(\theta^{\ast}, \beta)$.
\hfill\(\Box\)

\hspace{0.5cm}

While the GII-COV estimators can be defined as in equation \eqref{aii_1}, it is nonetheless interesting to note that the criterion functions $Q_{n}^{\text{LM}}(\theta,\theta)$ and $Q_{n}^{\text{W}}(\theta,\theta)$ have very regular behaviour. In particular, from the results of Proposition \ref{mc1} we can deduce that these II criterion functions are consistent estimators of their corresponding limit counterparts. 
To obtain this consistency result, the following assumption is additionally employed.

\vspace{0.2cm}
\begin{assumption}\label{three}
  \normalfont Recall
  $M(\theta,\beta)= \E[m(y_{i}^{r}(\theta),z_{i},\beta)]$. 
  \begin{itemize}
  \item[(a)]  
    $M(\theta^0,\beta) =
    \E[m(y_{i},z_{i},\beta)]    
    $
    for all $\beta\in\mathcal{B}$.
  \item[(b)]
    The parameter
    $\beta^0$
    lies in the interior of $\mathcal{B}$
    and 
    is the unique solution to $M(\theta^{0}, \beta)=0$. 
  \item[(c)]
    $\Omega_{n}\pto \Omega$ for some positive-definite matrix $\Omega$. 
  \end{itemize}
\end{assumption}
\vspace{0.2cm}

Assumption \ref{three}(a) ensures that the structural model is correctly specified.
Assumption \ref{three}(b) requires that there exists an unique auxiliary parameter vector satisfying the population moment condition evaluated
at the true structural parameter $\theta^{0}$.
From Assumption \ref{three}(c),
the possibly random weighting matrix $\Omega_{n}$ converges to
a positive-definite matrix $\Omega$. 


\vspace{0.5cm}
\begin{theorem}
  \label{theorem:Q}
  Suppose that Assumptions \ref{one}-\ref{three} hold
  and let
  $\ell \in \{0,1, 2\}$.
  Then,
  \begin{description}
  \item[(a)]
    uniformly in $\theta^{\ast} \in \Theta$
    as $n\rightarrow\infty$, 
    \begin{eqnarray*}
      \partial_{\theta}^{\ell} Q_{n}^{\text{\em LM}}(\theta, \theta^{*})|_{\theta = \theta^{\ast}}
      \pto
      \partial_{\theta}^{\ell} Q^{\text{\em LM}}(\theta^{\ast});
    \end{eqnarray*}
  \item[(b)]
    additionally, if
    (i)
    the map $\beta \mapsto m(y,z,\beta)$ is twice continuously differentiable
    given any $(y,z) \in \mathcal{Y} \times \mathcal{Z}$,
    (ii)
    there exists a random variable $\nabla^{\xi} \bar{m}_{i}$ 
    so that 
    $\sup_{\beta \in \mathcal{B}}\|\partial_{\beta}^{\xi} m(y_{i}, x_{i}, \beta)\| \le
    \nabla^{\xi} \bar{m}_{i}$
    and 
    $\mathbb{E}| \nabla^{\xi} \bar{m}_{i} |^2 < \infty$ for $\xi = 1, 2$,
    (iii)
    $\mathbb{E}
    \big[ \partial_{\beta}m \big ( y_{i}^{r}(\theta), z_{i}, \beta \big)
    \big]
    $ is non-singular
    for any $(\theta, \beta) \in \Theta \times \mathcal{B}$,
    and
    (iv)
    $\log(R) / n \to 0$, then 
 \begin{eqnarray*}
 	\partial_{\theta}^{\ell} Q_{n}^{\text{\em W}}(\theta, \theta^{*})|_{\theta = \theta^{\ast}}
 	\pto
 	\partial_{\theta}^{\ell} Q^{\text{\em W}}(\theta^{\ast}),
 \end{eqnarray*}
 uniformly in $\theta^{\ast} \in \Theta$
 as $n\rightarrow\infty$. 
  \end{description}
\end{theorem}
\vspace{0.5cm}

As a result of Theorem \ref{theorem:Q}, the differentiability of $Q_{n}^{\mathrm{LM}}(\theta,\theta^{\ast})$
and
$Q_{n}^{\mathrm{ W}}(\theta,\theta^{\ast})$,
with respect to $\theta$,
also allows us to define the GII-COV estimators
$\hat{\theta}^{\text{LM}}$ and $\hat{\theta}^{\text{W}}$ 
as the solution to the first-order conditions:
\begin{eqnarray*}
  \partial_{\theta} Q^{\text{LM}}_{n}\big (\theta, \hat{\theta}^{\text{LM}} \big)
  \big |_{\theta = \hat{\theta}^{\text{LM}}}
  = o_p(1)
  \hspace{1cm} \mathrm{and} \hspace{1cm} 
  \partial_{\theta}
  Q^{\text{W}}_{n} \big (\theta, \hat{\theta}^{\text{W}} \big)
  \big |_{\theta = \hat{\theta}^{\text{W}}}
  = o_p(1),
\end{eqnarray*}
subject to the satisfaction of the relevant second-order conditions.
From these characterisations, standard derivative-based optimisation routines can be
used to obtain $\hat{\theta}^{\text{LM}}$ and $\hat{\theta}^{\text{W}}$. 

\begin{remark}
  {\normalfont
    For the W-II estimator, Theorem \ref{theorem:Q} requires {several additional assumptions that are not required by the LM-II estimator. Since the auxiliary estimator is only defined as an implicit solution of the auxiliary moments, conditions (i)-(iii) in part \textbf{(b)} of the result, along with Assumptions 1-3, are needed to ensure that the auxiliary moment equations are regular enough to guarantee that: (i) the auxiliary parameter estimator exists; (ii) auxiliary parameter estimator is well-behaved. While not explicitly proven herein, we note that if the auxiliary estimator has a closed-form, these assumptions can be relaxed and the result of Theorem \ref{theorem:Q} \textbf{(b)} will follow under the same assumptions as part \textbf{(a)} of the result. Similarly, due to the implicit nature of the solution for the auxiliary estimator, an additional condition
    on the number of simulations $R$ relative to the sample size $n$ is required.}
    This rate requirement is extremely mild
    and 
    covers any feasible implementation of II estimation.
    For instance, it allows for a fixed number of simulations
    as well as the same order as the sample size.\footnote{{
      This condition is violated if researchers choose
      extremely large simulation sizes; for instance, 
      $R = e^n$.
      In this case, however, even for small sample sizes, say $n=100$,
      researchers would have to generate and store more than $e^{100}$
      simulated paths, each of length $n=100$.  }
    }   
    This condition
    together with 
    a type of maximal inequality over the simulated paths 
    ensures 
    the uniform convergence
    of 
    $\partial_\theta^\ell \hat{\beta}^r(\theta,\theta^{\ast})|_{\theta = \theta^{\ast}}$
    to 
    $\partial_\theta^\ell \beta(\theta^{\ast})$
    over $r \in \{1, \dots, R\}$
    in probability
    for every $\ell \in \{0, 1, 2\}$,
    as in Lemma \ref{lemma:beta} in the appendix,
    and
    thus leads to 
    the uniform convergence of
    $\partial_\theta^\ell Q_{n}^{\mathrm{W}}(\theta,\theta^{\ast})|_{\theta = \theta^{\ast}}$, $\ell\in\{0,1,2\}$.
    When $\log (R)/n \to c \in (0, \infty]$,
    the above result for the W-II estimator cannot be applied
    and 
    its asymptotic properties are beyond the scope of this paper,
    while 
    the result for the LM-II estimator is valid.
  }
\end{remark}

\begin{remark}
  {\normalfont It is important to realise that this COV must be carried out at any value of $\theta^{\ast}$ for which we wish to calculate
    $\partial_\theta Q_{n}^{\text{LM}}(\theta,\theta^{\ast})|_{\theta=\theta^*}$ or $\partial_\theta Q_{n}^{\text{W}}(\theta,\theta^{\ast})|_{\theta=\theta^*}$. However, the COV only impacts the calculation of the derivatives for the LM and Wald criterion function:
    $\partial_{\theta}{m}_{it}^{r}(\theta,\hat{\theta}^{\text{LM}},\beta)|_{\theta=\hat{\theta}^{\text{LM}}}$
    for the LM-II criterion, and   $\partial_{\theta}\bar{\beta}^{R}(\theta,\hat{\theta}^{\text{W}})|_{\theta=\hat{\theta}^{\text{W}}}$
    for the W-II criterion. Therefore, the only portion of this procedure that is ``approximate'' are the derivatives. In this way, this approach is similar to \cite{bruins2015generalized} in that our approach is a generalized II of a sort.  }
\end{remark}

\begin{remark}
	{\normalfont Computing derivatives using numerical finite-differencing has long been the standard approach in econometrics. However, it is well-known that such methods exhibit a trade-off between variance and bias that depends on the choice of the tuning parameter. In general, when central-differencing with common random numbers is employed, an optimal tuning parameter can be chosen as a function of the number of simulated random numbers, say $R$, which yields an optimal convergence rate of  $R^{-2/5}$ (\citealp{glynn1989optimization}). However, if the function is also continuous for almost all values of the random numbers, then the optimal tuning parameter is zero and the optimal convergence rate is $R^{-1/2}$.  }	 
\end{remark}	
  
\begin{remark}
  {\normalfont We note that there are alternative derivative estimation methods that could potentially alleviate the difficulties associated with non-smooth simulated criterion functions. {The most intuitive approach is via the so-called ``likelihood ratio method'' (\citealp{glynn1987likelilood}),\footnote{This notion of the likelihood ratio method is not related, in any sense, to the notion of the likelihood ratio approach to indirect inference, but is related to estimating derivatives of expectations at a fixed point in the parameter space.}} which, through a deliberate choice of the sample space, pushes the parameter of interest into the density function generating the simulated data.  The application of this method in the context of II is heavily dependent on the model and the ability to analytically evaluate this density as a function of the model parameter. Another possible approach is via  Malliavin integration-by-parts (\citealp{fournie1999applications}), which extends the calculus of variations from functions to stochastic processes.
    Also, the mollification method \citep{friedrichs1944identity}
    provides a general approach to approximate non-smooth functions through convolution,
    while the method generally cannot obtain unbiased derivative estimators.\footnote{
      The mollification is an approximation method for non-smooth functions through a convolution with some suitable differential functions, so-called mollifiers. The mollification can be considered as a local approach as in our approach, wheares it depends on smoothing
      parameters, which cause bias in its derivative estimator. We thank one of the anonymous referees for pointing out the relation between our approach and the mollification.}
    The exact implementation of these methods are unclear at this stage    
    for complicated models in econometrics, especially those with non-separable errors. While these alternative methods are potentially useful, a thorough exploration of these methods in this II context is beyond the scope of these paper. 
  }	
\end{remark}


\section{Asymptotic Properties}

This section presents the asymptotic properties of the GII-COV estimator.
We first show consistency of the estimator and
establish asymptotic normality. 
Subsequently, we explain how GII-COV aids the consistent estimation of
the asymptotic variance.

\vspace{0.3cm}
\begin{theorem}\label{cons}
  Suppose that Assumptions \ref{one}-\ref{three} hold.
  Additionally,
  for the Wald estimator,
  assume the condition stated 
  in 
  Theorem \ref{theorem:Q}(b).
  Then,
  as $n\rightarrow\infty$,
  \begin{eqnarray*}
    \hat{\theta}^{\text{\em LM}}_{n} \pto \theta^0    
    \ \ \ \mathrm{and} \ \ \  
    \hat{\theta}^{\text{\em W}}_{n}\pto \theta^0.    
  \end{eqnarray*} 
\end{theorem}
\vspace{0.3cm}

For asymptotic normality of the GII-COV estimator,
we require an additional regularity condition.

\begin{assumption} \
  \label{four}
  \normalfont
  \begin{itemize}
  \item[(a)]
    For some $\delta > 0$
    and 
    for any $(y, z) \in \mathcal{Y} \times \mathcal{Z}$,
    the function $m(y,z,\beta)$ is continuously differentiable in
    $\beta \in \mathcal{B}$
    with $\|\beta - \beta^{0} \| \le \delta$, and
    $\E\big [
    \sup_{\beta \in \mathcal{N}_\delta(\beta^{0})}
    \|
    \partial_{\beta} m(y_{i},z_{i},\beta)\|
    \big] < \infty$.
  \item[(b)]
    $\partial_{\beta} 
    M(\theta^{0}, \beta^{0})
    $ is non-singular.
  \item[(c)]
    Define 
    $\xi_{i}^r
    :=
    m(y_{i}^r(\theta^0),z_{i},\beta^0)
    -
    m(y_{i},z_{i},\beta^0)
    $
    and
    $\Xi:=
    \E [
    \xi_{i}^r
   \xi_{i}^{r'}
    ]
    $.
    \begin{eqnarray*}
      \frac{1}{\sqrt{n} R}
      \sum_{r=1}^{R}\sum_{i=1}^{n}
      \xi^r_{i}
      \dto
      N
      \big (0, (1+R^{-1}) \Xi \big).     
    \end{eqnarray*}
  \item[(d)] $\partial_{\theta}M(\theta,\beta^0)$ is continuous in $\theta$ and, for some $\delta>0$ and $\theta\in\mathcal{N}_{\delta}(\theta^0)$, has rank $d_{\theta}$.
  \end{itemize}
\end{assumption}

The regularity conditions in Assumption \ref{four} are fairly weak and, in conjunction with Assumptions \ref{one}-\ref{three}, ensure enough regularity on the auxiliary moments to ensure GII-COV estimators are asymptoticly normal. Moreover, unlike the high-level stochastic equicontinuity conditions required to demonstrate asymptotic normality of II estimators based on non-smooth criterion functions, see, e.g., \cite{pakes1989simulation}, \cite{chen2003estimation} and \cite{CFR2017}, the GII-COV approach yields asymptoticly normal estimators under much more primitive conditions. In particular, under Assumptions \ref{one}-\ref{four} we prove in the appendix that the uniform laws of large numbers and stochastic equicontinuity conditions that are usually required to prove asymptotic normality are satisfied. Moreover, we note here that these results are a direct consequence of the regularity that results from simulated moments and binding functions calculated using our GII-COV approach.

\vspace{0.3cm}
\begin{theorem}\label{norm}
  Let $R = \floor{cn^{\delta}}$ for some constant $c>0$ and $\delta \in [0, \infty)$.
  Suppose that Assumptions \ref{one}-\ref{four} are satisfied.
  Also,
  for the Wald estimator,
  assume the condition stated 
  in 
  Theorem \ref{theorem:Q}(b).
  \begin{itemize} 
  \item[(i)] If
    $\delta=0$ or $R=\floor{c}\geq1$, then,  as $n\rightarrow\infty$,
    \begin{eqnarray*}
      \sqrt{n}(\hat{\theta}^{\text{\em LM}}-\theta^0)
      \dto
      N 
      \Big (0,
      (1+\floor{c}^{-1})
      \Sigma^{\text{\em LM}}
      \Big )
      \ \ \mathrm{and} \ \
      \sqrt{n}(\hat{\theta}^{\text{\em W}}-\theta^0)
      \dto
      N 
      \Big (0,
      (1+\floor{c}^{-1})
      \Sigma^{\text{\em W}}
      \Big ),      
    \end{eqnarray*}
    where
    $\Sigma^{\text{\em LM}}:=
    (\Delta'\Omega \Delta)^{-1}
    (\Delta'\Omega \Xi \Omega \Delta)
    (\Delta'\Omega \Delta)^{-1}
    $
    and 
    $\Sigma^{\text{\em W}}:=
    (\Gamma'\Omega \Gamma)^{-1}
    (\Gamma'\Omega \Lambda^{-1'} \Xi \Lambda^{-1} \Omega \Gamma)
     (\Gamma'\Omega \Gamma)^{-1}
    $
    with
    $\Delta:= \partial_{\theta}M(\theta^{0},\beta^0)$,
    $\Lambda:= \partial_{\beta}M(\theta^{0},\beta^0)$
    and
    $\Gamma:= \partial_{\theta}\beta(\theta^{0})$,
    and
  \item[(ii)]
    if $\delta \in (0, \infty)$, then, as $n\rightarrow\infty$,
    \begin{eqnarray*}
      \sqrt{n}(\hat{\theta}^{\text{\em LM}}-\theta^0)
      \dto
      N
      \big (0,
      \Sigma^{\text{\em LM}}
      \big )
      \ \ \ \mathrm{and} \ \ \ 
      \sqrt{n}(\hat{\theta}^{\text{\em W}}-\theta^0)
      \dto
      N
      \big (0,
      \Sigma^{\text{\em W}}
      \big ).      
    \end{eqnarray*}
  \end{itemize}
  For the LM estimator, if the weighting matrix $\Omega_{n}$
  converges to $\Xi^{-1}$ in probability,
  then $\hat{\theta}^{\text{\em{LM}}}$ has the ``optimal'' asymptotic variance,
  which is proportional to 
  $(\Delta' \Xi^{-1} \Delta)^{-1}$. For the Wald estimator, if the weighting matrix $\Omega_{n}$
  converges to $(\Lambda'\Xi^{-1}\Lambda)$ in probability,
  then $\hat{\theta}^{\text{\em{W}}}$ has the ``optimal'' asymptotic variance,
  which is proportional to 
  $(\Delta' \Xi^{-1} \Delta)^{-1}$
\end{theorem}
\vspace{0.5cm}

\begin{remark}
{\normalfont	The above results demonstrate that the GII-COV approach yields estimators that have the same asymptotic properties, at first-order, as a standard II estimator. That is, the COV does not have an asymptotic impact on the resulting parameter estimates of $\theta^0$. {In addition, we note that if the original criterion is differentiable in $\theta$, this COV approach will lead to estimators that are numerically equivalent to standard II estimators: this follows by noting that if the function $y_i^r(\theta)$ is differentiable at $\theta^*$, with derivative $\partial_\theta y_i^r(\theta^*)$, our approach ensures that $\partial_\theta y_i^r(\theta,\theta^*)|_{\theta=\theta^*}=\partial_\theta y_i^r(\theta^*)$. {Hence, in cases where the original simulated outcomes are continuously differentiable, there is no benefit from using this GII-COV estimation strategy and the two approaches would be asymptotically equivalent not only at first-order but also at higher-order, assuming the required regularity conditions were satisfied.}}}

\end{remark}

\begin{remark}{\normalfont
The GII-K approach of \citet{bruins2015generalized} and our GII-COV approach are first-order asymptotically equivalent, so long as  the bandwidth parameter used in GII-K approach converges to zero faster than $1/\sqrt{n}$, so as to not impart asymptotic bias on the resulting parameter estimates. In contrast, our approach does not require any tuning parameters. This difference between GII-K and GII-COV could have ramifications for higher-order properties of the two estimators. While a higher-order analysis of the two estimators is beyond the scope of this paper, we conjecture that the GII-COV approach is likely to lead to estimators with smaller high-order bias, and thus smaller finite-sample bias, than those obtained from GII-K, {owing to the fact that no kernel smoothing (and no tuning parameter) is needed for GII-COV to construct derivatives of the criterion function in a high-order expansion.} This conjecture is further substantiated by Monte Carlo evidence given in the following section, which demonstrates that GII-COV generally leads to estimators with smaller finite-sample bias than GII-K.}	
\end{remark}

\begin{remark}
	{\normalfont Our GII-COV approach allows for simple consistent estimators of the asymptotic variances in Theorem \ref{norm} using sample analogs. This can be accomplished {using several different numerical methods to estimate the derivatives of the II criterion function. Firstly, since the sample moment function $m_{i}^{r}(\theta,\theta^*,\beta)$ is differentiable in $\theta$, a common approach would be to use standard numerical differentiation, which, for example, could be used to obtain some $\hat{\Delta}$ that estimates the Jacobian of the moments $\Delta$:} for $e_{j}$ denoting a $d_\theta\times1$ vector with $1$ in the $j$-th component, and zero else,  $\Delta_{_{j}}=\partial_{\theta_{j}}\mathbb{E}[m_{i}(\theta,\beta^0)]|_{\theta=\theta^0}$ can be estimated using, for $\delta_{n}>0$ and small, $$\hat{\Delta}_{{j}}(\delta_n)=\frac{1}{2\delta_n}\left[\frac{1}{nR}\sum_{i=1}^{n}\sum_{r=1}^{R}m_{i}^{r}(\hat{\theta}^{\text{LM}}+e_j\delta_n,\hat{\theta}^{\text{LM}},\hat{\beta})-\frac{1}{nR}\sum_{i=1}^{n}\sum_{r=1}^{R}m_{i}^{r}(\hat{\theta}^{\text{LM}}-e_j\delta_n,\hat{\theta}^{\text{LM}},\hat{\beta})\right]. $$ The results of Proposition \ref{mc1} imply that so long as $\hat{\theta}^{\text{LM}}\pto \theta^0$, $\hat{\beta}\pto \beta^0$, and $\delta_n\rightarrow0$, $\hat{\Delta}_{{j}}(\delta_n)\pto\Delta_j$ (see, e.g., \citealp{Hong2015}). Secondly, the derivative $\Delta_{_{j}}$ could also be consistently estimated {using the numerical technique of automatic differentiation, see, e.g., \cite{glasserman2003monte}. Since automatic differentiation is not commonly applied in econometrics, we further discuss this numerical technique in the supplemental appendix. }}
\end{remark}


\section{Illustrative Example: Discrete Choice Models}\label{prob_sec}

To further illustrate the GII-COV approach, we now consider the application of II to simulated data generated from various dynamic discrete choice models.
Our choice of the auxiliary model is the linear probability model, which has similarities with the auxiliary models in \cite{Li2010} and \cite{bruins2015generalized}, and allows for simple closed-form moments that we can use to estimate the underlying structural parameters. 
\subsection{Indirect Inference in Discrete Choice Models}

Individual $i=1,...,n$ chooses from among $(J+1)$-alternatives at each time $t=1,...,T$ by maximizing over the utilities associated with each of the alternatives. We assume the utilities follow the standard additive random utility framework. In particular, the $i$-th individual at time $t$ chooses the $j$-th alternative if 
\begin{equation}
\label{mnp}
\begin{aligned}
  y_{j,t} &= \1\left[y^{*}_{j,t} > \max\{0, y^{*}_{k,t}: k =1,\ldots, J \text{  and  } k \ne j\}\right],\text{ for }j=0,1,...,J,\\
  y^{*}_{j,t} & =\rho_{y} y_{j,t-1}+   w'_{j,t}\alpha + x_{it}'\gamma_j + v_{j,t},\\
  v_{j,t}&=\rho_{e}v_{j,t-1}+\epsilon_{j,t},
\end{aligned}
\end{equation}
where 
$(\epsilon_{1,t},\ldots, \epsilon_{J,t})'
= \Psi^{1/2} (\eta_{1,t},\ldots,\eta_{j,t})'$
with $\Psi^{1/2}$
lower triangular such that $\Psi^{1/2}\Psi^{1/2'} = \Psi$
and
with
$(\eta_{1,t},\ldots,\eta_{J,t})' \sim  N(0, I_J)$ for all $t=1,...,T$,
and 
$(\eta_{1,t},\ldots,\eta_{J,t})'$
is independent of
the observables 
$(w'_{1,t},\ldots, w'_{J,t})'$, i.e., say the alternative-dependent variables, and $x_{it}$, i.e., say the purely individual specific regressors. The structural parameters are $\theta = \left(\alpha', \gamma'_1,\ldots, \gamma'_J, \rho_{y}, \rho_{e},\omega'\right)'$, where $\omega$ are the unique unrestricted elements of $\Psi$.

We specialize the above model to consider three distinct dynamic discrete choice models. For each of the three models individual $i = 1,\ldots,n$ chooses one of two available alternatives ($J=1$ for simplicity) by maximizing the standard additive random utility. In particular, for $i = 1,\ldots, n$ and $t = 1,\ldots,T$, the individuals choice follows $y_{it} = \1[y^{*}_{it} \ge 0]$, where $y^{*}_{it}$ is the net utility from the choice of alternative $j=1$ (over alternative $j=0$). 
The particulars of each model are as follows:

\medskip

\noindent\textbf{Model 1:} $T = 5$ and $y^{*}_{it} =  x_{it}'\gamma + v_{it}$, where $v_{it} = \rho v_{i,t-1} + \epsilon_{it}$ with $v_{i,0}=0$ and $\epsilon_{it}\sim N(0,1)$. The structural parameters are $\theta = (\gamma, \rho)'$.\medskip

\noindent\textbf{Model 2:} A slight modification of Model 1 that includes a lagged dependent variable in the unobserved utility function as follows. $T = 5$ and $y^{*}_{it} = \alpha y_{i,t-1} + x_{it}'\gamma + v_{it} $, where $v_{it} = \rho v_{i,t-1} + \epsilon_{it}$, with $y_{i,0} = 0$, $v_{i,0}=0$ and $\epsilon_{it}\sim N(0,1)$. The structural parameters are $\theta = (\alpha, \gamma, \rho)'$.\medskip

\noindent\textbf{Model 3:} Similar to Model 2 but incorporates the well-known ``initial-conditions'' problem as follows. Model 2 holds with $T = 5$ but for $i=1,\ldots,n$, the econometrician observes the choices $y_{it}$ only for $t = 3,4,5$. The structural parameters remain $\theta = (\alpha, \gamma, \rho)'$.\medskip

For each model we consider application of the GII-COV approach using the LM criterion. For the choice of auxiliary moment vector $m(\cdot)$ in each model, we follow \cite{CFR2017} and take $m(\cdot)$ as follows:
\begin{itemize}
\item For Models 1 and 2 and $i=1,\ldots,n$
\begin{equation*}
  m
  \big (
  y_{i1}^{r}(\theta),..., y_{i5}^{r}(\theta), z_{i1},..., z_{i5};\beta
  \big )
  =
  \left[
    \begin{array}{c}
      z_{i1}
      \big (y_{i1}^{r}(\theta) - z_{i1}' \beta_{1} \big)\\
      \vdots \\
      z_{i5}
      \big (y_{i5}^{r}(\theta) - z_{i5}' \beta_{5} \big)
    \end{array}
  \right],   
\end{equation*}
where
$\beta = (\beta_{1}, \dots, \beta_{5})'$,
$z_{i1}=(1,x_{i1}')'$,
and $z_{it}=(1, x_{it}', x_{i,t-1}', y_{i,t-1})'$
for $t=2,...,5$.
\item For Model 3, we consider exactly the same $m(.)$ function as defined above but only for $t=3,4,5$.

\end{itemize}

The above choice for $m(\cdot)$ leads to equation-by-equation ordinary least squares computations in a seemingly unrelated regression (SUR) model, where there are $J$ response variables ($\ell=\1[y_{it}^{r}(\theta) = j]$ for $j = 1,\ldots,J$) and with the \textit{same set} of regressors $z$ used for all regressions. In particular, $m(\cdot)$ represents {the vector function specifying the} first-order conditions for the SUR model regression coefficients. 

It is clear that for each of the above auxiliary estimating equations, $m(\cdot)$ is not differentiable in $y_{it}^{r}(\theta)$. However, one can easily apply the derivative-based methods proposed herein to obtain quick and simple estimators. 

For any given time point $t$, Model 1 has a single discontinuity since $y_{it}=\1\left[y^{*}_{it}>0\right]$. Therefore, given our choice of moment function $m(\cdot)$, this example can be treated using the exact same approach outlined in Example 1. 
 Models 2 and 3 are similar to Model 1, but for any time point $t>1$, the function $m(\cdot)$ can suffer two different types of discontinuities. The first discontinuity results from $y_{it}=\1 [y^{*}_{it}>0]$, while the second discontinuity is due to the autoregressive nature of $y^{*}_{it}$: for $t>1$
 \begin{flalign*}
 y^{*}_{it} &=\alpha \1 [y^{*}_{i,t-1}>0 ]+x_{it}'\gamma +v_{it}.
 \end{flalign*}
This additional discontinuity means we require two measures changes to institute our method. However, both measure changes have the same form and so can easily be updated. The exact change or variables used in each of the above models is the same as in Example 1. Therefore, the corresponding Jacobian term is also the same. 

Let $\theta^*$ be a point we wish to evaluate the simulated outcomes. Then, we can use the same simulated uniforms $u_{it}^{r}(\theta,\theta^*)$ as in Example 1 within each of the above models to simulate the error terms $\epsilon_{it}$ according to
$
\Phi^{-1}\big (u_{it}^{r}(\theta,\theta^*) \big ).
$
For Model 1, the critical point functions are given by 
 \[
 c_{it}^{1}(\theta)=\Phi (-x_{it}'\gamma -\rho v_{i, t-1}),
 \]
while for Model 2 and 3, the critical point functions are given by 
\[
c_{it}^{1}(\theta)=\Phi\big(-\alpha \1\left[c_{i,t-1}^{1}(\theta^*)<u_{i,t-1}\right]-x_{it}'\gamma -\rho v_{i, t-1}\big),
\]
with $c_{it}^{0}(\theta) = 0 $ and $c_{it}^{2}(\theta) = 1$ in Model 1-3.
The following subsection considers a series of simulation examples within each of the three models considered above.  

\subsection{Simulation Results}
For each of the above models we generate 1,000 replications of simulated data according to the following set of true parameter values and sample size combinations: under each simulation, $x_i$ is generated iid as $\mathcal{N}(1,2)$ and we consider the following values for the model parameters and sample sizes
\begin{itemize}
\item For Model 1: $\theta^0=(\gamma^0,\rho^0)'=(1,.4)'$, $n\in\{200,1000\}$ and $T=5$;
\item For Model 2: $\theta^0=(\gamma^0,\alpha^0,\rho^0)'=(1,.2,.4)'$, $n\in\{200,1000\}$ and $T=5$;
\item For Model 3: $\theta^0=(\gamma^0,\alpha^0,\rho^0)'=(1,.2,.4)'$, $n\in\{200,1000\}$ and $T=5$ with $s=3$ unobserved periods;
\end{itemize} 
Across the different Monte Carlo designs, we consider four estimation procedures:
the GII-COV procedure, the GII-K estimation procedure of \cite{bruins2015generalized}, the  Nelder-Mead simplex-based search algorithm and the ``patternsearch'' genetic algorithm. All procedures are implemented in Matlab. The Nelder-Mead and genetic algorithms are both implemented using the Matlab default settings.\footnote{The precise implementation details for these methods can be found in their corresponding Matlab help files.}   

We consider two separate implementations of GII-K. In the first GII-K implementation, dubbed GII-1, we use a normal kernel to smooth the outcomes and consider a sample size dependent bandwidth $\lambda_n$: for $n=200$ we employ a bandwidth of $\lambda_{n}=.08$ and for $n=1000$ we consider $\lambda_{n}=.04$; these values correspond to a choice of $\lambda_{n}$ satisfying $\sqrt{n}\lambda_{n}=o(1)$, which is required for the GII-K approach to deliver consistent and asymptotically normal estimators.

The second GII-K implementation, dubbed GII-2, again uses the normal kernel to smooth the outcomes, but follows the two-step approach outlined in \cite{bruins2015generalized}. This two-step approach first implements GII-K with a large value of $\lambda_n$ and a small number of simulated outcomes, $R$, to obtain a preliminary estimator of $\theta^0$; in a second step, the GII-K procedure is re-run with a small value of $\lambda_n$ and a large value of $R$, and using the first-step estimator of $\theta^0$ as starting values for the second-step. We follow \cite{bruins2015generalized} and consider $(\lambda_n,R)=(.03,10)$ in the first-stage, and we employ $(\lambda_n,R)=(.003,300)$ in the second-stage. 

Across the 1,000 replications we report the mean bias (MBIAS), mean absolute bias (AB), standard deviation (STD) and the Monte Carlo coverage of a 95\% Wald-confidence interval (CV95) for all estimators. Except for the GII-2 approach, all other estimators use $R = 10$ simulated draws across all simulation designs. 

The GII-COV estimators are computed using the Newton-Raphson algorithm with both the Hessian and the gradient {estimated numerically using automatic differentiation techniques.\footnote{In the supplemental appendix, we give further compare between two versions of GII-COV: one version that uses automatic differentiation to numerically calculate the gradient and Hessian, and a separate version that uses central finite-differencing to estimate these derivatives. The results suggest that using automatic differentiation to numerically estimate the derivatives yields an finite-sample improvement, at least in terms of bias, over using standard finite-differencing derivative estimators.}} {For all GII-COV procedures we take $R=10$ across all simulation designs.} The smoothness of the GII-K criterion function also allows for derivative-based optimisation procedures. Following \cite{bruins2015generalized}, we implement both of the GII-K approaches using a quasi-Newton algorithm that calculates the Hessian and gradient of the smoothed criterion function using finite-differencing methods. As in \cite{bruins2015generalized}, the initial value across each of the simulation runs is set at the true parameter value for all estimation procedures and across each of the simulation design. For each estimator,  the efficient weighting matrix is used.

The results across the different simulation designs and estimation methods are collected in Tables \ref{tab1}-\ref{tab5}. In terms of bias (MBIAS) and standard deviation (STD), and across each Monte Carlo design, the GII-COV estimator gives superior performance relative to GII-K and the derivative-free estimators. The additional bias and variance inherent in GII-K, over and above that observed in GII-COV, reflects the procedures use of smoothed outcomes, which require the choice of kernel and an associated bandwidth parameter. However, we do note that at larger samples sizes, the GII-2 approach gives results that are closer to GII-COV.

It is also important to note that the differences between GII-COV and GII-K do not abate as $n$ increases. As can be seen from Table \ref{tab7}, relative to GII-COV, the biases and standard errors of GII-1 increase as $n$ increases across nearly all Monte Carlo designs. A similar pattern of results is also observed for GII-2 in Table \ref{tab8}. However, at larger sample sizes, GII-2 gives estimates with smaller bias and standard deviation than those obtained using the GII-1 approach.

The additional optimisation step required of GII-2 will lead to a much slower algorithm than GII-COV, GII-1 and the derivative-free methods. Table \ref{tab6} contains the average execution times of the various algorithms across the Monte Carlo designs. In all cases, GII-1 gives the fastest execution times, and is followed closely by GII-COV. In general, GII-COV is faster to implement than either of the derivative-free approaches. The execution time of GII-2 is always much larger than all other methods used in this simulation study, which is a consequence of the fact that GII-2 must generate, and manipulate, $R$=300 additional simulated data sets in the second stage of estimation. Therefore, while GII-2 can give better estimates than GII-1, and thus yield results that are closer to GII-COV in terms of MBIAS and STD, it does so at much greater computational cost. 

From these Monte Carlo results, we can conclude that GII-COV performs well relative to these competitors in terms of precision and computational properties.

\section{Discussion}
There are many interesting examples in II where the resulting criterion function is discontinuous in the parameters of interest. While global kernel smoothing approaches, such as the generalized indirect inference approach of \cite{bruins2015generalized}, have been proposed to alleviate this discontinuity in the criterion function, such methods require a user dependent bandwidth parameter that can negatively impact the resulting parameter estimates.

In this article, we have proposed a novel approach to II that can alleviate discontinuities within simulated criterion functions without the need of smoothing approaches. Applying this simulation approach within II results in criterion functions that yield uniformly consistent derivative estimates of their corresponding limiting counterparts, and  allows II estimators to be calculated using standard derivative-based optimisation routines. Furthermore, the resulting II estimators have standard asymptotic properties and perform well in finite-sample simulation experiments.

For simplicity, in this paper we have focused on the case where the simulated data is generated from independent uniforms, which consequently allows us to apply a conditionally independent change of variables to remove the discontinuities in the criterion function. In cases where dependent uniforms are required to generate simulated outcomes, such as, e.g., uniforms underlying certain multivariate random variables with a copula structure, a change of variables in each of the uniforms generating the endogenous variables will have flow-on effects to other dimensions. In these cases a similar approach to that considered here can be used. However, this extension requires additional technical details and explanation, and is therefore left for further research.


\newpage
\setstretch{0.15}            
\setlength{\bibsep}{6pt}
\bibliographystyle{ecta} 
\bibliography{Chainging_II}

\newpage
\appendix
\section*{Appendix A. Proofs of Main Results}
\setstretch{1.0}            

\setcounter{section}{0} 
\setcounter{equation}{0} 
\setcounter{lemma}{0}\setcounter{page}{1}\setcounter{proposition}{0} %
\renewcommand{\thepage}{A-\arabic{page}} \renewcommand{\theequation}{A.%
\arabic{equation}}\renewcommand{\thelemma}{A.\arabic{lemma}} 
\renewcommand{\theproposition}{A.\arabic{proposition}} 
  
In this appendix,
we prove the asymptotic properties of our proposed estimator
and some related technical results, 
applying similar arguments used in \cite{pakes1989simulation}
\citep[see also][]{chen2003estimation}.
We write 
$a \lesssim b$
if $a$ is smaller than or equal to $b$ up to a universal positive constant.

\vspace{0.2cm}
\begin{lemma}
  \label{lemma:exist}
  Suppose that 
  Assumptions \ref{one}-\ref{g} hold.
  Then,
  \begin{enumerate}

  \item [(a)] 
    the function 
    $M(\theta, \beta)$ is twice differentiable
    with respect to $\theta$ for any $\beta \in \mathcal{B}$ 
    and its derivative is given by,
    for $\ell \in \{1,2\}$,
    \begin{eqnarray}
      \label{eq:moment-dev}
      \partial_{\theta}^{\ell} M(\theta,\beta) 
      =
      \E
      \bigg [
      \sum_{j=0}^{J}
      m
      (
      \alpha_{j}, z_{i}, \beta
      )
      \partial_{\theta}^{\ell} 
      \{
      c^{j+1}_{i}(\theta)-c^{j}_{i}(\theta)
      \}
      \bigg ]; 
    \end{eqnarray}
  \item [(b)] 
    the moment function 
    ${m}_{i}^{r}(\theta,\theta^*,\beta)$
    is twice-continuously differentiable with respect to $\theta \in \Theta$,
    for any $(\theta^*, \beta) \in \Theta \times \mathcal{B}$.
  \end{enumerate}
\end{lemma}
\begin{proof}

  (a)
  Under Assumption \ref{g}(a), we can write
  $y_{i}^{r}(\theta) = \sum_{j=0}^{J}\alpha_{j}
  \1[c_{i}^{j}(\theta) < u_{i}^{r} \le c_{i}^{j+1}(\theta) ]
  $.
  We can show that
  $
  M(\theta,\beta)
  =
  \E
  [
    \sum_{j=0}^{J}
    m
    (
    \alpha_{j}, z_{i}, \beta
    )
    \{
    c^{j+1}_{i}(\theta)-c^{j}_{i}(\theta)
    \}
  ] 
  $,
  since $u_{i}^{r}$ is uniformly distributed over $(0,1)$.
  From Assumptions 2(c) and \ref{g}(b),  the derivative of  the map 
  $\theta \mapsto 
  m
  (
  \alpha_{j}, z_{i}, \beta
  )
  \{
  c^{j+1}_{i}(\theta)-c^{j}_{i}(\theta)
  \}
  $
  satisfies 
  \begin{eqnarray*}
    \|
    m
    (
    \alpha_{j}, z_{i}, \beta
    )
    \partial_{\theta}
    \{
    c^{j+1}_{i}(\theta)-c^{j}_{i}(\theta)
    \}
    \|
    \lesssim
    \bar{m}_{i}\cdot
    \nabla^{1} \bar{c}_{i}. 
  \end{eqnarray*}
  Under Assumptions 2(c) and 3(b),
  an application of Cauchy-Schwarz inequality yields 
  $ 
  E[\bar{m}_{i}\cdot  \nabla^{1} \bar{c}_{i}  ] < \infty
  $.  
  Thus, the dominated convergence theorem implies the desired conclusion
  \citep[see Theorem 2.27 of][for example]{folland2013real}.
  A similar argument shows the second derivative part 
  and thus we omit the details.

  (b)
  Let $(\theta^*,\beta) \in \Theta \times \mathcal{B}$ be fixed.
  It follows from Assumption 3(a) and the definition of the Jacobian $w_{i}^{r}(\cdot)$
  that 
  if 
  $c_{i}^{j}(\theta^*)<u^r_{i}\leq c_{i}^{j+1}(\theta^*)$
  for $j = 0, \dots, J$, then
  \begin{eqnarray*}
    y_{i}^{r}(\theta,\theta^*)
    =
    \alpha_{j}
    \hspace{0.7cm}
    \mathrm{and}
    \hspace{0.7cm}
    w_{i}^{r}(\theta,\theta^*)
    =
    \frac{
    c^{j+1}_{i}(\theta)-c^{j}_{i}(\theta)
    }{
    c^{j+1}_{i}(\theta^*)-c^{j}_{i}(\theta^*)
    },
  \end{eqnarray*}
  for all $i=1, \dots, n$.
  Also,
  it can be easily verified that
  $
  c_{i}^{j}(\theta^*)<u_{i}^{r} \leq c_{i}^{j+1}(\theta^*) 
  $
  if and only if  
  $
  c_{i}^{j}(\theta)<u_{i}^{r}(\theta,\theta^*)\leq c_{i}^{j+1}(\theta)
  $.
  Thus,
  the approximate moment function is written as 
  \begin{eqnarray}
    \label{eq:moment1}
    m^{r}_{i}(\theta,\theta^*,\beta)
    =
    \sum_{j=0}^{J}
    m
    (
    \alpha_{j}, z_{i}, \beta
    )
    \frac{
    c^{j+1}_{i}(\theta)-c^{j}_{i}(\theta)
    }{
    c^{j+1}_{i}(\theta^*)-c^{j}_{i}(\theta^*)  
    }
    \1
    \big [c_{i}^{j}(\theta^*)<u^{r}_{i}\leq c_{i}^{j+1}(\theta^*)  \big].
  \end{eqnarray}
  In the above moment function,  
  the parameter $\theta$ now
  appears as arguments of differentiable functions
  and no longer determines the value of discontinuous functions.
  More precisely,
  the right-hand side term of the equation above consists of
  twice-continuously differentiable functions with respect to $\theta$
  and thus
  the desired conclusion follows.
\end{proof}
\vspace{0.2cm}

To prove Proposition \ref{mc1}, 
we apply the idea used by \cite{jennrich1969}
to prove 
the uniform law of large numbers.

\vspace{0.2cm}
\begin{proof}[\textbf{Proof of Proposition \ref{mc1}}]
  First, we consider the unbiasedness of the first derivative of the moment function 
  $M_{n}(\theta, \theta^{\ast}, \beta)$.
  Let $\theta^{\ast} \in \Theta$ be fixed. 
  It follows from (\ref{eq:moment1}) that
  \begin{eqnarray*}
    \partial_{\theta}
    m^r_{i}(\theta,\theta^*,\beta)
    =
    \sum_{j=0}^{J}
    m
    (
    \alpha_{j}, z_{i}, \beta
    )
    \frac{
    \partial_{\theta}
    \{
    c^{j+1}_{i}(\theta)
    -
    c^{j}_{i}(\theta)
    \}
    }{
    c^{j+1}_{i}(\theta^*)-c^{j}_{i}(\theta^*)
    }
    \1
    \big [c_{i}^{j}(\theta^*)<u^r_{i}\leq c_{i}^{j+1}(\theta^*)  \big].
  \end{eqnarray*}
  Also, we have, for every $j=0, \dots, J$,
  \begin{eqnarray}
    \label{eq:ee1}
    \frac{
    1
    }{
    c_{i}^{j+1}(\theta^*)
    -
    c_{i}^{j}(\theta^*)
    }
    \int_{0}^{1}
    \1
    \big [c_{i}^{j}(\theta^*)<u \leq c_{i}^{j+1}(\theta^*)  \big]
    du 
    = 1.
  \end{eqnarray}
  Because uniform random variables ${u_{i}^{r}}$
  are independent from the rest of the model, and under the iid data assumptions in Assumptions \ref{one}(a) and \ref{two}(a), 
  we have 
  \begin{eqnarray}
    \label{eq:AA}
    \E
    \big [
    \partial_{\theta}
    M_{n}(\theta,\theta^*,\beta)
    |_{\theta = \theta^{\ast}}
    \big ]
    =
    \E
    \bigg [
    \sum_{j=0}^{J}
    m
    (
    \alpha_{j}, z_{i}, \beta
    )
    \partial_{\theta}
    \{
    c^{j+1}_{i}(\theta^{\ast})-c^{j}_{i}(\theta^{\ast})
    \}
    \bigg ].
  \end{eqnarray}
  It follows from Lemma \ref{lemma:exist}(a) that 
  $
  \partial_{\theta}
  M_{n}(\theta,\theta^*,\beta)
  |_{\theta = \theta^{\ast}}
  $
  is an unbiased estimator of
  $\partial_{\theta}
  M(\theta^*,\beta)
  $.

  We now consider the uniform consistency of
  the first derivative.
  Let $\delta>0$ an arbitrarily small scalar.
  Given the compact parameter spaces $\Theta$ and $\mathcal{B}$,
  it suffices to show that
  the desired conclusion holds
  for every fixed neighborhood 
  $\mathcal{N}_{\delta}:=\mathcal{N}_{1,\delta} {\times} \mathcal{N}_{1,\delta}{\times} \mathcal{N}_{2,\delta}$,
  where 
  $\mathcal{N}_{1, \delta}\subset \Theta$
  and 
  $
  \mathcal{N}_{2, \delta}  
  \subset \mathcal{B}$
  satisfy
  $
  \sup_{\theta_{1}, \theta_{2} \in \mathcal{N}_{1, \delta}}
  \|\theta_{2} - \theta_{1}\| \le \delta$,
  and 
  $
  \sup_{\beta_{1}, \beta_{2}\in \mathcal{N}_{2,\delta}}
  \|\beta_{2} - \beta_{1}\| \le \delta$.
  Define 
  \begin{eqnarray*}
    \pi_{n,\delta}^{-}
    :=
    \frac{1}{nR}
    \sum_{i=1}^{n}
    \sum_{r=1}^{R}
    \inf_{(\theta, \theta^{\ast}, \beta) \in \mathcal{N}_{\delta}}
    \partial_{\theta}
    m_{i}^{r}(\theta, \theta^{\ast}, \beta)
    \ \ \mathrm{and}  \ \ 
    \pi_{n,\delta}^{+}
    :=
    \frac{1}{nR}
    \sum_{i=1}^{n}
    \sum_{r=1}^{R}
    \sup_{(\theta, \theta^{\ast}, \beta) \in \mathcal{N}_{\delta}}
    \partial_{\theta}
    m_{i}^{r}(\theta, \theta^{\ast}, \beta).
  \end{eqnarray*}
  It can be easily shown that,
  from the triangle inequality,
  \begin{eqnarray*}
    \sup_{(\theta, \theta^{\ast}, \beta) \in \mathcal{N}_{\delta}}    
    \big \|
    \partial_{\theta} 
    M_{n}(\theta, \theta^{\ast}, \beta)
    -
    \partial_{\theta} 
    M(\theta, \theta^{\ast}, \beta)
    \big \|
    &\le&
          \big \|
          \pi_{n,\delta}^{+} - \E[    \pi_{n,\delta}^{+}]
          \big \|
          +
          \big \|
          \pi_{n,\delta}^{-} - \E[    \pi_{n,\delta}^{-}]
          \big \|\\
    &&  
       +
       \big \|
       \E[\pi_{n,\delta}^{+}] - \E[    \pi_{n,\delta}^{-}]
       \big \|,
  \end{eqnarray*}
  where
  $
  M(\theta, \theta^{\ast}, \beta)
  :=
  \E[
  m_{i}^r(\theta, \theta^{\ast},\beta)
  ]
  $.
  It follows from the weak law of large numbers that
  $  \big \|
  \pi_{n,\delta}^{\pm} - \E[    \pi_{n,\delta}^{\pm}]
  \big \|
  = o_p(1)
  $
  as $n \to \infty$.
  Thus, the remaining task is to show that
  $
  \big \|
  \E[\pi_{n,\delta}^{+}] - \E[    \pi_{n,\delta}^{-}]
  \big \| \to 0
  $
  as $\delta \to 0$.
  To this end, 
  using
  (\ref{eq:moment1})
  together with 
  the fact that
  $u_{i}^{r}$ is uniformly distributed
  and independent from the original data,  
  we have
  that
  $
  \big \|
  \E[\pi_{n, \delta}^{+}]
  -
  \E[\pi_{n, \delta}^{-}]
  \big \|
  \le
  \big \|
  \E
  \big [    
  \sum_{j=0}^{J}
  \Delta_{t, \delta}^{j}
  \big ]
  \big \|,
  $
  where 
  \begin{eqnarray*}
    \Delta_{t, \delta}^{j}
    &:=&
         \sup_{(\theta, \theta^{\ast}, \beta) \in \mathcal{N}_{\delta}}    
         m
         (
         \alpha_{j}, z_{i}, \beta
         )
         \partial_{\theta}
         w_{i}^{r}(\theta, \theta^{\ast})
         \Big (
         \sup_{\theta^{\ast}\in \mathcal{N}_{1, \delta}}    
         c_{i}^{j+1}(\theta^*)  
         -
         \inf_{\theta^{\ast}\in \mathcal{N}_{1, \delta}}    
         c_{i}^{j}(\theta^*)
         \Big ) \\
    &&-
       \inf_{(\theta, \theta^{\ast}, \beta) \in \mathcal{N}_{\delta}}    
       m
       (
       \alpha_{j}, z_{i}, \beta
       )
       \partial_{\theta}
       w_{i}^{r}(\theta, \theta^{\ast})
       \Big (
       \inf_{\theta^{\ast}\in \mathcal{N}_{1, \delta}}    
       c_{i}^{j+1}(\theta^*)  
       -
       \sup_{\theta^{\ast}\in \mathcal{N}_{1, \delta}}    
       c_{i}^{j}(\theta^*)
       \Big ) .
  \end{eqnarray*}
  The above equation consists of functions that are continuous in parameters
  and all functions are evaluated over $\mathcal{N}_{\delta}$.
  Thus,
  $
  \big  \|
  \sum_{j=0}^{J}
  \Delta_{t,\delta}^{j}
  \big \| \to 0$ as $\delta \to 0$.
  Also,
  we have 
  \begin{eqnarray*}
    \| \Delta_{t, \delta}^{j}\|
    \le
    2
    \bar{m}_{}(\alpha_{j}, z_{i})\nabla \bar{w}_{i}
    \cdot
    \Big (
    \sup_{\theta^{\ast}\in \mathcal{N}_{1, \delta}}    
    c_{i}^{j+1}(\theta^*)  
    -
    \inf_{\theta^{\ast}\in \mathcal{N}_{1, \delta}}    
    c_{i}^{j}(\theta^*)
    \Big ).
  \end{eqnarray*}
  The critical point functions take values in the unit interval
  and
  also
  by
  the Cauchy-Schwarz inequality,
  $
  \E[
  \bar{m}_{i}\nabla \bar{w}_{i}
  ]
  \le
  (
  \E |\bar{m}_{i}|^{2}
  )^{1/2}
  (
  \E|\nabla \bar{w}_{i}|^2
  )^{1/2}
  $,
  which is finite 
  by Assumption 2(c) and \ref{g}(b).
  Therefore, by the dominated convergence theorem,
  $\| \E[\Delta_{ir}^{(1)}(\delta)] \| \to 0$
  as $\delta \to 0$.
  Hence, the estimator for the first derivative
  is consistent uniformly in $(\theta^{\ast}, \beta) \in \Theta \times \mathcal{B}$.

  A similar argument shows the
  uniform consistency of the second derivative
  and thus we omit the details.
\end{proof}
\vspace{0.2cm}

Before proving asymptotic properties of the estimators $\hat{\theta}$ and $\hat{\beta}(\theta)$,
we present the existence of a global implicit solution $\beta(\theta)$
to $M(\theta, \beta) = 0$ and its properties.
We follow the approach in \cite{Idczak2016ANS}, who uses
the mountain path theorem \citep{ambrosetti1973JFA}
to establish a global implicit function theorem. 
The following theorem is a finite dimensional counterpart of the key result
(Theorem 3.3 and Corollary 3.3)
in \cite{Idczak2016ANS}.

\vspace{0.2cm}
\begin{lemma} 
  \label{lemma:MPT} 
  Let
  $\mathcal{V} \subset \R^{d_{v}}$
  and
  $\mathcal{W} \subset \R^{d_{w}}$
  with finite $d_{v}, d_{w} \in \mathbb{N}$
  and
  let
  $F:\mathcal{V} {\times} \mathcal{W} \to \R^{d_{w}}$ be a continuously differentiable function. 
  For all $v \in \mathcal{V}$,
  assume that
  the map $w \mapsto F(v, w)$ 
  satisfies
  \begin{itemize}
  \item [(a)]
    the Palais-Smale condition:
    any sequence $\{w_{k}\}_{k=1}^{\infty} {\subset} \mathcal{W}$
    with $ \{\|F(v, w_{k}) \|^{2}/2\}_{k=1}^{\infty}$ being bounded
    and $\partial_{v}F(v, w_{k}) \to 0$
    as $k \to \infty$
    admits a convergence subsequence,
  \item [(b)]
    a square matrix 
    $\partial_{w}F(v, w)$ is
    non-singular for any $w \in \mathcal{W}$.
  \end{itemize}
  Then,
  for any $v \in \mathcal{V}$,
  there exists a unique $w_{v} \in \mathcal{W}$
  such that $F(v, w_{v}) = 0$. 
\end{lemma}
\begin{proof}
  See the proof of Theorem 3.3 in \cite{Idczak2016ANS}.  
\end{proof}
\vspace{0.2cm}

The following lemma uses Lemma \ref{lemma:MPT} to present a global implicit function theorem,
which shows that the implicit function exists uniquely and is globally twice continuously differentiable.

\vspace{0.2cm}
\begin{lemma}
  \label{lemma:exist-beta-der}
  Suppose 
  that
  Assumptions 1-4 
  and 
  the conditions
  assumed in Theorem 1(b) hold.
  \begin{itemize}
  \item [(a)] 
    $M(\theta, \beta)$ is twice continuously differentiable with respect to
    $(\theta, \beta) \in \Theta \times \mathcal{B}$.
  \item [(b)]
    For any $\theta \in \Theta$,
    there exists a unique $\beta(\theta) \in \mathcal{B}$
    such that 
    $M\big(\theta, \beta(\theta) \big) =0$
    and 
    $\beta(\theta)$ is a twice continuously differentiable function with its first derivative, given by 
    \begin{eqnarray}
      \label{eq:beta-dev}
      \partial_{\theta}
      \beta(\theta)
      =
      - 
      \big [
      \partial_{\beta}M \big (\theta, \beta(\theta) \big)
      \big ]^{-1}
      \partial_{\theta} M \big (\theta, \beta (\theta) \big).
  \end{eqnarray}
  \end{itemize}
\end{lemma}
\begin{proof}
  (a) We can write 
  $
  M(\theta,\beta)
  =
  \E
  [
  \sum_{j=0}^{J}
  m
  (
  \alpha_{j}, z_{i}, \beta
  )
  \{
  c^{j+1}_{i}(\theta)-c^{j}_{i}(\theta)
  \}
  ] 
  $
  under Assumption 3(a). 
  In Lemma \ref{lemma:exist}, 
  we have shown that
  $M(\theta,\beta)$ is twice continuously differentiable with respect to $\theta \in \Theta$.
  Using similar arguments in Lemma \ref{lemma:exist}, we can prove that
  $M(\theta,\beta)$ is twice continuously differentiable with respect to $\beta \in \mathcal{B}$
  and
  $\partial_{\theta}M(\theta,\beta)$ is
  continuously differentiable with respect to $\beta \in \mathcal{B}$,   
  under the condition (ii) of Theorem 1(b) and Assumption \ref{g}(b).
  Thus, the desired conclusion follows.
   
  (b)
  First, the classical implicit function theorem implies that the desired result
  holds locally.  
  That is, in the neighborhood of $(\bar{\theta}, \bar{\beta}) \in \Theta \times \mathcal{B}$
  with $M(\bar{\theta}, \bar{\beta)} =0$,
  there is a continuously differentiable function $\theta \mapsto \beta(\theta)$ 
  such that 
  $M\big(\theta, \beta(\theta) \big) =0$ and its derivative takes the form 
  in (\ref{eq:beta-dev}).
  The result in (a) of this lemma implies that
  the right-hand side of (\ref{eq:beta-dev}) consists of differentiable functions
  with respect to $(\theta, \beta)$ and thus $\beta(\theta)$
  is twice continuously differentiable locally
  \citep[see Theorem 1B.1 and 1B.5 of ][for instance]{DR2009book}.
     
  Next, the parameter space $\mathcal{B}$ is assumed to be compact
  and $\| M(\theta, \beta)\|^2/2$ satisfies the Palais-Smale condition
  for each $\theta \in \Theta$.
  Also, $\partial_{\beta} M(\theta, \beta)$ is shown to be non-singular
  under the condition (iii) of Theorem 1(b).
  It follows from Lemma \ref{lemma:MPT} that
  $\beta(\theta)$ uniquely solves $M(\theta, \beta) =0$
  for all $\theta$
  and $\beta(\theta)$ satisfies the desired properties globally. 
\end{proof}
\vspace{0.2cm}

In the below lemma, we present a technical result, which
is useful when an estimator depends on observed and simulated data.
More precisely, this technical lemma establishes
the uniform law of large numbers over simulations
as long as
the simulation size $R$ is not large to be compared with
the sample size $n$, or $\log(R) = o(n)$.
We prove this lemma by using the symmetrization argument.
A similar approach can be found in Lemma 8 of \cite{chernozhukov2015comparison}
in a different context.

\vspace{0.3cm}
\begin{lemma} 
  \label{lemma:FN}
  Given positive integers $n$ and $R$,
  let
  $\bm{z}_{n}:=\{z_{i}\}_{i=1}^{n}$
  be
  a sequence of independent random
  variables in $\R$
  and 
  $\bm{u}_{n,R}:=\{(u_{i1}, \dots, u_{iR})\}_{i=1}^{n}$
  be
  a collection of vectors whose
  elements are 
  independent and identically distributed
  in $\R$.
  Assume that 
  $\bm{z}_{n}$ is independent of
  $\bm{u}_{n,R}$.
  For
  $i=1, \dots, n$
  and 
  $r = 1, \dots, R$,
  define
  $x_{ir}:=g(z_{i},u_{ir})$
  for some measurable function $g: \R^{2} \to \R$
  and 
  $\bar{x}_{r}:= n^{-1} \sum_{i=1}^{n}x_{ir}$.
  Assume that
  $\E[\max_{1 \le r \le R}|x_{ir}|^2]<\infty$
  for every $i=1, \dots, n$.
  If 
  $n\to\infty$
  with 
  $
    \log(R) /
    n
    \to 0 
  $, then
  \begin{eqnarray*}
    \max_{1 \le r \le R}
    \big |
    \bar{x}_{r}
    -
    \E [ \bar{x}_{r}]
    \big |
    =
    o_p(1).
  \end{eqnarray*}
\end{lemma}
\begin{proof}
  By Markov's inequality, 
  $
    \Pr
    (
    \max_{1 \le r \le R}
    |
    \bar{x}_{r} - \E^{\ast}[\bar{x}_{r}]
    |
    \ge
    \eta
    )
    \le
    \eta^{-1}
    \E
    ( 
    \max_{1 \le r \le R}
    |
    \bar{x}_{r} - \E[\bar{x}_{r}]
    |
    )  
  $ 
  for any $\eta>0$.
  Thus, it suffices to show that
  $
    \E
    ( 
    \max_{1 \le r \le R}
    |
    \bar{x}_{r} - \E[\bar{x}_{r}]
    |
    )
    \to 0
  $
  as $n \to \infty$.
  We denote by the conditional expectation $\E^{\ast}[\cdot]:= \E[\cdot|\bm{z}_{n}]$.
  By the triangle inequality, we have 
  \begin{eqnarray*}
    \big |
    \bar{x}_{r} - \E[\bar{x}_{r}]
    \big |
    \le 
    \big |
    \bar{x}_{r} - \E^{\ast}[\bar{x}_{r}]
    \big |
    +
    \big |
    \E^{\ast}[\bar{x}_{r}]    
    -
    \E[\bar{x}_{r}]
    \big |.
  \end{eqnarray*}
  Because
  $\bar{x}_{r}$ is identically distributed
  conditional on $\bm{z}_{n}$,
  an application of Jensen's inequality yields
  $
    \E
    ( 
    \max_{1 \le r \le R}
    |
    \E^{\ast}[\bar{x}_{r}]    
    -
    \E[\bar{x}_{r}]
    |
    )
    \le
    (
    \E
    |
    \E^{\ast}[\bar{x}_{r}]    
    -
    \E[\bar{x}_{r}]
    |^{2}
    )^{1/2}
    \le
    (n^{-1} \E[x_{ir}^2])^{1/2}
    \lesssim n^{-1/2}
  $.
  Thus,
  it suffices
  to show that
  $
    \E
    \big [ 
    \max_{1 \le r \le R}
    |
    \bar{x}_{r} - \E^{\ast}[\bar{x}_{r}]
    |
    \big ] 
    =
    o(1)
  $.
  To use the symmetrization technique,
  let $\varepsilon_{1}, \dots, \varepsilon_{n}$
  be independent Rademacher random variables
  that are independent of
  $\bm{z}_{n}$
  and 
  $\bm{u}_{n,R}$.
  Given an independent
  sequence 
  $\{x_{ir}\}$
  conditional on $\bm{z}_{n}$, 
  Lemma 2.3.1 in \cite{van1996weak} implies that
  \begin{eqnarray*}
    \E^{\ast}
    \Big [
    \max_{1 \le r \le R}
    \big |
    \bar{x}_{r} - \E^{\ast}[\bar{x}_{r}]
    \big |
    \Big ]
    \le 
    2
    n^{-1}
    \E^{\ast}
    \Big [　
    \max_{1 \le r \le R}
    \Big |
    \sum_{i=1}^{n}
    \varepsilon_{i}
    x_{ir}
    \Big |
    \Big ].
  \end{eqnarray*}
  By Lemma 2.2.2 and 2.2.7 in \cite{van1996weak}, we have 
  \begin{eqnarray*}
    \E
    \bigg [　
    \max_{1 \le r \le R}
    \Big |
    \sum_{i=1}^{n}
    \varepsilon_{i}
    x_{ir}
    \Big |
    \
    \Big |
    \bm{z}_{n},
    \bm{u}_{n,R}
    \bigg ]　
    \lesssim
    \max_{1 \le r \le R}
    \bigg (
    \sum_{i=1}^{n}
    x_{ir}^{2}
    \bigg )^{1/2}
    (\log R)^{1/2}. 
  \end{eqnarray*}
  By Fubini's theorem and Jensen's inequality, we have 
  \begin{eqnarray*}
    \E
    \Big [
    \max_{1 \le r \le R}
    \big |
    \bar{x}_{r} - \E^{\ast}[\bar{x}_{r}]
    \big |
    \Big ]
    \lesssim
    n^{-1}
    \Big (
    \E
    \Big [
    \max_{1 \le r \le R}
    \sum_{i=1}^{n}
    x_{ir}^{2}
    \Big]
    \Big )^{1/2}
    (\log R)^{1/2}
    \lesssim
    \bigg (
    \frac{\log R}{n}
    \bigg )^{1/2},
  \end{eqnarray*}
  where the last inequality holds because 
  $
  n^{-1}
    \E
    [
    \max_{1 \le r \le R}
    \sum_{i=1}^{n}
    x_{ir}^{2}
    ]
    \le
    \E[
    \max_{1 \le r \le R}
    |x_{ir}|^{2}
    ]
    < \infty
  $.
  Given that 
  $\log (R)/n \to 0$, the desired conclusion follows.
\end{proof}
\vspace{0.2cm}

For each $r=1, \dots, R$,
define a function 
$
M_{n}^{r}:
\Theta^2{\times}\mathcal{B} \to \mathbb{R}^{d_\beta}
$, given by 
\begin{eqnarray*}
  M_{n}^{r}(\theta, \theta^{\ast}, \beta)
  :=
  n^{-1}
  \sum_{i=1}^{n}
  \sum_{t=1}^{T}  
  m_{it}^{r}
  (\theta, \theta^{\ast}, \beta).  
\end{eqnarray*}
The estimator
$\hat{\beta}^{r}(\theta, \theta^{\ast})$
is the solution to the equation 
$M_{n}^{r}(\theta, \theta^{\ast}, \beta) = 0$,
given 
$(\theta, \theta^{\ast}) \in \Theta^2$.
The following lemma establishes the properties of
$\hat{\beta}^{r}(\theta, \theta^{\ast})$.

\vspace{0.2cm}
\begin{lemma}
  \label{lemma:beta}
  Suppose that Assumptions 1-4 hold. 
  If $m(y,z,\beta)$ is twice continuously differentiable in $\beta$
  for any $(y,z) \in \mathcal{Y} \times \mathcal{Z}$  
  and
  $(\log R)/ n \to 0$.
  Then,
  for $\ell \in \{0, 1, 2\}$,
  as $n \to \infty$, 
  \begin{eqnarray*}
    \partial_{\theta}^{\ell} 
    \hat{\beta}^{r}(\theta, \theta^{\ast})|_{\theta = \theta^{\ast}}
    \pto
    \partial_{\theta}^{\ell}  \beta(\theta^{\ast}),
  \end{eqnarray*}
  uniformly in 
  $(\theta, r) \in \Theta \times \{1, \dots, R\}$.
\end{lemma}
\begin{proof}

  Consider
  the uniform consistency of $\hat{\beta}^{r}(\theta^{\ast}, \theta^{\ast})$
  to
  $\beta(\theta^{\ast})$.
  For any $\theta^{\ast} \in \Theta$
  and 
  $r = 1, \dots, R$,
  it follows from the definition of 
  $\hat{\beta}^{r}(\theta^{\ast}, \theta^{\ast})$
  that 
  \begin{eqnarray}
    \label{eq:min-1}
    \big \|
    M_{n}^{r}\big (\theta^{\ast},\theta^{\ast}, \beta(\theta, \theta^{\ast}) \big) 
    \big \|
    - 
    \big \|
    M_{n}^{r}\big (\theta^{\ast},\theta^{\ast}, \hat{\beta}^{r}(\theta, \theta^{\ast})
    \big )
    \big \|
    \ge 0.
  \end{eqnarray}
  Let $\delta >0$ be an arbitrary constant. 
  The continuity of
  $M(\theta^{\ast}, \beta)$
  in $\beta \in \mathcal{B}$
  implies that, 
  for any $r \in \{1, \dots, R\}$
  with 
  $
  \| \hat{\beta}^{r}(\theta^{\ast}, \theta^{\ast})
  -
  \beta(\theta^{\ast}) \| > \delta$,
  there exists an $\epsilon >0$ such that 
  $
    \big \|
    M \big (\theta^{\ast}, \hat{\beta}^{r}(\theta^{\ast}, \theta^{\ast}) \big ) 
    \big \|
    - 
    \big \|
    M \big (\theta^{\ast},\beta(\theta^{\ast}) \big) 
    \big \|
    > 2 \epsilon,
  $
  which together with 
  (\ref{eq:min-1})
  and the triangle inequality 
  implies that 
  \begin{eqnarray*}
    \Pr 
    \bigg (
    \sup_{\theta^{\ast} \in \Theta}
    \max_{1 \le r \le R}
    \| \hat{\beta}^{r}(\theta^{\ast}, \theta^{\ast}) - \beta(\theta^{\ast}) \| > \delta
    \bigg )
    \le  
    \Pr 
    \bigg (
    \max_{1 \le r \le R}
    \sup_{(\theta^{\ast}, \beta) \in \Theta \times \mathcal{B}}
    \| 
    M_{n}^{r}(\theta^{\ast}, \theta^{\ast}, \beta) - M(\theta^{\ast},\beta)
    \| > \epsilon
    \bigg ).
  \end{eqnarray*}
  Given the compact parameter spaces,
  it suffices to show that,
  for an arbitrarily small $\delta>0$,
  \begin{eqnarray}
    \label{eq:m-conv}
    \max_{1 \le r \le R}
    \sup_{(\theta^{\ast},\beta) \in \mathcal{N}_{\delta}}
    \| 
    M_{n}^{r}(\theta^{\ast}, \theta^{\ast}, \beta) - M(\theta^{\ast},\beta)
    \|
    =
    o_p(1),
  \end{eqnarray}
  as $n \to \infty$,
  where
  $\mathcal{N}_{\delta}:=\mathcal{N}_{1,\delta}\times \mathcal{N}_{2,\delta}
  \subset \Theta \times \mathcal{B}
  $
  with
  $\sup_{\theta_{1}, \theta_{2} \in \mathcal{N}_{1, \delta}}
  \|\theta_{2} - \theta_{1}\| \le \delta$
  and 
  $
  \sup_{\beta_{1}, \beta_{2}\in \mathcal{N}_{2,\delta}}
  \|\beta_{2} - \beta_{1}\| \le \delta$.
  Let $\delta>0$ be fixed and define 
  \begin{eqnarray*}
    \mu_{n, \delta}^{r-}
    :=
    n^{-1}
    \sum_{i=1}^{n}
    \inf_{(\theta^{\ast}, \beta) \in \mathcal{N}_{\delta}}
    m_{it}^{r}(\theta^{\ast}, \theta^{\ast}, \beta)
    \ \ \ \mathrm{and}  \ \ \
    \mu_{n, \delta}^{r+}
    :=
    n^{-1}
    \sum_{i=1}^{n}
    \sup_{(\theta^{\ast}, \beta) \in \mathcal{N}_{\delta}}
    m_{it}^{r}(\theta^{\ast}, \theta^{\ast}, \beta).
  \end{eqnarray*}
  An application of the triangle inequality yields that
  \begin{eqnarray*}
    \max_{1 \le r \le R}
    \sup_{(\theta, \theta^{\ast}, \beta) \in \mathcal{N}_{\delta}}    
    \| 
    M_{n}^{r}(\theta^{\ast}, \theta^{\ast}, \beta)
    -
    M(\theta^{\ast}, \beta)
    \| 
    &\le&
    \max_{1 \le r \le R}
    \big \|
    \mu_{n,\delta}^{r-}
    -
    \E[\mu_{n,\delta}^{r-}]
    \big \|
    +
    \max_{1 \le r \le R}
    \big \|
    \mu_{n,\delta}^{r+}
    -
    \E[\mu_{n,\delta}^{r+}]
    \big \| \\
    && +
    2   
    \max_{1 \le r \le R}
    \big \|
    \E[\mu_{n,\delta}^{r+}]
    -
    \E[\mu_{n,\delta}^{r-}]
    \big \|.
  \end{eqnarray*}
  Lemma \ref{lemma:FN} implies that 
  $
    \max_{1 \le r \le R}
    \|
    \mu_{n,\delta}^{r\pm}
    -
    \E[\mu_{n,\delta}^{r\pm}]
    \|
    = o_p(1)
  $
  as $n \to \infty$.
  Also,
  a similar argument
  in Proposition \ref{mc1}
  shows that 
  $  \max_{1 \le r \le R}
     \|
     \E[\mu_{n,\delta}^{r+}]
     -
     \E[\mu_{n,\delta}^{r-}]
    \|
    \to 0
  $ as $\delta \to 0$.
  Hence,
  (\ref{eq:m-conv}) follows and the result follows for the case of $\ell=0$.

  Next, we consider the first derivative.
  The estimator $\hat{\beta}^{r}(\theta, \theta^{\ast})$
  satisfies that
  $
  M_{n}^{r}\big (\theta,\theta^*, \hat{\beta}^{r}(\theta, \theta^{\ast}) \big)
  = 0
  $
  for each $r=1, \dots, R$.
  Taking the first derivative of the implicit function with respect to $\theta$,
  we obtain 
  \begin{eqnarray}
    \label{eq:exp-500}
    \partial_{\theta}
    M_{n}^{r}\big (\theta,\theta^*, \hat{\beta}^{r}(\theta, \theta^{\ast}) \big)
    +
    \partial_{\beta}
    M_{n}^{r}\big (\theta,\theta^*, \hat{\beta}^{r}(\theta,\theta^{\ast}) \big )
    \partial_{\theta}
    \hat{\beta}^{r}(\theta, \theta^{\ast})
    =
    0.
  \end{eqnarray}
  Using the same argument to show (\ref{eq:m-conv}),
  we can apply
  the argument for the uniform law of large numbers
  in \cite{jennrich1969}
  with 
  the result in
  Lemma \ref{lemma:FN} 
  to obtain 
  \begin{eqnarray}
    \label{eq:ex-503}
    \partial_{\beta}
    M_{n}^{r} (\theta,\theta^*, \beta )|_{\theta = \theta^{\ast}}
    \pto 
    \partial_{\beta} M (\theta^{\ast}, \beta),
  \end{eqnarray}
  uniformly in 
  both 
  $(\theta^*, \beta) \in \Theta \times \mathcal{B}$
  and 
  $r = 1, \dots, R$. 
  Since 
  $
  \hat{\beta}^{r}(\theta^{\ast}, \theta^{\ast})
  \pto 
  \beta(\theta^{\ast})
  $
  uniformly in 
  $(\theta^{\ast}, r) \in \Theta \times \{1, \dots, R\}$,
  we have that 
  $
    \partial_{\beta}
    M_{n}^{r}\big (\theta^*,\theta^*, \hat{\beta}^{r}(\theta^*, \theta^{\ast}) \big)
    $
  converges in probability to   
  a non-singular matrix 
  $
    \partial_{\beta} M \big (\theta^{\ast}, \beta(\theta^{\ast}) \big)
  $
  uniformly in
  $(\theta^{\ast}, r) \in \Theta \times \{1, \dots, R\}$.
  Similarly,
  it can be shown that 
  $
  \partial_{\theta} 
  M_{n}^{r}\big (\theta,\theta^*, \hat{\beta}^{r}(\theta^*, \theta^{\ast}) \big)
  |_{\theta = \theta^{\ast}}
  \pto 
  \partial_{\theta}
  M\big (\theta^*, \beta(\theta^*) \big)
  $
  uniformly in
  $(\theta^{\ast}, r) \in \Theta \times \{1, \dots, R\}$.
  This together with
  (\ref{eq:exp-500})-(\ref{eq:ex-503})
  and 
  Lemma \ref{lemma:exist-beta-der}(b)
  implies the desired result. 

  A similar argument can show the second derivative part and thus we omit the details. 
\end{proof}
\vspace{0.3cm}

\begin{proof}[\textbf{Proof of Theorem \ref{theorem:Q}}]

We first show that the result holds for $\ell=0$. First, we consider $Q_n^{\text{LM}}(\theta,\theta)$. Using a similar argument to Lemma \ref{lemma:beta}, we can show that $\hat{\beta}\pto\beta^0$. Also, from Proposition \ref{mc1}, we have that $M_n(\theta,\theta,\beta)\pto M(\theta,\beta)$ uniformly over $(\theta, \beta) \in \Theta {\times} \mathcal{B}$. Because $\Omega_n\pto\Omega$ under Assumption \ref{three}(c), the uniform convergence of $Q_n^{\text{LM}}(\theta,\theta)$ to $Q^{\text{LM}}(\theta)$ follows. Now, consider $Q_n^{\text{W}}(\theta,\theta)$. By Lemma \ref{lemma:beta}, $\bar{\beta}^{R}(\theta,\theta)\pto\beta(\theta)$, uniformly over $\theta \in\Theta$. Therefore, by Assumption \ref{three}(c),  we can conclude that $Q_n^{\text{W}}(\theta,\theta)$ converges to $Q^{\text{W}}(\theta)$ uniformly.

  We next show that the desired conclusion holds for
  the first derivative of two criterion functions.
  A similar argument can be used to prove
  the result for the second derivative and we omit the details for brevity. 

  First, we consider the LM criterion function.
  The first derivative of $Q_{n}^{\text{LM}}(\theta, \theta^{*})$
  is 
  given by 
  \begin{flalign*}
    \partial_{\theta} Q_{n}^{\text{LM}}(\theta, \theta^{*})
    |_{\theta = \theta^{\ast}}
    & 
    =
    2
    \big [
    \partial_{\theta} M_{n}(\theta,\theta^{\ast},\hat{\beta})
    |_{\theta = \theta^{\ast}}
    \big ]'
    \Omega_{n}
    M_{n}(\theta^{\ast},\theta^{\ast},\hat{\beta}),
  \end{flalign*}
  and
  its population counterpart is
  \begin{eqnarray*}
    \partial_{\theta} Q^{\text{LM}}(\theta^{\ast})
    =
    2 [\partial_{\theta}M(\theta^{\ast}, \beta^{0}) ]'
    \Omega
    M(\theta^{\ast}, \beta^{0}).
  \end{eqnarray*}
  Using a similar argument in Lemma \ref{lemma:beta},
  we can show that $\hat{\beta} \pto \beta^{0}$
  as $n\to\infty$.
  Thus,
  Proposition \ref{mc1} implies that
  $
  \partial_{\theta} M_{n}(\theta,\theta^{\ast},\hat{\beta})
  |_{\theta = \theta^{\ast}}
  \pto
  \partial_{\theta} M(\theta^{\ast}, \beta^{0})
  $
  uniformly in $\theta^{\ast} \in \Theta$.
  Also,
  a similar argument used in 
  Proposition \ref{mc1}
  shows that 
  $M_{n}(\theta^{\ast},\hat{\beta}) \pto M(\theta^{\ast}, \beta^{0})$
  uniformly in $\theta^{\ast} \in \Theta$.
  These results with the condition 
  that 
  $\Omega_{n} \pto \Omega$
  under Assumption \ref{three}(c)
  lead to the desired conclusion.

  Next, we consider the Wald approach.
  The first derivative of $Q_{n}^{\text{W}}(\theta, \theta^{*})$
  is given by 
  \begin{flalign*} 
    \partial_{\theta} Q_{n}^{\text{W}}(\theta, \theta^{*})|_{\theta=\theta^*}
    &=
    2 
    \big [
    \partial_{\theta}
    \bar{\beta}^{R}(\theta, \theta^{*})|_{\theta=\theta^*}
    \big ]'
    \Omega_{n}
    \big [
    \bar{\beta}^{R}(\theta^{\ast}, \theta^{\ast}) - \hat{\beta}
    \big ],
  \end{flalign*}
  and
  Lemma \ref{lemma:exist-beta-der}(b) implies that 
  its population counterpart is 
  \begin{eqnarray*}
    \partial_{\theta} Q^{\text{W}}(\theta^*)
    =
    2
    \big [
    \partial_{\theta}\beta(\theta^*)
    \big ]'
    \Omega
    \big [
    \beta(\theta^*) - \beta^{0}  
    \big ],
  \end{eqnarray*}
  where
  $ \partial_{\theta}\beta(\theta^*)
  =
  -
  \big[
  \partial_{\beta}M \big (\theta^*,\beta(\theta^*) \big)
  \big ]^{-1} 
  \partial_{\theta}M \big (\theta, \beta(\theta^*) \big)|_{\theta=\theta^*}
  $.
  By definition, we have 
  \begin{eqnarray*}
    \bar{\beta}^{R}(\theta^{\ast}, \theta^{\ast})
    =
    R^{-1}\sum_{r=1}^{R} \hat{\beta}^{r}(\theta^{\ast}, \theta^{\ast})
    \ \ \ \mathrm{and} \ \ \
    \partial_{\theta}\bar{\beta}^{R}(\theta, \theta^{\ast})|_{\theta = \theta^{\ast}}
    =
    R^{-1}\sum_{r=1}^{R}
    \partial_{\theta}
    \hat{\beta}^{r}(\theta, \theta^{\ast})
    |_{\theta = \theta^{\ast}}.
  \end{eqnarray*}
  By Lemma \ref{lemma:beta},
    $
  \bar{\beta}^{R}(\theta^{\ast}, \theta^{\ast})  
  \pto 
  \beta(\theta^{\ast})
  $
  and, by repeating the arguments in Lemma \ref{lemma:beta} for $\partial_\theta \bar{\beta}^{R}(\theta,\theta^*)|_{\theta=\theta^*}$, we can obtain a similar result, namely 
  $
    \partial_{\theta}
    \bar{\beta}^{R}(\theta, \theta^{\ast})
    |_{\theta = \theta^{\ast}}
    \pto
    \partial_{\theta}
    \beta(\theta^{\ast})
  $
  uniformly in $\theta^{\ast} \in \Theta$.
  Given $\hat{\beta}\pto \beta^0$ and
  $\Omega_{n} \pto \Omega$ under Assumption \ref{three}(c),
  we can conclude that
  $
  \partial_{\theta} Q_{n}^{\text{W}}(\theta, \theta^{\ast})
  |_{\theta = \theta^{\ast}}  \pto
  \partial_{\theta} Q^{\text{W}}(\theta^{\ast})
  $
  uniformly in $\theta^{\ast} \in \Theta$.
\end{proof}

\vspace{0.3cm}
\begin{proof}[\textbf{Proof of Theorem \ref{cons}}]
	Let
	$(Q, Q_{n}, \hat{\theta})$
	denote either
	$(Q^{\text{W}}, Q_{n}^{\text{W}}, \hat{\theta}^{\text{W}})$
	or 
	$(Q^{\text{LM}},Q_{n}^{\text{LM}}, \hat{\theta}^{\text{LM}})$.
	By Assumption 4(b) or the results of Lemma
	\ref{lemma:exist-beta-der}(b),
	for every $\delta>0$
	with   
	$\|\theta-\theta^0\|\geq \delta$,
	there exists some $\epsilon>0$ such that
	$Q(\theta) - Q(\theta^0) \ge \epsilon$.
	This implies that,
	for any $\delta>0$, there exists an $\epsilon>0$ such that 
	\begin{eqnarray*}
		\Pr
		\big (
		\|\hat{\theta}-\theta^0\|\geq \delta
		\big )
		\leq
		\Pr
		\big (
		Q(\hat{\theta}) - Q(\theta^{0}) \geq \epsilon
		\big ).  
	\end{eqnarray*}
	Thus, it suffices to show that
	$Q(\hat{\theta}) - Q(\theta^{0})=o_{p}(1)$.
	As the minimizer of $Q_{n}(\theta, \theta)$
	the estimator $\hat{\theta}$ satisfies that
	$Q_{n}(\theta^{0}, \theta^{0}) - Q_{n}(\hat{\theta}, \hat{\theta}) \ge 0$.
	Thus, an application of the triangle inequality yields 
	\begin{eqnarray*}
		|Q(\hat{\theta}) - Q(\theta^{0})|
		\le
		|
		Q_{n}(\hat{\theta}, \hat{\theta})
		-
		Q(\hat{\theta})
		|
		+
		|
		Q_{n}(\theta^{0}, \theta^{0})
		-
		Q(\theta^0)
		|.
	\end{eqnarray*}
	The right-hand side of the above equation
	is bounded from above by
	$2
	\sup_{\theta \in \Theta}
	|
	Q_{n}(\theta, \theta)
	-
	Q(\theta)
	|
	$,
	which
	converges in probability to 0 as $n\to \infty$
	by Theorem \ref{theorem:Q} (with $\ell =0$).
\end{proof}

\vspace{0.3cm}
\begin{lemma}
  \label{lemma:BL}
  Suppose that Assumptions 1-5 hold
  and $\log (R)/n \to \infty$ as $n \to \infty$.
  Also, assume that 
  $m(y,z,\beta)$ is continuously differentiable in $\beta$
  for any $(y,z) \in \mathcal{Y} \times \mathcal{Z}$
  and
  $\partial_{\beta} M\big (\theta, \beta(\theta) \big )$
  is non-singular
  in the neighborhood
  $\mathcal{N}_{\delta}(\theta^{0})$
  for some constant $\delta>0$.
  Then, 
  \begin{eqnarray*}
    \sqrt{n}
    \big (
    \hat{\beta}^{r}(\theta, \theta)
    -
    \beta(\theta)
    \big )
    =
    -
    \big [
    \partial_{\beta}
    M
    \big (
    \theta,\beta(\theta)
    \big )
    \big ]^{-1}
    \sqrt{n}
    M_{n}^{r}
    \big (\theta,\theta, \beta(\theta)\big)
    +o_{P}(1),    
  \end{eqnarray*}
  uniformly in $(\theta,r) \in \mathcal{N}_{\delta}(\theta^{0})  \times \{1, \dots, R\}$.
\end{lemma}
\begin{proof}
  For each $r= 1, \dots, R$,
  the estimator
  $\hat{\beta}^{r}(\theta, \theta)$
  is the solution to
  the equation 
  $M_{n}^{r}
  (
  \theta, \theta, \beta  
  ) = 0$,
  given $\theta \in \Theta$.
  For any $\theta \in \Theta$ and $r \in \{1, \dots, R\}$,
  Taylor's expansion yields
  \begin{eqnarray*}
    0
    =
    \sqrt{n}
    M_{n}^{r}
    \big (
    \theta, \theta,\beta(\theta)
    \big )
    +
    \partial_{\beta}
    M_{n}^{r}
    \big (
    \theta, \theta,\tilde{\beta}^{r}(\theta)
    \big )
    \sqrt{n}
    \big (
    \hat{\beta}^{r}(\theta, \theta)
    -
    \beta(\theta)
    \big ),
  \end{eqnarray*}
  where
  $\tilde{\beta}^{r}(\theta)$ is 
  between
  $\hat{\beta}^{r}(\theta, \theta)$
  and 
  $\beta(\theta)$.
  As explained in the proof of Lemma \ref{lemma:beta},
  we can show that
  $\partial_{\beta}
  M_{n}^{r}(\theta, \theta,\beta^{r})
  \pto
  \partial_{\beta} M(\theta,\beta^{r})
  $
  uniformly in $(\theta, \beta^{r}) \in \Theta {\times} \mathcal{B}$
  and $r \in \{1, \dots, R\}$.
  Lemma \ref{lemma:beta} implies that
  $\tilde{\beta}^{r}(\theta) \pto
  \beta(\theta)
  $
  uniformly in $(\theta,r) \in \Theta \times  \{1, \dots, R\}$.
  Thus, 
  $
  \partial_{\beta}
  M_{n}^{r}
  \big (
  \theta, \theta,\tilde{\beta}^{r}(\theta)
  \big )
  $
  converges to a non-singular matrix
  $\partial_{\beta}
  M
  \big (
  \theta, \beta(\theta)
  \big )
  $. 
  Therefore,
  the desired conclusion holds.
\end{proof} 

\vspace{0.3cm}
\begin{proof}[\textbf{Proof of Theorem \ref{norm}}]

  Let
  $(\hat{\theta}, Q_{n}, Q)$ be either
  $(\hat{\theta}^{\text{LM}}, Q_{n}^{\text{LM}}, Q^{\text{LM}})$
  or
  $(\hat{\theta}^{\text{W}}, Q_{n}^{\text{W}}, Q^{\text{W}})$.
  Suppose that 
  the estimator $\hat{\theta}$ satisfies that 
  $\partial_{\theta} Q_{n}(\theta, \hat{\theta})
   |_{\theta = \hat{\theta}}
   = o_p(n^{-1/2}) 
  $. 
 Taylor's expansion yields that 
  \[
  o_p(1)
  =
  \sqrt{n}
  \partial_{\theta} Q_{n}(\theta, \hat{\theta})|_{\theta = \theta_0}
  +
  \partial_{\theta}^{2} Q_{n}(\theta, \hat{\theta})|_{\theta = \tilde{\theta}}
  \sqrt{n}
  (\hat{\theta}- \theta^{0}),
  \] 
  where $\tilde{\theta}$ is between $\hat{\theta}$ and $\theta^{0}$.
  Theorem \ref{theorem:Q}
  with the consistency of $\hat{\theta}$
  in Theorem \ref{cons}
  implies that
  $
  \partial_{\theta}^{2} Q_{n}(\theta, \hat{\theta})|_{\theta = \tilde{\theta}}
  \pto
  \partial_{\theta}^{2} Q(\theta^{0})
  $.
  We have 
  \begin{eqnarray*}
    \sqrt{n}
    (\hat{\theta} - \theta_{0})
    = 
    -
    \big \{
    \partial_{\theta}^{2} Q(\theta^{0})
    \big \}^{-1}
    \sqrt{n}
    \partial_{\theta} Q_{n}(\theta, \hat{\theta})|_{\theta = \theta_0}
    +
    o_p(1).
  \end{eqnarray*}
  For the Hessian $\partial_{\theta}^{2} Q(\theta^{0})$, we can show that 
  \begin{eqnarray}
    \label{eq:H-Q} 
    \partial_{\theta}^{2} Q^{\text{LM}}(\theta^{0})
    = 2\Delta' \Omega \Delta
    \ \ \ \ \ \ \mathrm{and} \ \ \ \ \ \
    \partial_{\theta}^{2} Q^{\text{W}}(\theta^{0})
    =
    2\Gamma' \Omega \Gamma.
  \end{eqnarray}
  It remains to consider
  the convergence of 
  $
    \sqrt{n}
    \partial_{\theta} Q_{n}(\theta, \hat{\theta})|_{\theta = \theta_0}
  $
  in distribution.

  First, we consider the LM estimator.
  We can show that
  \begin{eqnarray*}
   \sqrt{n}
    \partial_{\theta} Q_{n}^{\text{LM}}(\theta, \hat{\theta}^{\text{LM}})
    |_{\theta = \theta^{0}}
    &=&
    2
    \big [
    \partial_{\theta} M_{n}
    \big (\theta, \hat{\theta}^{\text{LM}},\hat{\beta} \big)
    |_{\theta = \theta^{0}}
    \big ]'
    \Omega_{n}
   \sqrt{n}
    M_{n}
        \big (\theta^{0}, \hat{\theta}^{\text{LM}},\hat{\beta} \big)
    \\
    &=&
    2
    \big [
    \partial_{\theta} M
    \big (\theta^{0}, \beta^{0} \big)
    \big ]'
    \Omega
   \sqrt{n}
    M_{n}
    \big (\theta^{0}, \hat{\theta}^{\text{LM}},\hat{\beta} \big)
    + o_p(1),
  \end{eqnarray*}
  where the second equality holds due to Proposition 1
  and Assumption \ref{three}(c).
  Lemma \ref{lemma:don} implies that
  the map  
  $(\theta, \beta) \to M_{n}(\theta^{0}, \theta, \beta)$
  is stochastic equicontinuous:
  there exists a $\delta>0$ such that
   \begin{eqnarray*}
    \sup_{(\theta, \beta) \in \mathcal{N}_{\delta}}
    \big \|
    \{M_{n}(\theta^{0}, \theta, \beta)
    -
    \mathbb{E}[M_n(\theta^{0}, \theta, \beta)]
    \}
    - 
    \{M_{n}(\theta^{0}, \theta^{0}, \beta^{0})
    -
    M(\theta^{0}, \theta^{0}, \beta^{0})
    \} 
    \big \|
    = o_p(n^{-1/2}),
  \end{eqnarray*}
  for
  every neighborhood 
  $\mathcal{N}_{\delta} \subset \Theta \times \mathcal{B}$
  satisfying
  that 
  $
  \|\theta_{2} - \theta_{1}\| \le \delta
  $
  and 
  $
  \|\beta_{2} - \beta_{1}\| \le \delta
  $
  for any
  $(\theta_{1}, \beta_{1}),
  (\theta_{2}, \beta_{2})  
  \in
  \mathcal{N}_{\delta}$.
  Because 
  $ \mathbb{E}[M_n(\theta^{0}, \theta, \beta)] = M(\theta^{0}, \beta)$,
  $M_{n}(\theta^{0}, \theta^{0}, \beta^{0}) = M_{n}(\theta^{0}, \beta^{0})$
  and 
  $
   \mathbb{E}[M_n(\theta^{0}, \theta^{0}, \beta)] =0
  $,
  we have 
  \[
   \sqrt{n}
    M_{n}
    \big (\theta^{0}, \hat{\theta}^{\text{LM}},\hat{\beta} \big)
    = 
   \sqrt{n}
    M(\theta^{0}, \hat{\beta})
    +
   \sqrt{n}
    M_{n}
    \big (\theta^{0},  \beta^{0}  \big)
    + o_p(1).
  \]  
  By Taylor's expansion,
  $   \sqrt{n}
    M(\theta^{0}, \hat{\beta})
    =
    \partial_{\beta} M(\theta^{0}, \beta^{0})'
     \sqrt{n} (\hat{\beta} - \beta^{0}) +o_p(1)
  $,
  Applying a similar argument used in Lemma \ref{lemma:BL},
  we can show that
  $
  \sqrt{n}
  M(\theta^{0}, \hat{\beta})
  =
  -
  n^{-1/2}
  \sum_{i=1}^{n}
  m(y_{i}, z_{i}, \beta^{0})
  +o_p(1)
  $.  
  This together with (\ref{eq:H-Q})
  yields 
  \[                              
   \sqrt{n}
   (\hat{\theta}^{\text{LM}} - \theta^{0})
   =
   (\Delta' \Omega \Delta)^{-1}
   \Delta' \Omega
   \sqrt{n}
   \bar{\xi}_{n, R}
    + o_p(1),
  \]
  where
  $
  \bar{\xi}_{n, R}
  :=
  (nR)^{-1}
  \sum_{r=1}^{R}
  \sum_{i=1}^{n}
  \sum_{t=1}^{T}     
  \xi^r_{it}
  $.
  The desired result follows from Assumption \ref{four}(c).

  Next, we consider the case of the Wald estimator. We have  
  \begin{eqnarray*}
    \partial_{\theta} Q_{n}^{\text{W}}(\theta, \hat{\theta}^{\text{W}})
    |_{\theta = \hat{\theta}^{\text{W}}}
    =
    2 
    \big [
    \partial_{\theta}
    \bar{\beta}^{R}(\theta, \hat{\theta}^{\text{W}})
    |_{\theta = \hat{\theta}^{\text{W}}}
    \big ]'
    \Omega_{n}
    \sqrt{n}    
    \big [
    \bar{\beta}^{R}(\hat{\theta}^{\text{W}}, \hat{\theta}^{\text{W}})- \hat{\beta}
    \big ].
  \end{eqnarray*}
  As shown in the proof of Theorem \ref{theorem:Q},
  $
    \partial_{\theta}
    \bar{\beta}^{R}(\theta, \theta^*)|_{\theta=\theta^*}
    \pto
    \partial_{\theta}
    \beta(\theta^*)
  $
  uniformly in $\theta^* \in \Theta$.
  Also,
  $\hat{\theta}^{\text{W}} \pto \theta^{0}$
  from Theorem \ref{cons}
  and
  $\Omega_{n} \pto \Omega$
  under Assumption \ref{three}(c). It follows that 
  \begin{eqnarray*}
    \partial_{\theta} Q_{n}^{\text{W}}(\theta, \hat{\theta}^{\text{W}})
    |_{\theta = \hat{\theta}^{\text{W}}}
    =
    2 
    \big [
    \partial_{\theta}
    \beta(\theta^{0})    
    \big ]'
    \Omega
    \sqrt{n}    
    \big [
    \bar{\beta}^{R}(\hat{\theta}^{\text{W}}, \hat{\theta}^{\text{W}}) -\hat{\beta}
    \big ] + o_p(1).
  \end{eqnarray*}
  Given the consistency of $\hat{\theta}^{\text{W}}$,
  Lemma \ref{lemma:BL} implies that 
  \begin{eqnarray*}
    \sqrt{n}    
    [\bar{\beta}^{R}(\hat{\theta}^{\text{W}}, \hat{\theta}^{\text{W}}) - \hat{\beta}]
    =
    -
    [\partial_{\beta}M(\theta^{0}, \beta^{0})]^{-1}
    \sqrt{n} \bar{\xi}_{n,R}
    +o_p(1).
  \end{eqnarray*}
  Collecting the results so far, we obtain
  \begin{eqnarray*}
    \sqrt{n}
    (\hat{\theta}^{\text{W}} - \theta^{0})
    =
    -
    (\Gamma' \Omega \Gamma )^{-1} \Gamma'\Omega
    [\partial_{\beta}M(\theta^{0}, \beta^{0})]^{-1}
    \sqrt{n} \bar{\xi}_{n,R}
    + o_p(1).
  \end{eqnarray*}
  The desired result follows from Assumption \ref{four}(c). 
\end{proof}
\vspace{0.3cm}

We provide a technical lemma to obtain the asymptotic distribution of the proposed estimator. 
To this end, we will use some notations and results from the literature on empirical process only in the lemma below.    
Let $(z, u)$ be random variables with a probability distribution $\mathbb{P}$,
where $z$ is a random vector with support $\mathcal{Z}$
and 
$u$ is a standard uniform random variable. 
We assume that $z$ and $u$ are statistically independent.
Let $\mathcal{B}$ and $\Theta$ are compact parameter space with finite dimensions $d_{\beta}$ and $d_{\theta}$, respectively.
Define measurable functions
$\mu_{\beta}:\mathcal{Z}  \to \mathbb{R}$
and  
$\phi_{\theta}: \mathcal{Z} \to \mathbb{R}$
for $\beta \in \mathcal{B}$ and $\theta \in \Theta$.
Let $\mathcal{F}$ be a collection of measurable functions $f_{\theta, \beta}:\mathcal{Z}\times[0,1] \to \mathbb{R}$
indexed by prameters $(\theta, \beta) \in \Theta \times \mathcal{B}$, given by 
$f_{\theta, \beta}(z, u):= \mu_{\beta}(z)\1[u \le \phi_{\theta}(z)]$
for $(z, u) \in \mathcal{Z} \times [0,1]$.
For some $\epsilon > 0$, let 
$N_{[ \ ]}(\epsilon, \mathcal{F}, L_{2}(\mathbb{P}))$ be the braketing number 
and 
the bracketing integral is given by 
$J_{[ \ ]}(\delta, \mathcal{F}, L_{2}(\mathbb{P})):= \int_{0}^{\delta} \sqrt{ \log N_{[ \ ]}(\epsilon, \mathcal{F}, L_{2}(\mathbb{P}))} d \epsilon$
for some $\delta>0$. 
The lemma below will prove that the collection $\mathcal{F}$ is Donsker by showing that the bracketing integral is finite.
A similar result can be found in Lemma 1 of \cite{SantAnnaSong2019}.
We use $\mathrm{diam}(A)$ to denote the diameter of a set $A$.
  
\vspace{0.3cm}
\begin{lemma}
	\label{lemma:don}
	Assume that
	(i) the functions $\mu_{\beta}$ and
	$\phi_{\theta}$
	are twice continuously differentiable
	for $\beta \in \mathcal{B}$ and $\theta \in \Theta$ respectively,
	(ii) 
	$\sup_{\beta \in \mathcal{B}}\|\partial_{\beta} \mu_{\beta}(z) \| \le \nabla \bar{\mu}(z)$
	and 
	$\sup_{\theta \in \Theta}\|\partial_{\theta} \phi_{\theta}(z) \| \le
	\nabla \bar{\phi}(z)$
	with 
	$\E[ \nabla \bar{\mu}(z)^2] < \infty$
	and 
	$\E[ \nabla \bar{\phi}(z)^2] < \infty$,
	(iii) the parameter spaces $\Theta$ and $\mathcal{B}$ are comapct. 
	Then, 
	the collection of functions 
	$\mathcal{F}$  is $\mathbb{P}$-Donsker. 
	
\end{lemma}
\begin{proof}
	Let $\epsilon>0$ be an arbitrary small constant.
	Consider partitions 
	$\{\Theta_{k}\}_{k=1}^{K}$ of $\Theta$
	and $\{\mathcal{B}_{l}\}_{l=1}^{L}$ of $\mathcal{B}$.
	Under the condition (iii),
	there exists finite constants 
	$K \le \big ( \mathrm{diam}(\Theta) /\epsilon \big )^{d_{\theta}}$
	and 
	$L \le \big ( \mathrm{diam}(\mathcal{B}) /\epsilon \big )^{d_{\beta}}$
	such that 
	$\mathrm{diam}(\Theta_{k}) \le \epsilon$  
	and  
	$\mathrm{diam}(\mathcal{B}_{l}) \le \epsilon$ 
	for every $k=1, \dots, K$ and  $l = 1, \dots, L$.
	Fix
	$(k, l) \in \{1, \dots, K\} \times \{1, \dots, L\}$
	and 
	pick up
	some elements 
	$(\theta_{k}, \beta_{l}) \in \Theta_{k} \times \mathcal{B}_{l}$.
	Then, we can show that,
	for any $(\theta,\beta) \in \Theta_{k} \times \mathcal{B}_{l}$
	and  
	for any $z \in \mathcal{Z}$,
	\[
	\phi_{k}^{-}(z) \le \phi_{\theta}(z) \le \phi_{k}^{+}(z)
	\hspace{1cm} \mathrm{and} \hspace{1cm} 
	\mu_{l}^{-}(z) \le \mu_{\beta}(z) \le \mu_{l}^{+}(z) ,
	\]
	where 
	$
	\phi_{k}^{\pm}(z):=\phi_{\theta_{k}}(z) \pm \epsilon \nabla \bar{\phi}(z) 
	$ and $
	\mu_{l}^{\pm}(z):=\mu_{\beta_{l}}(z) \pm \epsilon  \nabla \bar{\mu}(z) .
	$
	We can show that
	\begin{eqnarray*}
		\mu_{\beta_{l}}(z)
		\1[ u \le \phi_{\theta}(z)]
		-
		\epsilon  \nabla \bar{\mu}(z) 
		\le 
		f_{\theta, \beta}(z)
		\le 
		\mu_{\beta_{l}}(z)
		\1[ u \le \phi_{\theta}(z)]
		+
		\epsilon  \nabla \bar{\mu}(z) ,
	\end{eqnarray*}
	for 
	$(\theta, \beta) \in \Theta_{k}\times \mathcal{B}_{l}$.
	It follows that 
	$
	f_{k,l}^{-}
	\le 
	f_{\theta, \beta} 
	\le 
	f_{k,l}^{+} 
	$
	for 
	$(\theta, \beta) \in \Theta_{k}\times \mathcal{B}_{l}$,
	where 
	\begin{eqnarray*}
		f_{k,l}^{+}(z, u)
		&:=&
		\mu_{\beta_{l}}(z)
		\big \{
		\1[  \mu_{\beta_{l}}(z) \ge 0, u \le \phi_{k}^{+}(z)]
		+
		\1[  \mu_{\beta_{l}}(z) < 0, u \le \phi_{k}^{-}(z)]
		\big \}
		+
		\epsilon
		\nabla \bar{\mu}(z),\\
		f_{k,l}^{-}(z, u)
		&:=&
		\mu_{\beta_{l}}(z)
		\big \{
		\1[  \mu_{\beta_{l}}(z) \ge 0, u \le \phi_{k}^{-}(z)]
		+
		\1[  \mu_{\beta_{l}}(z) < 0, u \le \phi_{k}^{+}(z)]
		\big \}
		-
		\epsilon
		\nabla \bar{\mu}(z).
	\end{eqnarray*}
	Using the triangle inequality, we can show that
	\begin{eqnarray*}
		|
		f_{k,l}^{+}(z,u)
		-
		f_{k,l}^{-}(z,u)
		|
		\le 
		| 
		\mu_{\beta_{l}}(z)
		| 
		\1[\phi_{k}^{-}(z) < u \le \phi_{k}^{+}(z)]
		+
		2 \epsilon \nabla \bar{\mu}(z).
	\end{eqnarray*}
	Thus, an application of $c_r$-inequality and Holder's inequality yields 
	\[
	\E
	| 
	f_{k,l}^{+}(z, u)
	-
	f_{k,l}^{-}(z, u)
	|^2
	\le 
	2
	\big \{
	\E
	\big [
	\mu_{\beta_{l}}(z)^2]
	\E\big [
	\1[\phi_{k}^{-}(z) < u \le \phi_{k}^{+}(z)]
	\big]
	+
	(2\epsilon)^2
	\E \big [ 
	\nabla \bar{\mu}(z)^2
	\big ]
	\big \}.
	\]
	Because $u$ is uniformly distributed and independent of $z$, 
	we have that 
	$$ \E[\1[\phi_{k}^{-}(z) < u \le \phi_{k}^{+}(z)]| z]
	\le \phi_{k}^{+}(z)- \phi_{k}^{-}(z) = 2 \epsilon \nabla \bar{\phi}, 
	$$
	which implies that   
	$
	\E
	| 
	f_{k,l}^{+}(z, u)
	-
	f_{k,l}^{-}(z, u)
	|^2
	\lesssim \epsilon
	$.
	Thus, 
	the bracketing number
	$N_{[ \ ]}(\epsilon, \mathcal{F}, L_{2}(\mathbb{P}))$
	is of polynomial order $(1/\epsilon)$
	and 
	the entropy is of smaller order than $\log(1/\epsilon)$.
	Therefore, the bracketing entropy satisfies  
	that $J_{[ \ ]}(\delta, \mathcal{F},  L_{2}(\mathbb{P})) \lesssim \int_{0}^{\delta} \sqrt{\log(1/\epsilon)} d \epsilon$
	for any $\delta \in (0,1)$
	and
	$
	J_{[ \ ]}(\delta, \mathcal{F},  L_{2}(\mathbb{P}))
	\to 0
	$
	as $\delta \to 0$.
	Hence, $\mathcal{F}$ is $\mathbb{P}$-Donsker
	\citep[see Section 2.5 of][for more details]{van1996weak}.
\end{proof}


\clearpage
\setcounter{section}{0} 
\setcounter{equation}{0} 
\setcounter{lemma}{0}\setcounter{page}{1}\setcounter{proposition}{0} %
\renewcommand{\thepage}{B-\arabic{page}} \renewcommand{\theequation}{B.%
\arabic{equation}}\renewcommand{\thelemma}{A.\arabic{lemma}} 
\renewcommand{\theproposition}{B.\arabic{proposition}} 

\section*{Appendix B. Tables}

This appendix contains Tables 1-8 referenced in Section 5.2 of the main text.

\begin{table}[ht]
	\centering
	\caption{Simulation Results for GII-COV}
	\begin{tabular}{lcccccccccc}
		\hline\hline
		&& \multicolumn{4}{c}{$n=200$} &       & \multicolumn{4}{c}{$n=1000$}       \\ \cline{2-6} \cline{8-11}
		& & MBIAS & AB    & STD   & CV95  &       & MBIAS & AB    & STD   & CV95 \\
		\hline
		Model 1&$\gamma$ & 0.0052 & 0.0110 & 0.0281 &  0.9570 &       & -0.0001 & 0.0015 & 0.0049 & 0.9590 \\
		&$\rho$   & -0.0043 & 0.0183 & 0.0419 &  0.9430 &       & -0.0009 & 0.0028 & 0.0075 &  0.9440 \\ 
		&      &       &       &       &       &       &       &       &       &             \\
		Model 2 &$\gamma$& 0.0038 & 0.0105 & 0.0246  & 0.9440 &       & 0.0002 & 0.0016 & 0.0045 & 0.9460 \\
		&$\alpha$ & 0.0039 & 0.0230 & 0.0463  & 0.9370 &       & 0.0008 & 0.0041 & 0.0107 & 0.9510 \\
		&$\rho$& -0.0034 & 0.0174 & 0.0341  & 0.9410 &       & -0.0009 & 0.0032 & 0.0097 & 0.9600 \\
		&&       &       &       &       &              &       &       &       &       \\
		Model 3&$\gamma$ & 0.0033 & 0.0115 & 0.0286 &  0.9500 &       & 0.0002 & 0.0021 & 0.0071 & 0.9610 \\
		&$\alpha$  & 0.0057 & 0.0241 & 0.0519 &  0.9530 &       & 0.0011 & 0.0053 & 0.0141 &  0.9460 \\
		&$\rho$ & -0.0060 & 0.0193 & 0.0443  & 0.9510 &       & -0.0011 & 0.0037 & 0.0110& 0.9600 \\
		\hline\hline
	\end{tabular} 
        \begin{minipage}{0.9\linewidth}
          \textit{Notes.}
          The number of replications for the Monte Carlo simulation
          is 1,000.
          The cross-sectional sample size $n$ is
          200 or 1,000.
          We report the mean bias (MBIAS),
          mean absolute bias (AB), standard deviation (STD)
          and the Monte Carlo coverage of a 95\% Wald-confidence interval
          (CV95). {For the GII-COV procedure, we used $R=10$ simulated samples across every Monte Carlo design.}
        \end{minipage}          
  	\label{tab1}%
\end{table}%

\begin{table}[!htbp]
	\centering
	\caption{Simulation Results for GII-1}
	\begin{tabular}{lcccccccccc}
		\hline\hline
		&& \multicolumn{4}{c}{$n=200$} &       & \multicolumn{4}{c}{$n=1000$}       \\ \cline{2-6} \cline{8-11}
		& & MBIAS & AB    & STD   & CV95  &       & MBIAS & AB    & STD   & CV95 \\		\hline
		Model 1  &$\gamma$        & 0.0541 & 0.0800 & 0.0918 & 0.9240 &       & 0.0061 & 0.0272 & 0.0337  & 0.9530 \\
		&$\rho$  & -0.0300 & 0.0968 & 0.1188 &  0.9450 &       & -0.0228 & 0.0438 & 0.0499 &  0.9310 \\
		& &       &       &       &       &       &       &       &       &                \\
		Model 2&$\gamma$  & 0.0532 & 0.0833 & 0.0990 & 0.9140 &       & 0.0055 & 0.0298 & 0.0370 & 0.9510 \\
		&$\alpha$ & 0.0073 & 0.0836 & 0.1049 & 0.9610 &       & 0.0006 & 0.0339 & 0.0425 & 0.9570 \\
		&$\rho$ & -0.0293 & 0.1062 & 0.1309 &  0.9470 &       & -0.0223 & 0.0479 & 0.0561 & 0.9280 \\
		&&       &       &       &       &       &       &       &       &               \\
		Model 3& $\gamma$& 0.0132 & 0.1478 & 0.3140 & 0.9910 &       & 0.0094 & 0.0611 & 0.0784 & 0.9570 \\
		&$\alpha$ & -0.0003 & 0.1898 & 0.2421 &  0.9580 &       & -0.0183 & 0.0793 & 0.0990 & 0.9470 \\
		&$\rho$  & 0.0067 & 0.2704 & 0.4558 &  0.9860 &       & 0.0164 & 0.1205 & 0.1536 & 0.9440 \\
		\hline\hline
	\end{tabular}%
	\label{tab2} 
        \begin{minipage}{0.9\linewidth}
          \small \textit{Notes.}
          For $n=200$, the bandwidth is $\lambda_n=.08$
          and for $n=1{,}000$, the bandwidth is $\lambda_n=.045$.
          Also see Table \ref{tab1}. {For the GII-1 procedure, we used $R=10$ simulated samples across every Monte Carlo design.}
        \end{minipage}          
\end{table}%

\begin{table}[ht]
		\centering
		\caption{Simulation Results for GII-2}
		\begin{tabular}{lcccccccccc}
		\hline\hline
		&& \multicolumn{4}{c}{$n=200$} &       & \multicolumn{4}{c}{$n=1000$}       \\ \cline{2-6} \cline{8-11}
		& & MBIAS & AB    & STD   & CV95  &       & MBIAS & AB    & STD   & CV95 \\		\hline
			Model 1  &$\gamma$        & 0.0539 & 0.0788 & 0.0939 & 0.9180 &       & 0.0091 & 0.0271 & 0.0333  & 0.9400 \\
			&$\rho$  & -0.0010 & 0.0936 & 0.1185 &  0.9550 &       & -0.0011 & 0.0386 & 0.0485 &  0.9490 \\
			& &       &       &       &       &       &       &       &       &                \\
			Model 2&$\gamma$  & 0.0495 & 0.0791 & 0.0928 & 0.9260 &       & 0.0077 & 0.0282 &0.0350 & 0.9460 \\
			&$\alpha$         & 0.0111 & 0.0820 & 0.1026 & 0.9400 &       & 0.0008 & 0.0324 & 0.0405 & 0.9490 \\
			&$\rho$           & -0.0064 & 0.1022 & 0.1293 &  0.9510 &     & -0.0023 & 0.0437 & 0.0547 & 0.9430 \\
			&&       &       &       &       &       &       &       &       &               \\
			Model 3& $\gamma$& 0.0180 &   0.1011 &   0.1389 &   0.9570 &       &  0.0057  &  0.0434  &  0.0549&    0.9470\\
			&$\alpha$ & 0.0002 &   0.1425&    0.1803 &   0.9540 &               & -0.0071   & 0.0565 &   0.0705  &  0.9440 \\
			&$\rho$  & 0.0243 &   0.2266  &  0.2979 &   0.9400 &                & 0.0042  &  0.0980  &  0.1152  &  0.9420 \\
			\hline\hline
		\end{tabular}%
		\label{tab3}%

                \begin{minipage}{0.9\linewidth}
          \small \textit{Notes.}
          For the first-step of GII-2,
          the bandwidth is $\lambda_n=.03$ and $R=10$ simulated data replications are used; in the second-step, the bandwidth is $\lambda_n=.003$ and we use $R=300$ simulated data sets.          
          Also see  Table \ref{tab1}.
        \end{minipage}          
\end{table}%

\begin{table}[ht]
	\centering
	\caption{Simulation Results for Nelder-Mead Simplex-based Search}
	\begin{tabular}{rcccccccccc}
          \hline\hline
		&& \multicolumn{4}{c}{$n=200$} &       & \multicolumn{4}{c}{$n=1000$}       \\ \cline{2-6} \cline{8-11}
		& & MBIAS & AB    & STD   & CV95  &       & MBIAS & AB    & STD   & CV95 \\		\hline
		Model 1  &$\gamma$        & 0.0346  & 0.0893 & 0.1357 & 0.9580 &       & 0.0220 & 0.0443 &0.0559   &  0.9370 \\
		&$\rho$                   & -0.0231 & 0.0535 & 0.0725 &  0.9440 &       & 0.0062 & 0.0311 &  0.0430 &   0.9340 \\
		& &       &       &       &       &       &       &       &       &                \\
		Model 2& $\gamma$&   0.0459 & 0.0703& 0.0826& 0.9180 &       & 0.0225  &  0.0347  &  0.0384  &  0.9080 \\
		&$\alpha$ & 0.0041 & 0.0275 & 0.0509& 0.9280 &       & 0.0063   & 0.0248  &  0.0352  &  0.9190 \\
		&$\rho$  &-0.0371  &  0.0617&  0.0746 &  0.9300 &       & -0.0007 &   0.0373&    0.0509  &  0.9310 \\
		&&       &       &       &       &       &       &       &       &               \\
		Model 3& $\gamma$&   0.0198  &  0.0578 &   0.0859 &   0.9370 &       &  -0.0035  &  0.0498  &  0.0638 &   0.9440 \\
		&$\alpha$ & 0.0016  &  0.0048  &  0.0086 &   0.9550 &       & 0.0018&    0.0137 &   0.0274 &   0.9540 \\
		&$\rho$  &-0.0507  &  0.0508  &  0.0182  &  0.1550 &       & 0.0096  &  0.0271  &  0.0417  &  0.9390 \\
		\hline\hline
	\end{tabular}%
	\label{tab4}%

        \begin{minipage}{0.9\linewidth}
          \small \textit{Notes.}
           See  Table \ref{tab1}. {$R=10$ simulated samples across every Monte Carlo design.}
        \end{minipage}          
\end{table}%

\begin{table}[ht]
	\centering
	\caption{Simulation Results for Evolutionary Algorithm (Patternsearch)}
	\begin{tabular}{rcccccccccc}
		\hline\hline
		&& \multicolumn{4}{c}{$n=200$} &       & \multicolumn{4}{c}{$n=1000$}       \\ \cline{2-6} \cline{8-11}
		& & MBIAS & AB    & STD   & CV95  &       & MBIAS & AB    & STD   & CV95 \\		\hline
		Model 1  &$\gamma$        &0.2101 & 0.2335 & 0.2687 & 0.8680 &       & 0.0376 & 0.0578 & 0.0765  & 0.9360 \\
		&$\rho$  & 0.0826 & 0.1374 & 0.1581 &  0.8990 &       & 0.0154 & 0.0465 & 0.0591 &  0.9320 \\
		& &       &       &       &       &       &       &       &       &                \\
		Model 2&$\gamma$  & 0.0761  &  0.0951   & 0.1264  &  0.8990 &       & 0.0234  &  0.0339  &  0.0407  &  0.9140 \\
		&$\alpha$         & 0.0279   & 0.0854   & 0.1095  &  0.9230 &       & 0.0089  &  0.0337  &  0.0434  &  0.9380 \\
		&$\rho$           & -0.0006   & 0.1161 &   0.1494  &  0.9480 &     & -0.0006  &  0.0454 &   0.0576  &  0.9510 \\
		&&       &       &       &       &       &       &       &       &               \\
		Model 3& $\gamma$& 0.1110  &  0.1622  &  0.3047 &   0.9680 &       & 0.0076 &   0.0517  &  0.0722  &  0.9520 \\
		&$\alpha$ & 0.0532  &  0.1914  &  0.2945  &  0.9730 &       & 0.0021 &   0.0765  &  0.0996  &  0.9520 \\
		&$\rho$  &-0.0106  &  0.3228  &  0.4076  &  0.9500 &       & 0.0585 &   0.1376 &   0.1652  &  0.9290 \\
		\hline\hline
	\end{tabular}%
	\label{tab5}%

        \begin{minipage}{0.9\linewidth}
          \small \textit{Notes.}
           See  Table \ref{tab1}. {$R=10$ simulated samples across every Monte Carlo design.}
        \end{minipage}          
\end{table}%

\vspace{0.5cm}

\begin{table}
	\centering
	\caption{Comparison of Bias and Standard Deviation between GII-COV and GII-1. 
        }
        (MBIAS and STD of GII-1 relative to these of GII-COV) \\
	\begin{tabular}{rcccccc}
		\hline\hline
          &     & \multicolumn{2}{c}{$n=200$}     &   & \multicolumn{2}{c}{$n=1000$}   \\  
          \cline{3-4}
          \cline{6-7}
		& & MBIAS &STD&  &MBIAS&STD\\\hline 
		Model 1 &  $\gamma$     & 10.40 & 3.27   &   & -61.00 &6.88\\
		&  $\rho$     & 0.71 &  2.84  &   & 25.33 &6.65\\
		 &      &   &  & & & \\
		Model 2 & $\delta$      & 14.00    &   4.02   & & 27.5 &8.22\\
		&   $\alpha$    & 1.87 &  2.27  &   & 0.75 &3.97\\
		&  $\rho$     & 8.62 &  3.84  &   & 24.78 &5.78\\
		 &      &   &  & & & \\
		Model 3 &   $\gamma$    & 4.00     &  10.98   &  & 47.00& 11.04\\
		&   $\alpha$    & -0.05 &  4.66 &    & -16.64& 7.20\\
		&  $\rho$     & 1.12 &   10.29  &  & -14.91 &13.84\\
		\hline\hline
	\end{tabular}%
	\label{tab7}%

        \begin{minipage}{0.55\linewidth}
          \small \textit{Notes.}
          We report the mean bias (MBIAS) and standard deviation (STD)
          of GII-1 relative to those of GII-COV.
          The number of replications for the Monte Carlo simulation
          is 1,000.
          The cross-sectional sample size, $n$, is
          200 or 1,000.
        \end{minipage}          

\end{table}%

\clearpage

\begin{table}
	\centering
	\caption{Comparison of Bias and Standard Deviation between GII-COV and GII-2. 
        }
        (MBIAS and STD of GII-2 relative to these of GII-COV) \\
	\begin{tabular}{lcccccc}
		\hline\hline
          &     & \multicolumn{2}{c}{$n=200$}     &   & \multicolumn{2}{c}{$n=1000$}   \\  
          \cline{3-4}
          \cline{6-7}
		& & MBIAS &STD&  &MBIAS&STD\\\hline 
		Model 1 &  $\gamma$     & 10.37 & 3.34   &   & -91.00 &6.80\\
		&  $\rho$               & 0.24 &  2.82  &   & 1.22 &6.47\\
		 &      &   &  & & & \\
		Model 2 & $\delta$      & 13.03    &   3.77   & & 38.50 &7.78\\
		&   $\alpha$            & 2.84 &  2.22  &   & 1.00 &3.78\\
		&  $\rho$               & 1.88 &  3.79  &   & 2.56 &5.64\\
		 &      &   &  & & & \\
		Model 3 &   $\gamma$    & 5.45     & 4.86    &  & 28.50& 7.73\\
		&   $\alpha$            & -0.67 &  3.47 &       & -6.45& 4.96\\
		&  $\rho$               & 3.63 &   6.71  &      & -3.82 &10.47\\
		\hline\hline
	\end{tabular}%
	\label{tab8}%

        \begin{minipage}{0.55\linewidth}
          \small \textit{Notes.}
          We report the mean bias (MBIAS) and standard deviation (STD)
          of GII-2 relative to those of GII-COV.
          The number of replications for the Monte Carlo simulation
          is 1,000.
          The cross-sectional sample size, $n$, is
          200 or 1,000.
        \end{minipage}          
\end{table}%

\begin{table}
	\centering
	\caption{Comparison of Raw Computing Time (in seconds)}
	\begin{tabular}{cccc}
		\hline\hline
		\hspace{1.8cm} &  \hspace{1.5cm}    &\ \ $n=200$ \ \ &\ \  $n=1000$ \ \ \\\hline
		Model 1 	&  GII-COV    & 0.3353&	0.5734   \\
		&  GII-1    &0.1134&	0.2738\\
		&  GII-2    & 6.5922&	10.6386
		\\
		&  NM    &0.3148&	0.7579
		\\
		&  PS    & 0.3873&	0.8722
		\\\hline
		Model 2 	&  GII-COV    & 0.2283&	0.7783
		\\
		&  GII-1     & 0.0915&	0.1871
		\\&  GII-2     & 7.8088&	11.2114
		\\
		&  NM    & 0.3808&	0.9430
		\\
		&  PS    & 0.5561&	1.4457
		\\\hline
		Model 3 	&  GII-COV    & 0.0981&	0.4053
		\\
		&  GII-1    & 0.0473&	0.0635
		\\&  GII-2    & 2.1736&	3.9593
		\\
		&  NM    & 0.1456&	0.2977
		\\
		&  PS    & 0.2125&	0.5160
		\\\hline
		\hline
	\end{tabular}%
	\label{tab6}%

        \begin{minipage}{0.55\linewidth}
          \small \textit{Notes.}
          The entries represent the average
          execution time (in seconds)
          across one-hundred
          Monte Carlo replications.
          We report the results for 
           GII-COV,
          naive implementation of GII-K (GII-1),          
          two-step version of GII-K (GII-2),
          Nelder-Mead simplex algorithm (NM),
          Patternsearch algorithm (PS).          
          The cross-sectional sample size, $n$, is
          200 or 1,000.
        \end{minipage}          
\end{table}%

\appendix

\section*{Supplementary Materials for ``Indirect Inference with a Non-Smooth Criterion Function''}

\section{Implementation Details: Examples} 

In this section, we verify that Assumptions 1-3 in the paper are satisfied
for each of the Examples 2-4. 
In addition, for each example we give the specific change-of-variables (COV) needed to implement generalized indirect inference (GII) with COV (GII-COV).

\subsection*{Example 2 (Ordered Probit Model with Individual Effects).}

Let
$x_{i}= (x_{i1},...,x_{iT})\in \R^{d_x}\times \R^T$,
$y_{i}=(y_{i1},...,y_{iT})'\in \{0,1,2,...,J\}^T$, a  $d_z\times T$
matrix $z_{i}=(z_{i1},...,z_{iT})' $,
and
$w_{i}=(w_{i1},...,w_{iT})' \in \R^{T}$
for $i=1, \dots, n$.
Also, we set the parameter
$\theta=(\delta_1,\delta_2,..,\delta_J, \sigma, \gamma')'$.
\vspace{0.2cm}

\noindent\textbf{Assumption 1}
\begin{itemize}
	\item[(a)]
	The observed variables $x_{i}$ and $y_{i}$ are assumed to be iid and $x_i$ and $w_{i}$ are assumed to be independent.
	\item[(b)]
	The innovations $\{w_{i}\}_{i=1}^{n}$ are iid, and follow the standard normal distribution, which has continuously differentiable probability density function.
	\item[(c)]
	The (utility) function
	$$h(x_{i}, w_{i}; \theta)=x_{i}'\gamma+\sigma v_i+w_i $$
	is twice continuously differentiable in both $w\in \R$ and $(\gamma,\sigma)$ given $x_i$.
	\item[(d)]We assume that the parameter space $\Theta$ of $\theta$ is compact, with the added restriction that $\delta_1< \dots <\delta_J$. 
\end{itemize}

\noindent\textbf{Assumption 2}
\begin{itemize}
	\item[(a)] This assumption holds by construction and the compactness of $\Theta$.
	\item[(b)] Typically, one uses a seemingly unrelated regression (SUR) model as the auxiliary model and the moment function is then of the form given in Section 5.1: $$m(y_{i}, z_{i},  \beta)=\left[ \begin{array}{c}{z_{i 1}
		(y_{i 1}-z_{i 1}^{\prime} \beta_{1})} \\ {\vdots} \\ {z_{i T}(y_{i T}-z_{i T}^{\prime} \beta_{T})}\end{array}\right],$$ for some variables $z_{it}$ that are exogenous at time $t$. The function $\beta \mapsto m(y, z, \beta)$ is continuous for any $(y, z)\in\mathcal{Y}\times\mathcal{Z}$.
	\item[(c)] From the definition of the moment function $m(\cdot)$, we have 
	\[
	||z_{it}(y_{it}-z'_{it}\beta_t)||
	\leq||z_{it}y_{it}||+||z_{it}z_{it}'\beta||
	\leq||z_{it}y_{it}||+||z_{it}z_{it}'||\; ||\beta||\leq ||z_{it}y_{it}||+C||z_{it}||,
	\] 
	where the last inequality follows from compactness  of $\mathcal{B}$. Defining $\bar{m}_i=||z_{it}y_{it}||+C||z_{it}||$, we see that $\bar{m}_i$ has finite second moment if $y_{it}$ and $z_{it}$ have finite second moment for all $i,t$. Since $y_{it}$ is a step function, this is satisfied. And since $z_{it}$ is comprised of $x_{it}$ and lags of $x_{it},y_{it}$ the assumption follows so long as $x_{it}$ has a finite second moment for all $(i,t)$. 	
\end{itemize}

\noindent\textbf{Assumption 3}
\begin{itemize}
	\item[(a)] Rewrite the data generating process as 
	\[
	y_{it}=\sum_{j=0}^J j \1\left[{ c^j_{it}(\theta))<u_{it}\leq  c^{j+1}_{it}(\theta)}\right],
	\]
	where 
	\[
	c_{it}^0(\theta)=0, c_{it}^{J+1}(\theta)=1 \mbox{ and }  c^j_{it}(\theta)=\Phi(\delta_j-x_{it}'\gamma-\sigma \nu_i),
	\]	which has the explicit form as in the assumption.  Given that $\delta_j<\delta_{j+1}$ and the normal CDF is strictly monotonic and twice-continuously differentiable,  the random functions are with $ c^j(\theta)<c^{j+1}(\theta)$ and twice-continuously differentiability.
	\item[(b)]The derivatives of $ c_{it}^j(\theta)$ wrt $\theta$ satisfy
	\[
	\|\partial_\theta c_{it}^j(\theta)\|^2=\|\phi(\delta_j-x_{it}'\gamma-\sigma \nu_i)[-x_{it}', -v_i]\|^2\leq (||x_{i,t}||^2+v_i^2),
	\]
	where $\phi$ is the standard normal PDF,  and
	\[
	\begin{split}
	\|\partial_\theta^2 c_{it}^j(\theta)\|^2&=\|(-\delta_j+x_{it}'\gamma+\sigma v_i )\phi(\delta_j-x_{it}'\gamma-\sigma \nu_i)[-x_{it}', -v_i]'[-x_{it}', -v_i]\|^2\\
	&\leq 	(-\delta_j+x_{it}'\gamma+\sigma v_i )^2 \|[x_{it}', v_i]'[x_{it}', v_i]\|^2\\
	&\leq 	\left(\delta_j^2+(x_{it}'\gamma)^2+(\sigma v_i )^2\right) \|[x_{it}', v_i]'[x_{it}', v_i]\|^2\\
	&\leq \left(C_\delta+C_\gamma ||x_{it}||^2+C_\sigma v_i^2\right)\|[x_{it}', v_i]'[x_{it}', v_i]\|^2
	\end{split}
	\]
	where the last inequality follows from the compactness of $\Theta$. 
	Therefore,  $\mathbb{E} [\|  \nabla^{\ell} \bar{c}_{it}\|^2 ]<\infty$ so long as the sixth moment of $x_{it}$ is finite ($v_i$ is Gaussian so all its moments exist).
\end{itemize}

\subsection*{Critical Point Functions and COV}
The COV function in this example is then given by 
\begin{eqnarray*}
	u^r_{it}(\theta,\theta^*)=
	\left \{
	\begin{array}{ll}
		\frac{c^1_{it}(\theta)}{c^1_{it}(\theta^*)}u^r_{it}, & \mathrm{if} \ \ 0 \le u^r_{it} \le c^1_{it}(\theta^*)\\
		c^1_{it}(\theta) +    \frac{c^2_{it}(\theta)-c^1_{it}(\theta)}{c^2_{it}(\theta^*)-c^1_{it}(\theta^*)}\left(u^r_{it}-c^1_{it}(\theta^*)\right), & \mathrm{if} \ \ c^1_{it}(\theta^*) < u_{it}^{r} \le c^2_{it}(\theta^*) \\
		&  \vdots\\
		
		c^J_{it}(\theta) +     \frac{1-c^J_{it}(\theta)}{1-c^J_{it}(\theta^*)}\left(u^r_{it}-c^J_{it}(\theta^*)\right), & \mathrm{if} \ \ c^J_{it}(\theta^*) < u_{it}^{J} \le 1.
	\end{array}
	\right.
\end{eqnarray*}

Although in this cases, there are multiple points of discontinuity, per time step, there is only one COV needed to replace $u_{it}^r$. Similar to Example 1, since the past discontinuity does not flow on to the simulation algorithm of time $t$, the Jacobian term will not accumulate over time. Therefore, the moment function used in GII-COV is given by becomes
\[
M_{n}(\theta, \theta^{\ast}, \beta)
=
\frac{1}{nR}
\sum_{i=1}^{n}
\sum_{t=1}^{T}
\sum_{r=1}^{R}
m \big (y^r_{it}(\theta^*),z_{it},\beta \big)
w^r_{it}(\theta,\theta^* ),
\]
where $w^r_{it}(\theta,\theta^* ) = \partial u_{it}^{r}(\theta, \theta^{\ast})/ \partial u_{it}^{r}$.

\subsection*{Example 3 (Switching-type Models).}
Let $y_t\in \R_+$,    $z_{t}=(1,y_{t-1})'$, recall $\epsilon_t=(v_t,u_t)'$,  where $u_t\sim \text{U}[0,1]$,  $v_t\sim \text{Exp}(1)$,
for $t=1,2,...,T$. 
Also, we set the parameter
$\theta=(\phi, \mu)'$.\\
\vspace{0.2cm}
\noindent\textbf{Assumption 1}
\begin{itemize}
	\item[(a)] There are no exogenous variables.
	\item[(b)] 
	{ The random innovation in this case has two component, $v_t \sim \text{Exp}(1)$ and $u_t\sim \text{U}(0,1)$. Both components are iid and independent of each other, hence the assumption holds by construction.} 
	
	\item[(c)] 
	{{The function $h(\epsilon_t;\theta)=(\mu v_t, u_t)'$ is twice differentiable in $\mu$. }} 
	\item[(d)] The parameter space is compact by assumption. 	
\end{itemize}
\noindent\textbf{Assumption 2} 
\begin{itemize}
	\item[(a)] This assumption holds by compactness of the parameter space. 
	\item[(b)] For this example, we take the moment function to be the first-order conditions from a least squares regression: for $z_{t}=(1,y_{t-1})'$,$$ m(y_{t},z_{t},\beta)=z_{t}(y_{t}-z_{t}'\beta).  $$The assumption then follows.
	\item[(c)] Repeating the same arguments used to verify Assumption 2(c) in Example 2, it is simple to show that the assumption is satisfied. 
\end{itemize}
\noindent\textbf{Assumption 3} 
\begin{itemize}
	\item[(a)]We must slightly alter the original assumption in the main text to fit this more general structure. Rewrite the model in the form
	\[
	y_t=(\phi y_{t-1}+\mu v_t)\1[u_t\leq \phi]+\phi y_{t-1} \1[u_t>\phi].
	\] The above is a generalized version of Assumption 3(a) where, with reference to Assumption 3(a) in the main text, 
	\begin{flalign*}
	g(s_t;\theta)&=\alpha_1(\theta)\1[0\leq u_t<\phi]+\alpha_2(\theta)\1[u_t>\phi],\\\alpha_1(\theta)&=\phi y_{t-1}+\mu v_t,\\\alpha_2(\theta)&=\phi y_{t-1}.
	\end{flalign*}This is now in the form of Assumption 3(a) expect that the functions $\alpha_j(\theta)$ depend on $\theta$. However, these functions are differentiable in $\theta$ and thus do not create any further discontinuities. We then have 
	\[
	c_t^0(\phi)=0, c_t^1(\phi)=\phi \mbox{ and } c_t^2(\phi)=1,
	\] 
	are twice-continuously differentiable, and $ c_t^0<c_t^1<c_t^2$ as long as $0<\phi<1$. 
	\item[(b)] Assumption 3(b) is satisfied as long as $0<\phi<1$.
\end{itemize}

\subsection*{Critical Point Functions and COV}
Let $\theta^*$ be a value at which we wish to evaluate the simulated outcomes. Then, the COV we use is on $u^r_t$, and we replace $u^r_t$ with
\[
u^r_t(\phi,\phi^*)=
\left \{
\begin{array}{lc}
\frac{\phi}{\phi^*}u^r_t, & \mbox{ if } u^r_t\leq \phi^*\\
\frac{1-\phi}{1-\phi^*}(u^r_t-\phi^*)+\phi, & \mbox{ if } u^r_t> \phi^*.\\
\end{array}
\right .
\]
Unlike the previous example, or the example treated in the main paper, the discontinuity $y_{t-k}$, $k=1,\dots,t-1$, has a flow-on-effect on future values of $y_t$. Nonetheless, the COV leads to simulate values of $y_t^r(\theta,\theta^*)$ recursively according to
\[
y^r_{t}(\theta,\theta^*)=\phi y^r_{t-1}(\theta, \theta^*)-\frac{\log(w^r_t)}{\mu} \1[u^r_t\leq\phi^*], 
\]for $y_0(\theta,\theta^*)=0$. That is, although $\phi$ still shows up in $y^r_{t}(\phi,\phi^*)$, the COV has pushed it out of the indicator and now $\phi$ only shows up in a differentiable fashion. In this case, the auxiliary moment function becomes
\[
m_{t}(y^r_{t}(\phi^*,\phi),z_{t},\beta)\prod_{s=1}^tw^r_{s}(\phi,\phi^*).
\]

\subsection*{Example 4 (G/G/1 Queue).} 
\noindent\textbf{Assumption 1} 
\begin{itemize}
	\item[(a)] There are no exogenous variables.
	\item[(b)] The random innovation here can be taken as the joint vector of inter-arrival times of customers and the service times, $w_i$ and $v_i$. These terms are mutually independent and identically distributed by assumption. \footnote{In the original assumption, we set the random innovation to be a sequence of iid random variables such that its distribution does not depends on the parameter of interest. However, we are interested in their distributional parameters $\theta_w$ and $\theta_v$, which is a slight modification of the assumption. This modification is immaterial since, for all common distributions used for $w_i$ and $v_i$, we can always generate these variables using an appropriate transformation of an iid variable that does not depend on $\theta$. }
	
	\item[(c)] The state variable is $s_i=(v_i, w_i)',$ which implies that the function $h(\cdot)$ is the identity map, and satisfies the assumption by construction. 
	\item[(d)] The parameter spaces are compact by assumption. 	
\end{itemize}
\noindent\textbf{Assumption 2} 
\begin{itemize}
	\item[(a)] This assumption holds by compactness of $\theta$. 
	\item[(b)] Following \cite{heggland2004estimating}, we consider as our auxiliary moments those from a Gaussian mixture model with two components: let $\beta=(\mu_1,\sigma_1^2,\mu_2,\sigma^2_2,\pi)$, where $\pi$ denotes the mixing proportion, $(\mu_l,\sigma^2_l)'$, $l=1,2$, denotes the mean and variance of the normal model, and consider the auxiliary moments
	\begin{flalign*}
	m(y_i,\beta):=\begin{bmatrix}(1-\gamma_i(\beta))(y_i-\mu_1)\\\gamma_i(\beta)(y_i-\mu_1)\\(1-\gamma_i(\beta))\left\{(y_i-\mu_1)^2-\sigma_1^2\right\}\\\gamma_i(\beta)\left\{(y_i-\mu_2)^2-\sigma_2^2\right\}\\\gamma_i(\beta)-\pi
	\end{bmatrix},
	\end{flalign*}where $$\gamma_i(\beta):=\frac{{\pi} \phi_{\mu_2,\sigma_2}\left(y_{i}\right)}{(1-{\pi}) \phi_{\mu_1,\sigma_1}\left(y_{i}\right)+{\pi} \phi_{\mu_2,\sigma_2}\left(y_{i}\right)},$$ and where $\phi_{\mu,\sigma}$ denotes the normal pdf with mean $\mu$ and variance $\sigma^2$. The auxiliary moment function $m(y_i,\beta)$ is continuous in $\beta$ for all $y$.
	
	\item[(c)] By construction, for $y_i$ with bounded support, $\sup_{\beta\in\mathcal{B}}\|\gamma_i(\beta)\|\leq M<\infty$ for all $i$. Therefore, $$\|m(y_i,\beta)\|\leq M\left(\|y_i-\mu_1\|+\|y_i-\mu_2\|+\|(y_i-\mu_1)^2-\sigma_1^2\|,\|(y_i-\mu_2)^2-\sigma_2^2\|\right). $$From compactness of $\mathcal{B}$, we then have that, for some non-random constant $C$, $\|m(y_i,\beta)\|\leq  C\|y_i^2\|$. For all widely used DGP for $y_i$, $\mathbb{E}[y_i^2]<\infty$, and the assumption follows. 
	
\end{itemize}
\noindent\textbf{Assumption 3} 
\begin{itemize}
	\item[(a)]Define $$e_i=\sum_{j=1}^{i-1}y_j-\sum_{j=1}^{i-1}w_j.$$
	
	Rewrite the model in the form
	\[
	y_i=v_i\1[u^w_i\leq F_w(e_i; \theta_w)]+(v_i+e_i+w_i)\1[u^w_i>F_w(e_i;\theta_w)]
	\] 
	
	Similar to Example 3, we have $\alpha_0(\theta_v)=v_i$ and $\alpha_1(\theta)=v_i+e_i+w_i$. They are no longer known constants, but they are known functions of $\theta$ given $\{u^v_1, u^w_1, ..., u^v_{i-1}, u^w_{i-1}\}$, the sequence of standard uniforms used to generate the $v_j$ and $w_j$ terms. In the original algorithm, the states $s_i$ are not differentiable in $\theta$, we shall explain how COV are used to construct a differentiable algorithm for generating $s_i$ sequentially.  
	
	The critical value functions are
	\[
	c_i^0(\theta)=0, c_i^1(\theta)= F_u(e_i; \theta_u), \mbox{ and } c_i^2(\theta)=1;
	\]	
	we have $c_i^0<c_i^1<c_i^2$ as long as $w_i$ is an absolutely continuous random variable.   

	\item[(b)]  This follows from (a).	
	The twice differentiability is guarantee by the sequence of COV up to $i-1$, and twice continuously differentiability of $F_w$ and $F_v$ in $\theta$ and continuous ${\partial_{x}\partial_\theta} f_x(x,\theta)$ for $x$ being $v$ and $ w$ respectively. We shall see this more in the COV section below. Continuous derivatives in a compact set is sufficient for the assumption. 
	
\end{itemize}

\subsection*{Critical Point Functions and COV}

There is an irregularity for the first customer. For $n=2,....,N$, the COV is performed by replacing the standard uniform random variable for simulating $w_n$ by 
\[
u_n(\theta,\theta^*)=
\begin{cases}
\frac{c_{i}^{1}(\theta)}{c_n(\theta^*)}u_n  \mbox{ if } u_n\leq c_{i}^{1}(\theta^*)\\
c_{i}^{1}(\theta) +
\frac{1- c_{i}^{1}(\theta)}{1-c_n(\theta^*)}\left(u_n-c_{i}^{1}(\theta^*)\right)\mbox{ if } u_n> c_{i}^{1}(\theta^*).\\
\end{cases}
\]
More specifically:
\begin{itemize}
	\item  $y_1=v_1$, COV is not needed. 
	\item For the second customer, the original non-differentiable algorithm is 
	\[
	y_2=v_2\1[u^w_2>F_w(v_1;\theta_w)]+(v_2+w_2-y_1)\1[u^w_2\leq F_w(v_1;\theta_w)].
	\]
	For the COV, we replace the random uniform, $u^w_2$, for simulating $w_2$ with $ u_2(\theta,\theta^*)$,  then have
	\[
	y_2(\theta,\theta^*)=v_2(\theta_v)\1[u^w_2>c_2(\theta^*)]+(v_2(\theta_v)+w_2(u_2(\theta,\theta^*))-y_1(\theta_v))\1[u^w_2\leq c_2(\theta^*)]
	\]
	and the COV has pushed the parameters of interest out of the indicator function.  
	\item Assume that COV is performed up to the $i-1$th customer and ensured twice-differentiability of $y_j(\theta,\theta^*)$ for $j=1,2,...,i-1$, hence 
	\[
	e_i(\theta)=\sum_{j=1}^{n-1}y_j(\theta,\theta^*)-\sum_{j=1}^{n-1}w_j\left(u_j(\theta,\theta^*)\right)
	\]
	is now twice-continuously differentiable. For the $i$th draw,   
	\[
	y_i=v_i\1[u^w_i>c_i^1(\theta)]+(v_i+w_i-e_i(\theta))\1[u^w_i\leq c_i^1(\theta)]
	\]
	we perform similar COV, and obtain
	\[
	y_i(\theta,\theta^*)=v_i(\theta) \1[u^w_i>c_i(\theta^*)]+(v_i(\theta)+w_i(\theta)-e_i^1(\theta))\1[u^w_i\leq c_i(\theta^*)]. 
	\]
\end{itemize}
In the auxiliary moment function, like the Example 3 , we have the product of the Jacobian terms through time due to the sequence of COV.

\section{HOPP versus GII-COV}

In this section, we briefly compare and contrast the Hessian optimal partial proxy method (hereafter, HOPP) in \cite{joshi2016optimal} and the GII-COV approach in this paper. 

\subsection{HOPP}
Following \cite{joshi2016optimal}, assume that our goal is to calculate the derivative of an option price at the specific point $\theta^*\in \mbox{int}(\Theta)\subset \mathbb{R}_{}$. For $g:\mathcal{S}\times\Theta\rightarrow\mathbb{R}_{+}$ denoting the discounted payoff function of the option, which depends on the parameter $\theta$ and the random variable $S$, where $S$ has support $\mathcal{S}$ and density function $f(\cdot)$, the option price can be expressed as a function of $\theta$:
\[
G(\theta)=\mathbb{E}[g(S,\theta)]=\int_{\mathcal{S}}g(s,\theta)f(s)ds.
\]
Also, consider that $G(\theta)$ does not have an analytic form and that $g(\cdot,\theta)$ is non-differentiable in $\theta$, which is the case for many common option prices. In this case, the derivative can not be passed through the integral, which means that the standard approach of differentiating $g(\cdot,\theta)$, with respect to $\theta$, and then taking the expectation of the resulting quantity will be biased for $\partial_\theta G(\theta^*)$, the derivative of interest.

An alternative approach to calculate $\partial_\theta G(\theta^*)$ is to approximate this derivative using finite-differencing: $$\{G(\theta^*+h)-G(\theta^*)\}/h$$ for some small $h$. However, such a finite-differencing approach will be time consuming if $G(\theta)$ is difficult to calculate, and, in addition, for any non-zero $h$, the resulting approximation is still a biased estimator of $\partial_\theta G(\theta^*)$, where the order of the bias is $O(h)$. 

The goal of the original Hessian optimal partial proxy method (HOPP), \cite{joshi2016optimal}, is to construct an approximation to the derivative of the option price that is faster to calculate and which has smaller bias than finite-differencing approaches. Even though $g(\cdot,\theta)$ is not differentiable in $\theta$, under very weak conditions on $g(\cdot,\theta)$, the derivative $\partial_\theta G(\theta)$ will still exist. Using this fact, \cite{joshi2016optimal} propose to use a change-of-variables (COV) approach, in conjunction with a second-order Taylor series, to approximate the derivative $\partial_\theta G(\theta)$. 

The first step in the application of HOPP is the realization  that we can represent the expectation in question as an integral over a uniform random variable. In particular, if the function $g(\cdot)$ is of the form $\1[\alpha_0\leq S\leq \alpha_1]l(S)$, where $l(\cdot)$ is twice-differentiable, for $u\in[0,1]$, we can  represent $G(\theta)$ as
$$G(\theta)=\int_{0}^{1} \1\{u\in\left[\alpha_0, \alpha_{1}\right]\}f(\theta, u)l[f(\theta, u)] d u,$$where $f$ denotes the algorithm that takes in $(\theta, u)\in\Theta\times [0,1]$ and produces the random variable $S$. Under weak conditions on $f(\cdot)$, we have $f(\theta, u) \in\left[\alpha_{0}, \alpha_{1}\right]$ if and only if $u \in\left[c^{0}(\theta), c^{1}(\theta)\right]$, which allows us to rewrite the integral as $$G(\theta)=\int_{ c^{0}(\theta)}^{ c^{1}(\theta)} l[f(\theta, u)] d u.$$ The functions $c^j(\theta)$, $j=0,1$, are referred to as critical point functions. 

To alleviate the integrals dependence on $\theta$ in the bounds of integration, HOPP \textit{\textbf{must consider a particular}}  change-of-variables (COV) around the point we wish to calculate the derivatives. Say that we eventually wish to calculate $\partial_{\theta} G(\theta^*)$ and $\partial^2_{\theta} G(\theta^*)$, then HOPP considers the COV 
\begin{flalign*}
u(\theta, \theta^*)&=u+\left(\theta-\theta^*\right) \gamma(u)+\frac{1}{2}\left(\theta-\theta^*\right)^{2} \delta(u),\\
\gamma(u)&=\partial_\theta c^{0}\left(\theta^*\right)+\frac{\partial_\theta c^{1}\left(\theta^*\right)-\partial_\theta c^{0}\left(\theta^*\right)}{ c^{1}\left(\theta^*\right)- c^{0}\left(\theta^*\right)}\left(u- c^{0}\left(\theta^*\right)\right),
\\
\delta(u)&=\partial^2_{\theta} c^{0}\left(\theta^*\right)+\frac{\partial^2_{\theta} c^{1}\left(\theta^*\right)-\partial^2_{\theta} c^{0}\left(\theta^*\right)}{ c^{1}\left(\theta^*\right)- c^{0}\left(\theta^*\right)}\left(u- c^{0}\left(\theta^*\right)\right)\end{flalign*}These specific choices of $\gamma(\cdot)$ and $\delta(\cdot)$ are needed to ensure that: one, the dependence on $\theta$ in the bounds of integration is completely removed, and in a manner that does not affect the value of the integral; two, derivatives calculated using this COV will agree with the exact derivative, at least up to a third-order term. That being said, it is important to realize that in the majority of applications, the above functions are not analytically tractable, due to $c^0(\theta),c^1(\theta)$ being intractable, and numerical methods must be used to approximate the functions $\gamma(\cdot)$ and $\delta(\cdot)$. 

Using this COV and expanding the integral using Taylor's theorem around $(\theta-\theta^*)$, up the third-order, then yields the following form for the integral  
$$G(\theta)=\int_{ c^{0}\left(\theta^*\right)}^{ c^{1}\left(\theta^*\right)} l[f(\theta, u(\theta, \theta^*))] \frac{\partial u(\theta, \theta^*)}{\partial u} d u.$$ \cite{joshi2016optimal} demonstrate that derivatives calculated from the above will coincide with $\partial_\theta G(\theta^*)$ and $\partial^2_{\theta}G(\theta^*)$, up to a third-order term that is $O(\|\theta-\theta^*\|^{3})$.

Stated in words, through a clever COV, the HOPP procedure is able to bring the underlying differentiability of $G(\theta)$ to the fore and produce a derivative approximation that is quicker and more accurate than finite-differencing (\citealp{joshi2016optimal}). However, it is also clear from the above that the HOPP procedure \textbf{\textit{explicitly requires}} that the criterion of interest be differentiable. Whithout this, the Taylor series arguments that underly the HOPP procedure would be invalidated and the resulting theoretical results developed in \citet{joshi2016optimal} would not be valid. That is, HOPP is not applied in situations where the criterion of interest is non-differentiable. Instead, HOPP is applied to settings where the criterion is differentiable, but where we would like to obtain computationally convenient and accurate approximations for the derivatives of this criterion.

\subsection{GII-COV}
The goal of the GII-COV procedure is to obtain consistent and asymptotically normal (CAN) parameter estimates using derivative-based optimization routines, and in the specific context where the simulated endogenous variables are discontinuous in the parameter of interest, denoted by $\theta$. To carry out such a task, GII-COV uses a COV approach to construct approximate derivatives that are smooth for all possible values of the sample size $n$, irrespective of the simulation size, $R$, and uniformly for all $\theta\in\Theta$. That is, in contrast to HOPP, the goal of GII-COV is \textbf{not} to approximate the derivative of an expectation, but to approximate derivatives of sample functions, and to ensure that these approximation are regular enough to allow for CAN parameter estimation. 

Recall that the LM-II approach is based on simulated auxiliary moments $M_{n} : \Theta \times \mathcal{B} \rightarrow \mathbb{R}^{d_{\beta}}$, given by 
$$M_{n}(\theta, \beta) :=\frac{1}{n R} \sum_{i=1}^{n} \sum_{r=1}^{R} m\left(y_{i}^{r}(\theta), z_{i}, \beta\right).$$In this context, the idea behind the GII-COV approach is to replace the non-smooth, in $\theta$, moment function $m\left(y_{i}^{r}(\theta), z_{i}, \beta\right)$ with an approximation so that derivatives of this approximation exist and can be used to estimate the unknown parameter $\theta$. 

To carry out the above task,  GII-COV relies on a COV approach around the point where we wish to calculate the derivative, denoted by $\theta^*$. For $u_{i}^{r}$ denoting the uniforms that are used to simulate the outcomes $y_{i}^{r}(\theta)$, GII-COV uses the following COV:
$$u_{i}^{r}\left(\theta, \theta^{*}\right) :=
c_{i}^{j}(\theta) +
\frac{c_{i}^{j+1}(\theta)-c_{i}^{j}(\theta)}{c_{i}^{j+1}\left(\theta^{*}\right)-c_{i}^{j}\left(\theta^{*}\right)}\left\{u_{i}^{r}-c_{i}^{j}\left(\theta^{*}\right)\right\},$$ where $j=0,\dots,J$ and $J$ denotes the total number of possible discontinuities for $y_{i}^r(\theta)$ (see Section three in the main paper for details). Once the COV has been obtained, GII-COV approximates the discontinuous $y_{i}^{r}(\theta)$ with new simulated outcomes $$y_{i}^{r}\left(\theta, \theta^{*}\right)=\sum_{j=0}^{J} \alpha_{j} \1\left[c_{\mathrm{i}}^{j}(\theta)<u_{i}^{r}\left(\theta, \theta^{*}\right) \leq c_{i}^{j+1}(\theta)\right],$$ and the original moment function is then approximate using 
\begin{flalign*}
m_{i}^{r}\left(\theta, \theta^{*}, \beta\right) &:=m\left(y_{i}^{r}\left(\theta, \theta^{*}\right), z_{i,} \beta\right) \cdot w_{i}^{r}\left(\theta, \theta^{*}\right)\\w_{i}^{r}\left(\theta, \theta^{*}\right)&:=\frac{c_{i}^{j+1}(\theta)-c_{i}^{j}(\theta)}{c_{i}^{j+1}\left(\theta^{*}\right)-c_{i}^{j}\left(\theta^{*}\right)}.
\end{flalign*}

In comparison with the COV required by HOPP, the COV required by GII-COV is simple. This is because, in contrast to HOPP, GII-COV operates on sample functions and does not operate on integrals. It is this fact that necessitates the different COV used in the two procedures: since HOPP is interested in derivatives of integrals (or expectations), it must use a COV that respects the bounds of integration; since GII-COV is interested in derivatives of sample functions, GII-COV can use a much simpler COV. Indeed, it is important to note that the HOPP approach, as described in \cite{joshi2016optimal}, is not meant o approximate non-differentiable sample quantities, such as $\partial_\theta y_{i}^{r}(\theta)$, whereas the GII-COV approach is specifically designed to approximate such quantities, such as $\partial_\theta y_{i}^{r}(\theta)$. 

In conclusion, the GII-COV and HOPP procedures both employ a clever COV and share similar ideas. However, the GII-COV procedure operates on sample functions that are not differentiable, i.e., the individual simulated outcome $y_{i}^{r}(\theta)$, whereas HOPP operates on differentiable functions (expectations or integrals). The nature of the COV used in the two procedures is also completely different, with GII-COV procedure requiring a much simpler COV due to the fact that it operates on sample functions and not integrals. 

Lastly, we note that the  HOPP procedure is a pointwise procedure, and the resulting theoretical validity is only guaranteed in a pointwise sense. In contrast, since the goal of the GII-COV procedure is to obtain consistent and asymptotically normal estimators, the GII-COV procedure has been constructed to ensure that the derivatives used in estimation exist for all sample sizes, $n$, all choices of the simulation size, $R$, and uniformly in the parameter space, for all $\theta\in\Theta$. This is the content of Proposition 1, and Theorem 1 in the main paper.

\section{Derivative Estimation}

In this section, we briefly present, in general terms, the main ideas underlying automatic differentiation. Since automatic differentiation is a complex and active research field, we only provide the intuition behind the method, and refer the interested reader to more specialized texts for full details.

\subsection{Pathwise Derivatives via Automatic Differentiation }
The problem of sensitivity computation, i.e., derivative computation, for criterion functions constructed via simulations has received much attention both in finance and engineering disciplines; we refer the reader to \cite{fu2006gradient}  for an overview and discussion of various methodologies. In the econometrics community, the most prominent method for derivative computation is numerical finite-differencing. 

For a criterion $Q_n(\theta)$, the simplest finite-differencing approach constructs an estimate of the Jacobian at the point $\theta^*$, $\partial_\theta Q_n(\theta^*)$, using $$\{Q_n(\theta^*+h)-Q_n(\theta^*)\}/h,$$ 
where $h$ is a differencing parameter. While such an estimator is intuitive and often simple to construct, there are a few well-known trade-offs when using this method. Most important is the fact that the choice of $h$ leads to a bias-variance trade-off in terms of the accuracy with which the finite-differencing estimator approximates $\partial_\theta Q_n(\theta^*)$.


An alternative to numerical finite-differencing methods is to instead compute derivatives using the so-called ``pathwise'' approach. The pathwise approach does not construct the derivative indirectly, as with the finite-differencing, but instead directly differentiates the algorithm that is used to construct the simulated criterion. For instance, in our case, this requires differentiating, say, at the point $\theta^*$, the sequence of steps that are needed to construct the simulated criterion, including the mechanism that is used to simulate the endogenous variables, and with each intermediate derivative produced via the chain rule. Such a procedure implicitly requires that all of the required intermediate derivatives are known, or can be exactly calculated. Once each intermediate derivative has been calculated, the derivative $\partial_\theta Q_n(\theta^*)$ is simply calculated via the chain rule. If each intermediate derivative can be exactly calculated, the resulting pathwise derivative is equal to $\partial_\theta Q_n(\theta^*)$. That is, in contrast to finite-differencing, the pathwise approach yields \textbf{\textbf{the derivative in question}} and not an estimate thereof. The superiority of this method over finite-differencing approaches is discussed in detail in \cite{glasserman2003monte}. 


While many simulated criterion functions admit application of this pathwise method, the issue is
how to carry out this differentiation efficiently, since, if done naively, it can be very slow when there are a large number of steps needed to form the criterion. In the paper, we recommend the use of the automatic differentiation (AD) techniques to calculate these derivatives. AD refers to a suite of computationally efficient tools that bridge the gap between numeric and symbolic differentiation; i.e., AD computes derivatives through accumulation of partial derivative values during code execution to generate exact numerical derivatives, and does so by respecting the mathematical rules of function composition and chain rule differentiation. 

Implementation of AD depends on how the intermediate derivatives are calculated and stored. In general, AD is most commonly implemented using the so-called forward and backward modes of derivative calculation, with the particular application of interest determining which approach is more appropriate. The differences between the two approaches can be represented in terms of how the interleaving derivatives in the chain rule are evaluated and stored. For instance, the forward mode of AD traverses the chain rule for the derivatives of the algorithm from the inner-most step to the outer-most step. That is, we traverse the chain rule in the same manner in which we evaluate the function $Q_n(\theta)$, which is intuitively appealing.

In contrast to the forward mode, the backward mode of AD traverses the chain rule starting with the outer-most step of the algorithm and working inward. This means that the backward mode requires storing the entire string of
computations and so can require greater memory costs than the
the forward mode;  i.e., in the backward mode, each partial derivative that is needed to calculate the overall derivative must be stored, rather than simply evaluated as with the forward mode. In this way, the forward mode of AD can often be computationally more efficient than the backward model. We refer the reader to
Griewank--Walther (2008) for a detailed comparison on these two models of AD computation.

\subsubsection*{AD is not Symbolic Differentiation}

Symbolic differentiation is another tool widely complimented in many computer software that takes in a mathematical expression, $Q_n$ in our case, and returns a symbolic expression of the derivative $\partial_\theta Q_n$. The problem of using symbolic differentiation in many problems of statistical inference is usually associated with the complexity of the original algorithm for computing $Q_n$. In cases of insufficient code, the program is unable to convert $\partial_\theta Q_n$ into one single expression. Even for cases where the expression $\partial_\theta Q_n$ is obtainable, it is often slow to execute as certain sub-expressions appear in various places in the symbolic expression, and hence evaluated several times.

Unlike symbolic differentiation, AD is a mechanism for evaluating derivative values $\partial_\theta Q_n(\theta)|_{\theta=\theta^*}$ without directly calculating the expression analytically. Rather, taking a flat view of the algorithm $\theta \rightarrow Q_n$, AD considers a sequence of intermediate value, where each mapping  and its derivatives are relatively simple to evaluate. This means that the program AD uses to compute the gradient has exactly the same structure as the function we are evaluating, which implies that AD will yield manageable execution times. However, similar to symbolic differentiation, AD also returns an exact derivatives of the original algorithm evaluated at $\theta^*$ up to floating point error.

\section{Further Numerical Results: AD versus Numerical Derivatives}
In this section, we compare the impact of using AD versus standard numerical derivatives. To this end, we compare two versions of GII-COV: the first version is the same procedure that was implemented in Section five of the main paper, which uses automatic differentiation techniques to obtain derivatives used within a Newton-Raphson procedure; the second version of GII-COV replaces the derivatives obtained via AD in the Newton-Raphson procedure with those based on central finite-differences using an optimal differencing parameter.

The results for GII-COV with AD are given in Table \ref{tab2_su}, and the result of GII-COV with finite-differencing are given in Table \ref{tab3_su}. Comparing the results of Tables \ref{tab2_su} and \ref{tab3_su}, we see that across all chosen measures of accuracy, GII-COV based on AD is more accurate than GII-COV based on numerical derivatives. This is not surprising and parallels other results based on comparing the accuracy of AD and numerical derivatives more broadly (see \citealp{renaud1997automatic}, \citealp{baydin2018automatic} and \citealp{kucukelbir2017automatic}). 

Given the striking difference between the results in Tables \ref{tab2_su} and \ref{tab3_su}, we speculate that in many econometric models, it may be possible to obtain substantive accuracy gains by simply switching from the more common numerical finite-differencing derivatives estimators to derivatives calculated using AD techniques.

\begin{table}[h]
	\centering
	\caption{Simulation Results for GII-COV procedure using numerical derivatives.}
	\begin{tabular}{lcccccccccc}
		\hline\hline
		&& \multicolumn{4}{c}{$n=200$} &       & \multicolumn{4}{c}{$n=1000$}       \\ \cline{2-6} \cline{8-11}
		& & MBIAS & AB    & STD   & CV95  &       & MBIAS & AB    & STD   & CV95 \\
		\hline
		Model 1&$\gamma$ & 0.0121 & 0.0185 & 0.0261 &  0.9220 &       & 0.0015 & 0.0015 & 0.0090 & 0.9660 \\
		&$\rho$   & -0.0007 & 0.0287 & 0.0528 &  0.8960 &       & 0.0043 & 0.0047 & 0.0240 &  0.9610 \\ 
		&      &       &       &       &       &       &       &       &       &             \\
		Model 2 &$\gamma$& 0.0584 & 0.1035 & 0.0940  & 0.9800 &       & 0.0058 & 0.0244 & 0.0271 & 1.000 \\
		&$\alpha$ & 0.0059 & 0.0266 & 0.0463  & 0.9460 &       & 0.0002 & 0.0094 & 0.0146 & 0.9280 \\
		&$\rho$& -0.0289 & 0.0416 & 0.0491  & 0.9040 &       & -0.0081 & 0.0141 & 0.0189 & 0.9260 \\
		&&       &       &       &       &              &       &       &       &       \\
		Model 3&$\gamma$ & 0.0076 & 0.0787 & 0.3567 &  0.9960 &       & 0.0001 & 0.0042 & 0.0068 & 0.9500 \\
		&$\alpha$  & 0.0096 & 0.0187 & 0.0607 &  0.9840 &       & 0.0011 & 0.0053 & 0.0141 &  0.9460 \\
		&$\rho$ & -0.0093 & 0.0494 & 0.4016  & 0.9940 &       & 0.0011 & 0.0017 & 0.0030& 0.9600 \\
		\hline\hline
	\end{tabular} 
	\begin{minipage}{0.9\linewidth}
		\textit{Notes.}
		The number of replications for the Monte Carlo simulation
		is 1,000.
		The cross-sectional sample size $n$ is
		200 or 1,000.
		We report the mean bias (MBIAS),
		mean absolute bias (AB), standard deviation (STD)
		and the Monte Carlo coverage of a 95\% Wald-confidence interval
		(CV95).
	\end{minipage}          
	\label{tab2_su}%
\end{table}%

\begin{table}[h]
	\centering
	\caption{Simulation Results for GII-COV procedure using automatic differentiation.}
	\begin{tabular}{lcccccccccc}
		\hline\hline
		&& \multicolumn{4}{c}{$n=200$} &       & \multicolumn{4}{c}{$n=1000$}       \\ \cline{2-6} \cline{8-11}
		& & MBIAS & AB    & STD   & CV95  &       & MBIAS & AB    & STD   & CV95 \\
		\hline
		Model 1&$\gamma$ & 0.0052 & 0.0110 & 0.0281 &  0.9570 &       & -0.0001 & 0.0015 & 0.0049 & 0.9590 \\
		&$\rho$   & -0.0043 & 0.0183 & 0.0419 &  0.9430 &       & -0.0009 & 0.0028 & 0.0075 &  0.9440 \\ 
		&      &       &       &       &       &       &       &       &       &             \\
		Model 2 &$\gamma$& 0.0038 & 0.0105 & 0.0246  & 0.9440 &       & 0.0002 & 0.0016 & 0.0045 & 0.9460 \\
		&$\alpha$ & 0.0039 & 0.0230 & 0.0463  & 0.9370 &       & 0.0008 & 0.0041 & 0.0107 & 0.9510 \\
		&$\rho$& -0.0034 & 0.0174 & 0.0341  & 0.9410 &       & -0.0009 & 0.0032 & 0.0097 & 0.9600 \\
		&&       &       &       &       &              &       &       &       &       \\
		Model 3&$\gamma$ & 0.0033 & 0.0115 & 0.0286 &  0.9500 &       & 0.0002 & 0.0021 & 0.0071 & 0.9610 \\
		&$\alpha$  & 0.0057 & 0.0241 & 0.0519 &  0.9530 &       & 0.0011 & 0.0053 & 0.0141 &  0.9460 \\
		&$\rho$ & -0.0060 & 0.0193 & 0.0443  & 0.9510 &       & -0.0011 & 0.0037 & 0.0110& 0.9600 \\
		\hline\hline
	\end{tabular} 
	\label{tab3_su}%
\end{table}%

\end{document}